\documentclass[11pt]{article}
\usepackage[utf8]{inputenc}
\usepackage[margin=1in]{geometry}
\usepackage{amsmath}
\usepackage{amsfonts}
\usepackage{hyperref}
\usepackage[english]{babel}
\usepackage{amssymb}
\usepackage{bbm}
\usepackage{bm}
\usepackage{mathrsfs}
\usepackage{graphicx}
\usepackage{subcaption}
\usepackage{listings}
\usepackage{cancel}
\usepackage{centernot}
\usepackage{enumerate}
\usepackage{enumitem}
\usepackage{forest}
\usepackage{tikz}
\usepackage{nicefrac}
\usepackage[colorinlistoftodos]{todonotes}
\usepackage[linesnumbered,ruled]{algorithm2e}
\usepackage{dsfont}
\usepackage{amsthm}
\usepackage{thmtools}
\usepackage{thm-restate}

\theoremstyle{plain}
\newtheorem{theorem}{Theorem}[section]
\newtheorem{lemma}[theorem]{Lemma}
\newtheorem{proposition}[theorem]{Proposition}
\newtheorem{corollary}[theorem]{Corollary}

\newtheorem{definition}[theorem]{Definition}

\newtheorem{remark}[theorem]{Remark}
\newtheorem{assumption}[theorem]{Assumption}

\DeclareMathOperator{\Var}{Var}
\DeclareMathOperator{\Cov}{Cov}
\DeclareMathOperator{\poly}{poly}
\DeclareMathOperator{\Exp}{Exp}
\DeclareMathOperator{\Geo}{Geo}

\renewcommand{\P}[1]{\mathds{P}\left[{#1}\right]}
\newcommand{\E}[1]{\mathds{E}\left[{#1}\right]}

\newcommand{\1}[1]{\mathds{1}[{#1}]}

\newcommand{\R}{\mathbb{R}}

\newcommand{\M}{\mathcal{M}}
\newcommand{\W}{\mathcal{W}}
\newcommand{\D}{\mathcal{D}}
\newcommand{\pim}{\pi_{\min}}
\newcommand{\pima}{\pi_{\max}}
\newcommand{\epsilonDotAlpha}{\bm{\epsilon}\cdot\bm{\alpha}}
\newcommand{\deltaDotBeta}{\bm{\delta}\cdot\bm{\beta}}
\newcommand{\rinit}{r^{\mathrm{init}}}
\newcommand{\rgeo}{r^{\mathrm{geo}}}
\newcommand{\rcontd}{r^{\mathrm{contd}}}
\newcommand{\rrest}{r^{\mathrm{rest}}}
\newcommand{\Wlow}{W_{\mathrm{low}}}
\newcommand{\Whigh}{W_{\mathrm{high}}}
\newcommand{\mathcalWlow}{\W_{\mathrm{low}}}
\newcommand{\mathcalWhigh}{\W_{\mathrm{high}}}
\newcommand{\BarRM}{{\overline R}^{M}}
\newcommand{\BarRW}{{\overline R}^{W}}

\makeatletter
\newcommand\Ps@textstyle[2]{\mathbb{P}_{#1}\left[{#2}\right]}
\newcommand\Es@textstyle[2]{\mathbb{E}_{#1}\left[{#2}\right]}
\newcommand\Ps[2]{%
  \mathchoice 
  {\underset{{#1}}{\mathbb{P}}\left[{#2}\right]}
  {\Ps@textstyle{#1}{#2}}
  {\Ps@textstyle{#1}{#2}}
  {\Ps@textstyle{#1}{#2}}
}
\newcommand\Es[2]{%
  \mathchoice 
  {\underset{{#1}}{\mathbb{E}}\left[{#2}\right]}
  {\Es@textstyle{#1}{#2}}{\Es@textstyle{#1}{#2}}{\Es@textstyle{#1}{#2}}
}
\makeatother

\graphicspath{{./Figures/}}

\title{Tiered Random Matching Markets:\\
Rank is Proportional to Popularity}

\author{
Itai Ashlagi \\
iashlagi@stanford.edu
\and 
Mark Braverman\thanks{Research supported in part by the NSF Alan T. Waterman Award, Grant
No. 1933331, a Packard Fellowship in Science and Engineering, and the
Simons Collaboration on Algorithms and Geometry. Any opinions,
findings, and conclusions or recommendations expressed in this
publication are those of the author and do not necessarily reflect the
views of the National Science Foundation.}  \\
mbraverm@gmail.com
\and
Amin Saberi\\
saberi@stanford.edu
\and 
Clayton Thomas \\
claytont@cs.princeton.edu
\and
Geng Zhao\\
gengz@stanford.edu
}

\begin{document}

\maketitle

\begin{abstract}
We study the stable marriage problem in two-sided markets with randomly generated preferences. Agents on each side of the market are divided into a constant number of ``soft'' tiers,  which capture agents' qualities. Specifically, every agent within a tier has the same public score, and agents on each side have  preferences independently  generated  proportionally to the public scores of the other side. 

We compute the expected average rank which agents in each tier have for their partners in the man-optimal stable matching, and prove concentration results for the average rank in asymptotically large markets.
Furthermore, despite having a significant
effect on ranks, public scores do not strongly
influence the probability of an agent matching to a given tier of the other side.
This generalizes the results in~\cite{pittel1989average}, which analyzed markets with uniform preferences.
The results quantitatively demonstrate the effect of competition due to the heterogeneous attractiveness of agents in the market. 
\end{abstract}

\section{Introduction}
\label{sectionIntroduction}

The theory of stable matching, initiated by~\cite{gale1962college}, has led to a deep understanding of two-sided matching markets and inspired successful real-world market designs. Examples of such markets  include marriage markets, online dating, assigning students to schools, labor markets, and college admissions. In a market matching ``men'' to ``women'' (a commonly used analogy), a matching is stable if no man-woman pair prefer each other over their assigned partners.

A fundamental issue is characterizing stable  outcomes of matching markets, i.e. the outcome agents should {\it expect} based on market characteristics.    
Such characterizations are not only useful for describing outcomes but also likely to be fruitful in market designs. Numerous papers so far have studied stable matchings in random markets, in which agents' preferences are generated uniformly at random \cite{pittel1989average, knuth1990stable, ashlagi2017unbalanced, pittel2019likely}. This paper contributes to the literature by expanding these results to a situation where preferences are drawn according to different tiers of ``public scores'', generalizing the uniform case. We ask how  public scores, which correspond to the attractiveness of agents, impact the outcome in the market. 

Formally, we study the following class of {\it tiered random markets}. There are  $n$ men and $n$ women. Each side of the market is divided into a constant number of ``soft tiers''. There is a fraction of $\epsilon_i$ women in tier $i$, each of which has a public score $\alpha_i$. And there is a fraction of $\delta_j$ men in tier $j$, each of which has a  public score $\beta_j$.  For each agent we draw a complete preference list  by sampling without replacement proportionally to the public scores of agents on the other side of the market.\footnote{These are also termed popularity-based preferences   \cite{gimbert2019popularity,immorlica2015incentives} and also equivalent to  generating preferences according to a  Multinomial-Logit (MNL) induced by the public scores.}
So a man's preference list is generated by sampling women one at a time without replacement according to a distribution that is proportional to their public scores. Using $\bm{\alpha}, \bm{\epsilon}$ to denote the vector of scores and proportions of tiers on the women's side, we see that the marginal probability of drawing a woman in tier $i$ is $\alpha_i / (n\epsilonDotAlpha)$. An analogous statement holds for the tier configuration $\bm{\beta}, \bm{\delta}$ of the men.
These preferences are a natural next-step beyond the uniform distribution over preference lists, and provide a~priori heterogeneous quality of agents while still being tractable to theoretical analysis.

Our primary goal is to study  the \emph{average rank} of agents in each
tier under the man-optimal stable matching, with a focus on the asymptotic
behavior in large markets. The rank of an agent is defined to be the index
of their partner on their full preference list, where lower is better. 
Additionally, we prove results on the \emph{match type distribution},
i.e. the fraction of tier $i$ women matched to tier $j$ men
(for each $i, j$).

We show that, for large enough markets, the following hold
to within an arbitrarily small approximation factor:
\begin{enumerate}
    \item (Theorem~\ref{thrmMenRanksCentralConcentration}.)
    With high probability, the average rank of men in tier $j$ is
    \[ \frac{\epsilonDotAlpha}{\alpha_{\min}}\cdot
        \frac{1}{\bm{\delta}\cdot\bm{\beta}^{-1}}\cdot\frac{\ln n}{\beta_j}. \]
    \item (Theorem~\ref{thrmWomenRanks}.)
    With high probability, the average rank of women  in tier $i$ is 
      \[ (\deltaDotBeta)(\deltaDotBeta^{-1}) \frac{\alpha_{\min}}{\alpha_i}
      \frac{n}{\ln n}. \]
    \item (Theorem~\ref{thrmMatchTypes}.)
      The probability that a woman in tier $i$ matches to a man
      in tier $j$ is $\delta_j$.
\end{enumerate}

In the above, $\bm{\beta}^{-1} = \{1/\beta_j\}$ denotes
the vector of the reciprocals of men's public scores,
$\alpha_{\min}$ denotes the smallest public score
on the women's side, and $\bm{x}\cdot\bm{y}$ denotes
the dot product of the vectors $\bm{x}$ and $\bm{y}$.

\subparagraph*{Intuition and Observations}
As in the case of uniform 
preferences~\cite{pittel1989average}, in the man-optimal stable outcome, 
men get a much lower rank than women. Indeed,
both men and women get the same order of rank
as in the uniform case ($\ln n$ and $n/\ln n$, respectively). 
This in itself
is an interesting consequence of this work -- a constant
tier structure affects the market only up to constants.
This fact also highlights that determining these constants
is an interesting area for investigation,
as the constants capture how the outcome of the market changes
with respect to the public scores.
The first observation we make is that
agents on each side get a rank inversely
proportional to their public score.


Perhaps more interesting is the following observation:
The rank of both sides depends on the tier structure of the other side, but
\emph{each tier is affected the same amount} by the tier parameters of the
other side.
This is closely related to the fact 
that the probability of a woman matching to a man in tier $j$
is proportional to only the number of men in tier $j$
(regardless of the tier the woman lies in).
Moreover, both $\epsilonDotAlpha / \alpha_{\min}$ and 
$(\deltaDotBeta)(\deltaDotBeta^{-1})$ are always greater
than or equal to one\footnote{ 
  To prove $(\deltaDotBeta)(\deltaDotBeta^{-1}) \ge 1$,
  use Jensen's inequality to conclude that
  $\sum_j \delta_j \beta_j \ge \left(
  \sum_j \delta_j \beta_j^{-1}\right)^{-1}$.
}.
Thus, in these markets, any heterogeneity in the public scores 
of one side harms the average ranks of the other side
(but does not significantly affect the likelihood
that an agent matches to a certain tier on the other side).

Another interesting feature is the following: While the average
ranks for men's tiers  depend on public score distributions on both sides
of the market, the average rank of women in tier $i$ depends only on the
ratio between $\alpha_i$ and the  public score $\alpha_{\min}$ of the
bottom tier of women (and the distribution of public scores on the men's
side). Intuitively, the rank of the men depends on the distribution of
scores of the women because \emph{men are competing to avoid being
matched to the lowest tier of women}.

To elaborate on that last point, let us first consider the total number of
proposals made during the man-proposing deferred acceptance process (DA).
The algorithm will terminate
when the last woman receives a proposal. Naturally one would expect that
this woman will belong to the bottom tier. Therefore, using standard coupon
collector arguments,  the total number of proposals made to women {\it in
the bottom tier} until they all receive a proposal is expected to be
$(\epsilon_{\min} n)\ln(\epsilon_{\min} n)$,  where $\epsilon_{\min}$ is
the fraction of women in the bottom tier.
These proposals are a 
$\epsilon_{\min} \alpha_{\min} /\epsilonDotAlpha $ fraction of the total
proposals, so one expects the number of total proposals to be
\[ \frac{ (\epsilon_{\min} n)\ln(\epsilon_{\min} n) }
    {\epsilon_{\min} \alpha_{\min} /\epsilonDotAlpha}
  = \frac{\epsilonDotAlpha} {\alpha_{\min}}
    \cdot n \ln n - O(n).
\]
This introduces the factor of $\epsilonDotAlpha/\alpha_{\min}$ in result (i) on the men's
ranks (i.e. the number of proposals per man).

On the other hand, the probability that one of these proposals goes to a
woman in tier $i$ is $\alpha_i / (n\epsilonDotAlpha)$, implying that such a
woman should receive roughly $(\alpha_i / \alpha_{\min}) \ln n$ proposals.
Thus, for a given woman, the increase in the total number of proposals
caused by the tier proportions $\bm{\epsilon}$ is exactly canceled
out by the likelihood that a proposal goes to that woman,
and the only thing that matters is the woman's score (relative
to the bottom tier).
If men are uniform, women should then expect rank roughly $(\alpha_{\min} /
\alpha_{i})(n/\ln n)$, which helps explain the corresponding factors 
in result (ii). 

Consider now the public scores of the men, and for simplicity assume that the bottom tier of men has score $1$.
Suppose for the sake of demonstration that every time a man with public score $\beta_j$  proposes to a woman who is already matched, this man is $\beta_j$ times more likely to be accepted than a man with than a man with public score $1$.\footnote{
    As we discuss below,
    this approximation is only valid if the woman is already matched with a man she ranks highly.
    A major technical step in our proof is showing that,
    in certain situations, ``enough'' women are ``matched
    well enough'' for this approximation to be used.
}
We would expect that such a man makes a $1/\beta_j$ fraction fewer proposals before his next acceptance, and indeed $1/\beta_j$ fewer proposals overall.
Let $S$ be the total number of proposals, let $r_j$ denote the rank of a man in tier $j$, and $r_{\min}$ the rank of the bottom tier of men. If every tier of size $\delta_j n$ each accounts for a share of proposals proportional to $1/\beta_j$, then we should have
\[ S = \sum_j (n\delta_j) \beta_j^{-1} r_{\min}
  \qquad \Longrightarrow \qquad
  r_{\min} = \frac{S}{n\deltaDotBeta^{-1}},\qquad
  r_j = \frac{S}{(n\deltaDotBeta^{-1})\beta_j},
\]
which introduces the factor of $1/((\deltaDotBeta^{-1})\beta_j )$
in result (i) on the men's rank.

The final remaining factor in our results
is $(\deltaDotBeta)(\deltaDotBeta^{-1})$ in result (ii).
Deriving this term requires reasoning about the number
of proposals from each tier of men received by a fixed woman $w$.
Building from the previous paragraph,
we reason that each of the $\delta_j n$ men in tier
$j$ makes a number of proposals proportional to $1/\beta_j$.
Each such proposal has the same probability of going to $w$,
regardless of the tier $j$.
So the number of proposals $w$ receives from tier $j$ men
is proportional to $\delta_j/\beta_j$.
The factor $(\deltaDotBeta)(\deltaDotBeta^{-1})$ then
arises for somewhat technical reasons 
(described in section~\ref{sectionWomenRanks}) which have to do
with the way women generate their preference lists.

We now describe how result (iii), which may seem somewhat
more mysterious than the other results, emerges as a corollary
of computing the ranks women receive.
We argued above that a woman $w$ in tier $i$ receives approximately
$(\delta_j / \beta_j) U_i$ proposals from men in tier $j$,
for some value of $U_i$ independent of $j$.
Recall that $w$ applies weight $\beta_j$
to each proposal she sees from a man in tier $j$.
Moreover, the identity of $w$'s favorite proposal
is independent of the order in which $w$ saw proposals.
Thus, the probability that $w$'s favorite proposal
(i.e. the proposal of the man she matches to) came from tier $j$
is approximately $(\delta_j U_i) / U_i = \delta_j $, which is
\emph{independent of $\beta_j$}, as well as independent of the
tier $w$ is in.
Thus, up to lower order terms, the distribution of match
types is the same as it would be in a uniformly random matching 
market,
and the match is not assortative.

Intuitively, result (iii) arises when men make enough proposals
to offset any disadvantage (in the type of their match)
they have due to public score.
Due to the highly connected and relatively competitive nature
of our markets,
men in the lowest tier make more proposals, 
but they are not more likely to end up matched with lower tier agents. 
Put another way, men in lower tiers are less likely
to attain matches they idiosyncratically like,
but often settle for a high-quality agent which
is low on their personal preference list.
This indicates that public scores that  differ by  constant weight do not provide any significant a~priori predictive power
over the matches agents receive. 
In particular, agents with lower public scores
can still hope to achieve high-tier
matches if they consider enough options.

\subparagraph*{Techniques.}
Our proofs require developing some technical tools that may be of independent interest, especially when we reason about the ranks achieved by the men. We build on the analysis of DA from \cite{wilson1972analysis,pittel1989average,immorlica2015incentives,ashlagi2017unbalanced} to handle public scores rather than just uniform random preferences. As in these previous works, a key step in our proof is letting all men but one (call him $m$) first propose and match though DA, and then tracking the proposals of $m$ (this works because DA is independent of the order of proposals). For demonstration purposes, let's call the proposals before man $m$ the ``setup''. A key fact in previous works is that the distribution of proposals made by $m$ is identical for every man, and moreover that the distribution of setups is identical as well. This fails to hold in tiered random markets, and thus we must develop new techniques. 

We prove that, for ``most'' setups, the rank a man can achieve is approximately given by a certain geometric distribution, whose parameter $p$ is essentially the probability that a proposal by that man will be accepted. We then prove that, up to lower order terms, this success parameter scales up with the public score of the men. This gives the fact that the rank of men is inversely proportional to public score. 

Characterizing the setups where our proof goes through requires a technical analysis, and we term the setups which work ``smooth matching states''. The most crucial thing we need for these setups is that \emph{many women are matched to partners they rank highly}, which helps us prove that 1) men are likely to remain matched to their first acceptance (so our approximation with a geometric distribution is valid), and 2) a man with fitness $\beta$ is approximately $\beta$ times more likely to be accepted every time. For details, see section~\ref{sectionRankOfMen}.

Finally, to prove that the average rank of men within a tier \emph{concentrates}, we need to show the correlation between the ranks of different men is not too large. Thus, we track the proposals of the last \emph{two} men to propose, and find that the joint distribution of the ranks of these men can be approximated by a pair of independent geometric distributions. Intuitively, this is because men do not propose to very many women overall, and thus the last two men are unlikely to interfere with each other as they make proposals.

The crucial aspects of our model are that preferences
of each agent are independent and identically distributed,
that preference weights are constant, and that the market is roughly
in balance. While our techniques are useful to reason about
markets which do not have these properties,
the results are not nearly as clean; indeed the tier structure simplifies our analysis, but most of it 
goes through if each agent has an
individual, constant, bounded public score.

\subsection{Related literature}
\label{sec:relatedliterature}


Several  papers have studied matching markets with complete preference lists that are  generated uniformly at random. Coupon collector techniques are used in \cite{wilson1972analysis}  to upper bound the men's average rank  by $\ln n$.  \cite{pittel1989average,knuth1990stable,pittel1992likely} analyze further  balanced markets with $n$ men and $n$ women. They find that in the man-optimal stable matching in balanced markets, men and women match on average to their $\ln n$ and $\frac{n}{\ln n}$ ranks, respectively. Our results generalize these findings to markets with preferences induced by public scores, thus incorporating much more heterogeneity in the market. 

\cite{ashlagi2017unbalanced,pittel2019likely,cai2019short} study markets with uniformly drawn preferences but with an imbalance between men and women. These papers find that in any stable matching the average ranks of men and women are similar to the average ranks under the short-side-proposing DA. Additionally,  \cite{kanoria2020random} investigates the relation between the imbalance and the length of preference lists (though the model is still uniform for each agent). This paper does not consider imbalanced markets but we believe that similar techniques to those we develop will be useful to reason about unbalanced tiered random markets.

Several papers look at random matching markets in which preferences are generated based on public scores \cite{immorlica2015incentives,kojima2009incentives,ashlagi2014stability}. These papers restrict attention to the size of the core (a measure of the difference between the man-optimal and woman-optimal outcome) and strategic manipulation of agents under a stable matching mechanism.  Key assumptions in these papers generate outcomes which leave many agents unmatched. In particular, their models either assume that  preference lists of men are of constant length, or, alternatively, one side has many more agents than the other.\footnote{Some papers additionally consider manipulations in more restricted randomized settings~\cite{coles2014optimal} or in deterministic (worst case) settings~\cite{Gonczarowski14}.}

Closely related to this paper is~\cite{gimbert2019popularity},
which primarily studies a special case of highly correlated
popularity preferences which is termed ``geometric preferences''.
While our work focuses on the rank agents achieve in the
man-optimal outcome (a canonical stable matching),
\cite{gimbert2019popularity} focuses on the size
of the core (more specifically, they study the number of stable partners
that agents have in typical stable matchings) using techniques
specialized to geometric preferences.


Other papers have addressed tiered matching markets,
especially in market design settings.
However, these papers mostly study ``hard tiers'',
i.e. such that agents in higher tiers are deterministically
ranked above lower tiers by every agent on the other side.
Examples include~\cite{BeyhaghiST17, AshlagiBKS17}.
\cite{Lee16} also  considers a certain restricted
tiered model of cardinal utilities (which is incomparable
with our model), focusing on
which tier of agents match to which tier.

Our contribution to the literature is a detailed study
of ``soft tiers'', a natural special case of the
popularity preferences of~\cite{immorlica2015incentives, kojima2009incentives, gimbert2019popularity}.
In cases where each agent's utility for each match on the other side
is independent and identically distributed,
popularity preferences are the natural next step beyond
uniform markets, as they model situations where agents on each
side have significant but non-definitive variation in a~priori quality.
Our techniques build on the large body of work analyzing the 
``proposal dynamics'' of deferred acceptance for random preferences,
such as \cite{wilson1972analysis,immorlica2015incentives,
ashlagi2017unbalanced,gimbert2019popularity}.
Our results give insight into how constant-factor preference
biases affect stable matching markets,
including the first explicit calculations of expected rank
beyond uniform markets.

The rest of the paper is organized as follows: Section~\ref{sectionDefsAndPreliminaries} offers basic definitions and preliminaries for our discussion. Section~\ref{sectionCouponCollector} studies the tiered coupon collector process, which serves as an important coupling process for the deferred acceptance algorithm. Section~\ref{sectionRankOfMen} and~\ref{sectionWomenRanks} present the core results of this paper, namely the average rank among tiers of men and women. Section~\ref{sectionSimulations} showcases some numerical experiments.

\section{Definitions and Preliminaries}\label{sectionDefsAndPreliminaries}

A matching  market consists of a finite set of men $M$ and a finite set of women $W$. Each man (woman) has a complete and strict preference list over women (men). A matching is a mapping $\mu:M\cup W\to M\cup W$ such that: for every $m\in M$, $\mu(m)\in W$ (or $\mu(m)$ is undefined), for every woman $w\in W$, $\mu(w)\in M$ (or $\mu(w)$ is undefined), and for every $m\in M$ and $w\in W$, $\mu(m)=w$ if and only if $\mu(m)=w$. A matching $\mu$ is {\it stable} if no man-woman pair who are not matched in $\mu$ prefer each other to their matched partners.

It is well-known that there is a unique man-optimal stable matching, which can be found using the man-proposing deferred acceptance  algorithm (DA). 
While this algorithm does not fully specify an execution order, it is a
classically known result that the order does not affect the final outcome.

\begin{algorithm}
 \SetAlgoLined
 \caption{(Man-Proposing) Deferred Acceptance Algorithm (DA)}
 \label{algDA}
 Initialize matching $\mu$ to be empty (i.e. every agent's partner is undefined)\;
 Initialize $\mathcal{U}=M$ to be the set of all unmatched men\;
 \While{$|\mathcal{U}|>0$}{
  Choose any $m\in\mathcal{U}$\;
  Let $m$ propose to his most preferred woman
    $w$ to whom he has not made a proposal yet\;
  \If{$w$ prefers $m$ to $\mu(w)$ (or if $\mu(w)$ is undefined)}{
   \lIf{$\mu(w)$ is defined}{Add $\mu(w)$ to $\mathcal{U}$}
   Remove $m$ from $\mathcal{U}$\;
   Assign $\mu(w)=m$\;
   }
 }
\end{algorithm}

\begin{lemma}[\cite{gale1962college, mcvitie1970stable}]
  \label{thrmDAExecutionOrderDoesntMatter}
  The same proposals are made in every run of DA, regardless of which
  man is chosen to propose at each step.
\end{lemma}

We study the man-optimal stable matching in a class of tiered random markets, which will be defined below. We will assume that $|M| = |W|$ 
and  that no agent finds any other agent on the other side unacceptable.  We will also assume that each side draws their preferences from an identical and independent underlying distribution, and moreover these preferences are generated by repeatedly sampling
without replacement from a fixed distribution on the \emph{agents} of each side.
In \cite{immorlica2015incentives, gimbert2019popularity},
this  assumption is termed ``popularity-based preferences'', with the weight
of an agent in the distribution intuitively indicating their popularity for
agents on the other side. 

Our main goal is to study randomized matching markets with \emph{a constant number of constant weight tiers} of agents on each side.
For this entire paper, we consider the tier structure to be defined by
fixed proportions $\bm{\epsilon}, \bm{\delta}$ of agents in each tier
and constant weights $\bm{\alpha}, \bm{\beta}$ for each tier,
and we investigate the outcome of the man-proposing DA as $n\to\infty$.

\begin{definition}
  Consider constant vectors $\bm{\alpha}, \bm{\epsilon} \in \R^{k_1}_{> 0}$
  and $\bm{\beta}, \bm{\delta} \in \R^{k_2}_{>0}$,
  where $\|\bm{\epsilon}\|_1, \|\bm{\delta}\|_1=1$. 
  A \emph{tiered matching market} of size $n$ with respect to
  $\bm{\alpha}, \bm{\epsilon}, \bm{\beta}, \bm{\delta}$ is defined
  by generating agents' preference lists as follows:
  \begin{itemize}
    \item The set of $n$ women $W$ is divided into tiers $T_1, \ldots, T_{k_1}$,
      of size $|T_i| = \epsilon_i n$ each\footnote{
        Note that, for most vectors $\bm{\epsilon}, \bm{\delta}$,
        many values of $n$ will produce tier sizes which are not integers.
        However, as all our results are \emph{continuous} in
        $\bm{\epsilon}, \bm{\delta}$ this is not a problem -- for any particular fixed
        $n$, each tier size can be rounded in a way that effectively just changes
        $\bm{\epsilon}, \bm{\delta}$ by a tiny amount, and our results will still
        hold as written as $n\to\infty$.
      }.
      Define a distribution $\W$ on women such that a woman in tier $i$ is
      selected with probability proportional to $\alpha_i$.
      That is, the weight of $w\in T_i$ in $W$ is 
      $\alpha_i / (n\epsilonDotAlpha)$ (which we often denote by $\pi_i$).
    \item The set of $n$ men $M$ is divided into tiers
      $T_1, \ldots, T_{k_2}$, of size $|T_j| = \delta_j n$ each.
      Define a distribution $\M$ on men such that a man in tier $j$ is
      selected with probability proportional to $\beta_j$.
      That is, the weight of $m\in T_j$ in $M$ is 
      $\beta_j / (n\deltaDotBeta)$.  
  \end{itemize}
  For each man $m$ independently,
  women are repeatedly sampled from $\W$ without replacement,
  and the order in which women are selected is $m$'s preference list.
  Preferences for the women are analogously drawn over the distribution $\M$.
  The rank that a man has for a woman $w$ is the index of $w$ on his preference list (where lower is better).
\end{definition}

We refer to each $\alpha_i$ as the \emph{weight} or \emph{public score}
of the women in tier $i$, and similarly for the men.
For simplicity of certain arguments, we assume that each $\alpha_i\ge 1$
and each $\beta_j\ge 1$ (although for clarity of our results, we do not
assume that the smallest weight is exactly $1$). We write $\alpha_{\min}$ for the weight of the bottom tier of women, and $\epsilon_{\min}$ for the corresponding tier proportion.

Using a simple generalization of the ``principle of deferred decisions''
used in \cite{Knuth76}, we can arrive at a
characterization of the random process of running DA with a tiered matching
market.

\begin{lemma}\label{thrmWomenDeferedDecisions}
  The distribution of runs of DA
  for a tiered matching market can be generated as follows:
  For the men,
  every time a man is chosen to propose, he samples a woman at random from
  $\W$, and repeats this until he samples a woman who he has not yet
  proposed to.

  For the women,
  suppose $w$ has seen proposals from a set of men $p(w)$,
  and let $\Gamma_w = \sum_{m\in p(w)}\beta(m)$, where $\beta(m)$
   denotes the public score of a man $m\in p(m)$.
  Then if a proposal from a man $m_*$ with public score $\beta_*$
  arrives, $w$ accepts the proposal from $m_*$
  with probability
  \[
    \frac{\beta_*}{\beta_* + \Gamma_w}.
  \]
\end{lemma}
\begin{proof}
  The above formula  gives the probability that $m_*$ is chosen
  as $w$'s favorite out of the set of men $p(w)\cup\{m_*\}$.
  The only additional observation we need to make is that
  the probability that $m_*$ is the new
  favorite is independent of the identity of the old favorite.
\end{proof}
We often call $\Gamma_w$ the total ``weight of proposals'' woman $w$ has seen at some point during DA.

\subsection{Deferred acceptance with re-proposals}

With respect to any popularity-based model of preferences,
we can define a procedure analogous to DA.
In our case, we will show that the difference between DA and this procedure
is indeed small.

\begin{definition}
  Consider any random matching market with men's preferences determined by
  sampling from a distribution $\W$ over women.
  The \emph{deferred acceptance with re-proposals} algorithm
  is defined as being identical to Algorithm~\ref{algDA}, except
  \begin{itemize}
    \item Every time a man is chosen to propose to a woman, he draws a
      woman from $\W$ with replacement, and may propose more than once to a
      single woman.
    \item Women's preferences are consistent throughout proposals from the
      same man (so if a woman rejected a man before, she will reject him
      again).
  \end{itemize}
\end{definition}

Since re-proposals are ignored, this process will always yield
the same outcome as algorithm~\ref{algDA}.

%
%

\subparagraph*{Notation.}
We write $x = (1\pm \epsilon) y$ to mean 
$(1-\epsilon) y \le x \le (1+\epsilon) y$.
We let $\epsilon$ denote an arbitrarily small constant
greater than $0$, while $\bm{\epsilon}$ and $\epsilon_i$
denote the tier parameters of the women.
We let $\alpha_{\min}$ denote the smallest public score
for the women's side, and $\epsilon_{\min}$ denotes
the corresponding tier proportion.
We let $\bm{v}\cdot \bm{w}$ denote the inner product
of vectors $\bm{v}, \bm{w}$.
We denote the exponential and geometric distributions by $\Exp(\lambda)$
and $\Geo(p)$, respectively.
We denote the fact that a random variable $X$ is a draw from a distribution
$D$ by $X\sim D$.
We use $X\preceq Y$ to denote the fact that $X$ is statistically dominated
by $Y$ (i.e. for all $t\in \R$, we have $\P{X \ge t}\le \P{Y\ge t}$).
We let $\Cov(X,Y)$ denote the covariance of $X$ and $Y$.
We write $f(n) = \widetilde O(g(n))$ if there exists a constant $k$
such that $f(n) = O( g(n) \log^k (g(n)))$.

\section{The Coupon Collector and the Total Number of Proposals}\label{sectionCouponCollector}

Fix a tier structure $\bm{\alpha}, \bm{\epsilon}$ corresponding to men's
preferences over the women.
Consider running deferred acceptance with re-proposals.
Recall that each man samples a woman in tier
$i$ with probability $\pi_i = \alpha_i/(n\bm{\epsilon}\cdot\bm{\alpha})$
each draw.
Define $\pim = \alpha_{\min}/(n\bm{\epsilon}\cdot\bm{\alpha})$
as the probability of drawing a woman in the lowest tier
(and keep in mind that $\pim$ scales like $O(1/n)$).

The core tool we use to reason about the total number of proposals in DA is
the classically studied coupon collector process.
In particular, we study this process when 
coupons from different tiers are drawn with a constant-factor
difference in probability.

\begin{definition}
  \label{defCouponCollector}
  Given a probability distribution $(p_i)_{i\in [n]}$,
  we define the \emph{coupon collector with unequal probabilities}
  as follows: once every time step, an integer from $[n]$
  is drawn independently and with replacement according to distribution
  $(p_i)_{i\in [n]}$.
  The coupon collector random variable with respect to $(p_i)_{i\in [n]}$
  is defined as the number of total draws required
  before every integer in $[n]$ has appeared at least once.

  The coupon collector $T$ which we are interested in is defined
  by taking the distribution $\W$ of men's preferences.
\end{definition}

We will show in section~\ref{sectionTotalProposals} that, in our case,
this random process
is also very close to that of DA (without re-proposals).
For now, we simply bound the expectation of the coupon collector
(with the proof deferred to appendix~\ref{appendixCouponExpectation}).
Note that similar probabilistic problems have been considered before (see e.g.
\cite{brayton1963asymptotic, doumas2012couponsRevisited}) but we include
our own full proofs in appendix~\ref{appendixCouponExpectation}
and~\ref{appendixCouponConcentration} for completeness.

\begin{restatable}{theorem}{RestateThrmCouponExpectation} 
  \label{thrmCouponCentralConcentration}
  Let $T$ denote the number of draws in a coupon collector process with
  weights proportional to $\W$.
  We have
  \[ \E{T} = \big(1 \pm O(1/\ln n)\big)
    \frac{\epsilonDotAlpha}{\alpha_{\min}} n \ln n.
  \]

\end{restatable}

\begin{remark}\label{remarkLogEspMinErrorBody}
  While we are mostly interested in the asymptotic performance of these
  matching markets, we make one comment here that the above big-$O$ notation
  hides a constant factor of order $\ln(1/\epsilon_{\min})$. For small
  values of $\epsilon_{\min}$, this can be much larger than $\ln n$ for most
  realistic market sizes.
  Note that this error term already
  showed up in the intuition given in section~\ref{sectionIntroduction},
  where our estimate for the total number of proposals had an additive
  term of $O(\ln(\epsilon_{\min}) n)$. For more information,
  see proposition~\ref{thrmCouponExpectationLowerBound}.
\end{remark}

\subsection{The Total Number of Proposals in Deferred Acceptance}
\label{sectionTotalProposals}

%

Let $S = S_n$ denote the total number of proposals made a run of DA
with random preferences given by our tiered market.
As before, let $T = T_n$ denote the distribution of a coupon collector with
distribution $\W$.
As in many prior studies of randomized deferred acceptance,
our starting point is the fact that $S$ is statistically dominated by $T$:

The connection to stable matchings is the following very simple
observation, which has been used in many previous works~\cite{knuth1990stable, pittel1989average}:
\begin{proposition}
  \label{thrmCouponIsDAReprops}\label{thrmStatisticalDominance}
  The coupon collector random variable $T$ is distributed identically to
  the total number of proposals made in deferred acceptance with
  re-proposals (regardless of the preferences that women have for men).

  Moreover, if $S$ is the number of proposals in DA,
  then $S\preceq T$ (i.e. $S$ is statistically dominated by $T$).
\end{proposition}
\begin{proof}
  First, recall that DA terminates as soon as every man is matched.
  Observe that women never return
  to being unmatched once they receive a single proposal.
  Because the market is balanced (i.e. $|W|=|M|$), this means DA will
  terminate as soon as every woman has been proposed to.
  Moreover, because re-proposals are allowed, every proposal is
  distributed exactly according to $\W$.
  Thus, ignoring the identity of the man doing the proposing,
  $T$ is distributed exactly according to the coupon collector random
  process.

  Furthermore, we can recover the exact distribution $S$ of proposal in
  DA simply by ignoring each repeated proposal in $T$.
  Thus, $S \le T$ for each run of deferred acceptance with re-proposals, so
  $S\preceq T$.
\end{proof}

We  proceed to show that the upper bound provided by $T$ is essentially
tight, i.e. there is not a big difference between $T$ and $S$.
The key step will be to upper bound maximum number of distinct women any
man proposes to in $S$, and thus upper bound the
probability that any proposal
in $T$ is a repeat for the man making the proposal.
Crucially, this lemma will have to account for the preferences of the women
(which up until this point have been ignored, but which play a significant role
in the distribution of proposals in DA).
Recall that we denote the sizes of the tiers of the men by the vector
$\bm{\delta}$, and the public scores of the men in each tier by $\bm{\beta}$.

\begin{lemma} \label{thrmMaxPropsLastMan}
Consider running DA with all men except $m_*$,
and suppose that at most $O(n\ln n)$ proposals are made during this process.
Afterwards, consider $m_*$ joining and run DA until the end.
Then for any $C\ge 0$, with probability
$1 - 1/n^C$, the number of proposals made by $m_*$ is at most $O(C\ln^2 n)$.
\end{lemma}
\begin{proof}
This proof follows a similar logic as the proof of Lemma B.4 (ii) in
\cite{ashlagi2017unbalanced}. 
Suppose $m_*$ has public score $\beta_*$,
and that he proposes at the end (and $O(n\ln n)$ prior proposals have been made).
We proceed as follows:
\begin{enumerate}
  \item When $m_*$ makes a proposal, he will choose a woman who he has
    not yet proposed to. For some fixed proposal index $i$ of $m_*$,
    let's denote the set of all women $m_*$ has not proposed to by $W_*$, and denote by $\W_*$ the distribution
    of $m_*$'s next proposal, i.e. a sample over $W_*$ weighted by the
    public scores $\alpha_i$.
    For a women $w$ denote her sample weight by $\alpha(w)$ and the
    set of proposals she has received by $p(w)$.
    Further denote by $\Gamma_w = \sum_{m\in p(w)} \beta(m)$ the sum of the
    public scores of men who have proposed to $w$.

    Suppose that $|W_*|\ge n/2$, i.e. that $m_*$ has not yet proposed to
    over half the women.
    Using the assumption that the total number of proposals made is at most
    $O(n\ln n)$,
    we can bound the expected total weight of proposals women have seen by
    \[ \Es{w\sim \W^*}{\Gamma_w}
      = \frac{\sum_{w\in W_*} \alpha(w)\Gamma_w}
        {\sum_{w\in W_*} \alpha(w)}
      \le \frac{ \alpha_{\max}\sum_{w\in W}\Gamma_w}
          {|W_*|\alpha_{\min}}
      \le \frac{ \alpha_{\max}\beta_{\max}\cdot O(n\ln n)}
          {|W_*|\alpha_{\min}}
      \le O(\ln n).
    \]
    Thus, by lemma~\ref{thrmWomenDeferedDecisions}, the probability that
    the proposal by $m_*$ will be accepted is
    \[
      p_1
      := \Es{w \sim \W_*}{\frac{\beta_*}{\beta_* + \Gamma_w}}
      \ge \frac{\beta_*}{\beta_* + \Es{w\sim \W_*}{\Gamma_w}}
      \ge \Omega(1/\ln n).
    \]
    where the first inequality is due to Jensen's inequality.

  \item If $m_*$ proposes to $w$ and is accepted, then the subsequent
    rejection chain can either end at the last woman without proposals,
    $w_{\mathrm{last}}$,
    or cycles back to $w$ who this time rejects $m_*$. Notice that
    for each subsequent proposal, the ratio between the probability that it
    goes to $w_{\mathrm{last}}$ 
    (in which case the process will be terminated) and the
    probability that it returns to $w$ is at most
    $\alpha_{\max}:\alpha_{\min}$ (and possibly less if the proposing man
    has already proposed to $w$). Hence, the probability that the chain
    ends at the last women $w_{\mathrm{last}}$ is bounded below by
    \[p_2 :=
      \frac{\alpha_{\min}}{\alpha_{\max}+\alpha_{\min}}
      \ge \Omega(1).
    \]
    Note that this is ignoring the chance that a new proposal by $w$ is
    rejected, but it still suffices for a lower bound. 

  \item \label{itemMaxPropsLastProbWrapup}
    The probability that $m_*$ makes more than $K\ln^2 n$
    proposals is thus bounded above by
    \[ (1-p_1p_2)^{K\ln^2 n} \leq \exp(-p_1p_2 K\ln^2 n)
      = \exp(-\Omega(K\ln n))
      \le n^{-C}
    \]
    as long as we choose $K = \Omega(C)$ large enough.
\end{enumerate}
\end{proof}

\begin{corollary}\label{thrmMaxIndividualProposals}
  For any constant $C\ge 1$,
  with probability $1 - 1/n^C$, the maximum number of proposals made by any
  man in DA is $O(C\ln^2 n)$.
\end{corollary}
\begin{proof}
  By \ref{thrmStatisticalDominance} and \ref{thrmCouponUpperTail},
  the total number of proposals made in DA is $O(C n\ln n)$ with probability
  $1 - 1/n^C$. In particular, if we consider any $m_*$ and let all other
  agents propose, this will be true.
  Recall that by lemma~\ref{thrmDAExecutionOrderDoesntMatter},
  DA is independent of the order in which men are chosen to propose.
  Thus, for each man $m_*$ we can apply lemma~\ref{thrmMaxPropsLastMan} to get
  that, with probability $1 - 1/n^{C+1}$, $m_*$ makes fewer than $O((C+1)\ln^2 n)=O(C\ln^2 n)$
  proposals. Taking a union bound over the $n$ men gets the desired result.
\end{proof}

\begin{remark}
  Both of the above results hold for deferred acceptance with re-proposals
  as well as deferred acceptance.
  Indeed, even with re-proposals, deferred acceptance will be independent
  of the order of proposals (as re-proposals are ignored by the women).
  Moreover, the logic required to prove points 1. and 2. of the proof of
  lemma~\ref{thrmMaxPropsLastMan} is only easier to prove when men sample
  over all of $W$ as opposed to just the set $W_*$.
\end{remark}

The above result is enough to show that
proposition~\ref{thrmCouponCentralConcentration} holds for DA as well for
the coupon collector,
because repeated proposals are at most a $O(\ln^2 n / n) = o(1)$
fraction of total proposals in deferred acceptance
with re-proposals. We defer the proof to 
appendix~\ref{appendixTotalProposalsExpectation}.

\begin{restatable}{theorem}{RestateThrmDACentralConcentration}
\label{thrmDACentralConcentration}
  Let $S$ be the total number of proposals made in DA
  with tiers of women $\bm{\epsilon}, \bm{\alpha}$,
  and arbitrary constant tiers on the men.
  We have
  \[ \E{S} = \big(1 - O(\ln^2 n/n)\big) \E{T} 
    = \big(1 \pm O(1/\ln n)\big)
      \frac{\epsilonDotAlpha}{\alpha_{\min}} n \ln n.
  \]

\end{restatable}


\section{Rank Achieved by the Men}
\label{sectionRankOfMen}

Up until this point, our arguments have only crudely considered the 
preferences women have for men.
Due to the asymmetry across the different tiers, this
means we cannot yet calculate the expected rank men get.

Consider a man $m$ in tier $j$. Our main goal is to prove that the rank of $m$ is inversely proportional to $\beta_j$.
As in~\ref{thrmMaxPropsLastMan},
the core tool of our proof will be the fact that deferred acceptance is
independent of execution order (by~\ref{thrmDAExecutionOrderDoesntMatter}),
and thus we can wait until all other men
have finished proposing and found a match before letting $m$ propose.
Once this is done, the major ideas are
\begin{enumerate}
  \item Suppose $m$ has public score $1$, and define
    \[ p = \Es{w\sim\W}{\P{\text{$w$ accepts a proposal from $m$}}}. \]
    Note that, if $m$ were able to propose to a woman independently
    multiple times, the number of proposals until $m$ gets his first
    acceptance would be distributed exactly according to $\Geo(p)$,
    and the expected value would be $1/p$.
    We show that (because men make much less than $n$ proposals) the
    difference due to re-proposals is not large.

  \item Because $m$ is the last man to propose, most women have already
    seen many proposals and arrived at a decent match.
    When $m$ gets his first acceptance, he should thus be likely to stay
    where he is. We show that, while the probability of $m$ proposing to 
    more women is non-negligible,
    it still contributes only $O(1)$ in expectation.
    So $m$'s expected rank is $1/p$ up to lower-order terms.

  \item Another consequence of a woman $w$ receiving a large number 
    of proposals is the following:
    \begin{align*}
      \P{\text{$w$ accepts a proposal from $m'$ with weight $\beta$}} 
      \qquad\qquad\qquad\\
       \approx \beta\cdot
      \P{\text{$w$ accepts a proposal from $m$ with weight $1$}}.
    \end{align*}
    simply by~\ref{thrmWomenDeferedDecisions} and the fact
    that $\beta / (\beta + \Gamma_w) \approx \beta \cdot 1/(1+\Gamma_w)$
    for $\Gamma_w$ (the sum of public scores of men who proposed to $w$) large.
    Thus, if $m$ had public score $\beta$, the effective value of $p$ would be
    approximately $\beta p$, and the expected rank of $m$
    would become approximately $1/(\beta p)$.
    In other words, while we are not able to calculate $p$ directly, we
    show that 
    $p$ scales properly with $m$'s score.

  \item \label{itemIssueTierSelectionPrior}
    Finally, we prove that the above holds for most sequences of
    proposals of men before $m$,
    and thus holds in expectation over the entire execution of DA.
    Note that the distribution of proposals before $m$ changes slightly
    depending on which tier $m$ is chosen from, but in a large market, we
    do not expect this to make a big difference.
\end{enumerate}

The biggest difference between the above proof sketch
and its implementation is that we focus on \emph{two} men proposing
at the end of DA. This serves to address
point~\ref{itemIssueTierSelectionPrior} above -- we are able to show that,
for the vast majority of sequences of proposals before the last two men,
their expected ranks are proportional to the ratio of their scores.
Thus, this ratio holds in expectation over all of DA.
Focusing on two men also allows us to bound the \emph{correlation} between
the two men's ranks, which is crucial for our concentration results.

In our proof, we also formalize what it means for all men other than two to
propose, with the notion of a ``partial matching state''.
Moreover, we give the term \emph{smooth} to those states in which the proof
sketch above goes through. Most crucially, in smooth matching states,
``most women have received a lot of proposals'',
so that the reasoning in points~2 and~3 are valid.
Additionally, to address certain technicalities (such as being able to
bound the magnitude of the expected number of proposals) we define smooth
matching states to not have too many proposals in total.

\subsection{Smooth matching states}

\begin{definition}\label{defPartialMatchingState}
  Given a set of men $L$,
  we define the \emph{partial matching state excluding $L$},
  denoted $\mu_{-L}$, as follows:
  Run DA with men in $M\setminus L$ proposing to $W$,
  and keep track of which proposals were made.
  More specifically, if $\mu$ is the (partial) matching
  resulting from running DA with a set of men $M\setminus L$
  and set of women $W$, and 
  $P = \{ (m_{i_{\ell}}, w_{j_\ell}) \}_{\ell}$
  is the set of all tuples $(m_i, w_j)$
  where $m_i$ proposed to $w_j$ during this process,
  then $\mu_{-L} = (\mu, P)$.

  In a random matching market, we consider this state as a random
  variable. In a tiered random matching market, to specify this random
  variable, it suffices to give a multiset of tiers which the men in $L$
  belong to.
  For a fixed $\mu_{-L}$,
  denote by $\Gamma_w$ the total sum of
  weights which woman $w$ received in $P$.
\end{definition}

Note that the state $\mu_{-L}$ keeps track of which proposals
have been made (in addition to which current matches are formed)
before the men in $L$ propose.

\begin{definition}\label{defSmoothMatchingState}
  We call a partial matching state $\mu_{-L}$
  \emph{smooth} if the following hold for some constants 
  $C_1, C_2, C_3 > 0$:
  \begin{enumerate}
    \item At most $C_1 n\ln n$ proposals were made to women overall.
    \item At most $n^{1-C_2}$ women have received fewer than $C_3 \ln n$ 
      proposals.
  \end{enumerate}
\end{definition}

The constants $C_1, C_2, C_3$ in the above depend on the tier structure,
and can simply be chosen such that the following proposition holds.
Our arguments will go through if smoothness holds with respect to
any $C_1, C_2, C_3$ which are held constant as $n\to \infty$. 
The proof is given in appendix~\ref{appendixReachingSmooth}.

\begin{restatable}{proposition}{RestateThrmSmoothWHP}
  \label{thrmSmoothWHP}
  Let $L = \{m_1, m_2\}$ be any pair of men.
  After running deferred acceptance,
  $\mu_{-L}$ is smooth with probability $1 - n^{-\Omega(1)}$.
\end{restatable}

Once we know that $\mu_{-L}$ is smooth, our two main tasks are to show that
men's ranks scale inverse-proportionally to their score,
and that the ranks of different men do not correlate too highly.
These   are the main technical novelties of the paper. The exact details are given in Appendix~\ref{sectionSmooth}.

\begin{restatable}{proposition}{RestateThrmSmoothRanksScale}
  \label{thrmSmoothedRanksProportional}
  Suppose $\mu_{-L}$ is smooth, and
  let $r_1$ and $r_2$ be the ranks of $m_1$ and $m_2$ after running DA with
  $m_1$ and $m_2$ starting from $\mu_{-L}$.
  We have
  \begin{equation*}
    \mathbb{E}_{L}[r_1]
     = \big(1 \pm O(1/\ln n)\big)\frac{\beta_2}{\beta_1} \mathbb{E}_{L}[r_2].
  \end{equation*}
  where we use $\Es{L}{}$ to denote taking an expectation over the random
  process of $m_1, m_2$ proposing in DA after starting from state $\mu_{-L}$.
\end{restatable}

\begin{restatable}{proposition}{RestateThrmSmoothRanksCovariance}
  \label{thrmSmoothCovar}
  Suppose $\mu_{-L}$ is smooth, and
  let $r_1$ and $r_2$ be the ranks of $m_1$ and $m_2$ after running DA with
  $m_1$ and $m_2$ starting from $\mu_{-L}$.
  Then we have $\Cov(r_i, r_j) = O(\ln^{3/2}n)$.
\end{restatable}

\subsection{Expected rank of the men}
\label{sectionRankOfMenConclusion}

In this subsection, we show that overall, expected rank scales
proportionally to fitness (in addition to under smooth matching states). This allows us to compute the expected rank of the men.
The proofs (deferred to appendix~\ref{appendixMenRanksProof})
follow by carefully keeping track of the (limited)
effect of non-smooth matching states on the expectation.

\begin{restatable}{proposition}{RestateThrmMenRanksProportional}
\label{thrmMenRanksProportional}
  Let $r_i$ and $r_j$ denote the rank of a man in tiers $i$ and $j$.
  Then we have
  \[ \E{ r_i } = \big(1 \pm O(1/\ln n)\big)
    \frac{\beta_j}{\beta_i}\E{r_j}.
  \]
\end{restatable}

\begin{restatable}{theorem}{RestateThrmMenRanksExpectation}
\label{thrmMenRanksExpectation}
  Let $\bm{\beta}^{-1}$ denote the vector $(1/\beta_i)_{i}$.
  For each tier $j$, the rank $r_j$ of men in tier $j$ has expectation
  \[ \E{r_j} 
    = \big(1 \pm O(1/\ln n)\big)
    \frac{\E{S}}{(n\deltaDotBeta) \beta_j }
    = \big(1 \pm O(1/\ln n)\big)
    \frac{\epsilonDotAlpha}{\alpha_{\min}}\cdot 
    \frac{1}{(\bm{\delta}\cdot\bm{\beta}^{-1})}\cdot
    \frac{\ln n}{\beta_j}.
  \]
\end{restatable}

Finally, we also use our results on the covariance of men's ranks to prove concentration. We defer the proof to appendix~\ref{appendixMenRanksProof}. At a high level, the proof follows simply because the weak correlation implied by~\ref{thrmSmoothCovar} means that the variance of the average of the ranks is lower-order (compared to its expectation), so Chebyshev's inequality can be used.

\begin{restatable}{theorem}{RestateThrmMenRanks}
  \label{thrmMenRanksCentralConcentration}
  For any tier $j$, let 
  $\overline R^{M}_j = (\delta_j n)^{-1} \sum_m r_m$ denote the average rank of
  men in tier $j$.
  Then, for any $\epsilon > 0$,
  \[ \overline R^{M}_j = (1 \pm \epsilon)
    \frac{\epsilonDotAlpha}{\alpha_{\min}}\cdot 
    \frac{1}{(\bm{\delta}\cdot\bm{\beta}^{-1})}\cdot
    \frac{\ln n}{\beta_j}
  \]
  with probability approaching $1$ as $n\to\infty$.
\end{restatable}

\section{Expected rank of the women and the distribution of match types}
\label{sectionWomenRanks}

\subsection{Expected rank of women}

We saw in section~\ref{sectionRankOfMenConclusion}
that men achieve ranks proportional to the inverse 
of their public scores. In this section, we turn to the women.

To study the rank the women achieve, we
need to reason about the number of proposals women
receive on average.
By theorem~\ref{thrmMenRanksCentralConcentration}, we expect that
for each tier $j$ of men, the $\delta_j n$ men make a total number of
proposals approximately
\[ \frac{\delta_j \beta_j^{-1}}{\bm{\delta}\cdot\bm{\beta^{-1}}}
  \cdot\frac{\bm{\alpha}\cdot\bm{\epsilon}}{\alpha_{\min}}n \ln n.
\]
Each of these proposals goes to a woman in tier $i$ with
probability $\pi_i = \alpha_i / (n\epsilonDotAlpha)$,
so we expect such a woman to receive approximately
$ (\delta_j \beta_j^{-1})/(\bm{\delta}\cdot\bm{\beta^{-1}})
  \cdot(\alpha_i/{\alpha_{\min}})\ln n
$
proposals from men in tier $j$.
Each of these men has public score $\beta_j$, so we expect $\Gamma_w$,
the total sum of public scores of men proposing to $w$,
to be roughly
\[ \Gamma_w \approx
  \sum_j \beta_j\frac{\delta_j\beta_j^{-1}}{\bm{\delta}\cdot\bm{\beta^{-1}}}
  \cdot\frac{\alpha_i}{\alpha_{\min}}\ln n
  = \frac{\alpha_i \ln n}
    {\alpha_{\min}(\bm{\delta}\cdot\bm{\beta^{-1}})}.
\]

It is not immediately clear how the above
value of $\Gamma_w$ should translate to the \emph{rank}
that $w$ gets. Unlike in the case where men are uniform,
we cannot simply divide $n$ by the number of proposals
which $w$ receives.

Indeed, suppose a woman $w$ receives exactly the total
sum of weight $\Gamma_w$ predicted above. What should her rank be?
This is essentially the following:
across all tiers of $\delta_j n$ men each,
how many do we expect to beat her best proposal so far?
The probability that $w$ ranks a man $m$ higher than
her match, when
viewed according to~\ref{thrmWomenDeferedDecisions},
is a function only of the weight $\beta(m)$ of
$m$ and the weight of proposals $\Gamma_w$ which $w$ received.
Specifically, this probability is 
$\beta(m) / (\beta(m) + \Gamma_w) \approx \beta_j / \Gamma_w$.
Summing this
across all the men, we get
\[
\E{r_w}
\approx \sum_m \frac{\beta(m)}{\beta(m) + \Gamma_w} 
\approx \frac{n \bm{\delta}\cdot\bm{\beta}}{\Gamma_w}
\approx (\bm{\delta}\cdot\bm{\beta})
  (\bm{\delta}\cdot\bm{\beta}^{-1}) 
  \frac{\alpha_{\min}}{\alpha_i}
  \cdot \frac{n}{\ln n}.
\]
Note that this ignores the fact that a woman will never rank $m$ higher than her match if that $m$ already proposed to
  her during DA. But since $w$ only likely receives $\ln n \ll n/\ln n$ proposals, the difference is not noticeable.

It turns out that, with a detailed probabilistic
analysis, the above proof sketch goes through.
The details are given in appendix~\ref{appendixWomenRanks}.
\begin{restatable}{theorem}{RestateThrmWomenRanks}
  \label{thrmWomenRanks}
  Let $\BarRW_i = (\epsilon_i n)^{-1} \sum_{w\in T_i} r_w$
  denote the average rank of women in tier $i$.
  For all $\epsilon>0$, we have
  \[ \BarRW_i = (1 \pm \epsilon) 
    (\deltaDotBeta)(\deltaDotBeta^{-1}) \frac{\alpha_{\min}}{\alpha_i}
    \frac{n}{\ln n}
  \]
  with probability approaching $1$ as $n\to\infty$.
\end{restatable}

\subsection{The distribution of match types}

Fix a woman $w$ in tier $i$.
We now study the probability that $w$ is matched to a man
from some tier $j$.
In the previous section, we argued that with high probability
$w$ receives approximately a total of 
\[ \frac{\delta_j \beta_j^{-1}}{\bm{\delta}\cdot\bm{\beta^{-1}}}
  \cdot\frac{\bm{\alpha}\cdot\bm{\epsilon}}{\alpha_{\min}}n \ln n
\]
proposals from men in tier $j$.
Thus, the contribution to $\Gamma_w$ (the total weight of proposals $w$
received) from men in tier $j$ is 
\[ \Gamma_{j\to w} \approx
  \frac{\delta_j}{\bm{\delta}\cdot\bm{\beta^{-1}}}
  \cdot\frac{\bm{\alpha}\cdot\bm{\epsilon}}{\alpha_{\min}}n \ln n
  \approx \delta_j \Gamma_w.
\]
Moreover, it turns out that, with high probability,
the above holds up to $(1\pm\epsilon)$
for all tiers $j$ simultaneously.
Regardless of the order in which $w$ saw proposals,
the probability that $w$'s favorite proposal
came from a man in tier $j$ is $\Gamma_{j\to w}/\Gamma_w$.
Thus, this probability is approximately $\delta_j$.
We formally implement this proof in appendix~\ref{appendixMatchType}.

\begin{restatable}{theorem}{RestateThrmMatchTypes}
  \label{thrmMatchTypes}
  Consider an arbitrary tier $i$ of women and $j$ of men.
  For all $\epsilon>0$, there is an $n$ large enough such
  that the probability that a woman in tier $i$ matches
  to a man in tier $j$ is $(1\pm \epsilon) \delta_j$.
\end{restatable}


\section{Computational Experiments on Expected Rank}
\label{sectionSimulations}


In this section, we provide computational experiments to back up the main
features of our theorems -- the estimates for the rank which agents on each
side achieve. First, we find that, as the theory suggests, men have a large
advantage in rank compared to women, with men getting ranks
of order $\ln n$ and women of order $n/\ln n$.
More interestingly,
these two sets of simulations together isolate and
investigate all of the major constant factors
present in our estimates.
We find that, overall, our estimates correspond 
to the empirical averages.

\subsection{ Women divided into tiers }
Figures \ref{fig:vary_women_tier_avg_men},
\ref{fig:vary_women_tier_avg_women_1}, and
\ref{fig:vary_women_tier_avg_women_2} showcase the expected rank in a
market where women are broken into two tiers,
while men have a constant public score.
In such a market, theorems~\ref{thrmMenRanksCentralConcentration}
and~\ref{thrmWomenRanks} predict that the expected rank of men,
and the expected rank of women in tier $i$, are respectively:
\[
  \frac{\epsilonDotAlpha}{\alpha_{\min}}\cdot 
  \ln n
  \qquad\text{ and }\qquad
  \frac{\alpha_{\min}}{\alpha_i}\cdot \frac{n}{\ln n}.
\]

In our experiment, there are $n=1000$ agents on each side.
The tiers of women have fraction
$\bm{\epsilon} = (\epsilon_1, 1-\epsilon_1)$, and weight
$\bm{\alpha} = (\alpha_1, 1)$, where tier $1$ is the ``top tier''
(i.e. $\alpha_1 > 1$ and $\alpha_{\min} = 1$).
Each plot has $\alpha_1$ ranging from 1 to 10 at each multiple of $0.25$,
and $\epsilon_1$ ranging from 0.025 to 0.975 at each multiple of $0.025$.
In each plot, we show the average rank in the result of DA, i.e. the
man-optimal stable matching, as this is the quantity studied in our
theorems.

Figure \ref{fig:vary_women_tier_avg_men} shows
men's average rank of partners, which in this market is approximately
the total number of proposals divided by $n$, because men are uniform.
Note that our prediction becomes increasingly bad as $\epsilon_1$
approaches $1$, even though for any fixed constant $\epsilon_1$,
we have convergence by theorem~\ref{thrmMenRanksCentralConcentration}.
This is natural because, as per remarks~\ref{remarkLogEspMinErrorBody}
and~\ref{remarkLogEspMinErrorAppendix}, our estimates break
down for fixed $n$ as $\epsilon_{\min}\to 0$.
Indeed, we find that the total number of proposals
is much less than our estimate in cases where $\epsilon_{\min}$
is small, as proposition~\ref{thrmCouponExpectationLowerBound}
simply proves that the true average is at least
our estimate minus $O(\ln(1/\epsilon_{\min})n)$.
This comment also applies to figures~\ref{fig:vary_women_tier_avg_women_1}
and~\ref{fig:vary_women_tier_avg_women_2}.
Accounting for cases with very small tiers (say, tiers which
grow sublinearly with $n$) is an intriguing future research direction.

\begin{figure}[ph]
Figures \ref{fig:vary_women_tier_avg_men},
\ref{fig:vary_women_tier_avg_women_1}, and
\ref{fig:vary_women_tier_avg_women_2} display expected rank in a
market with women broken into two tiers.
  \centering
  \includegraphics[scale=0.45]{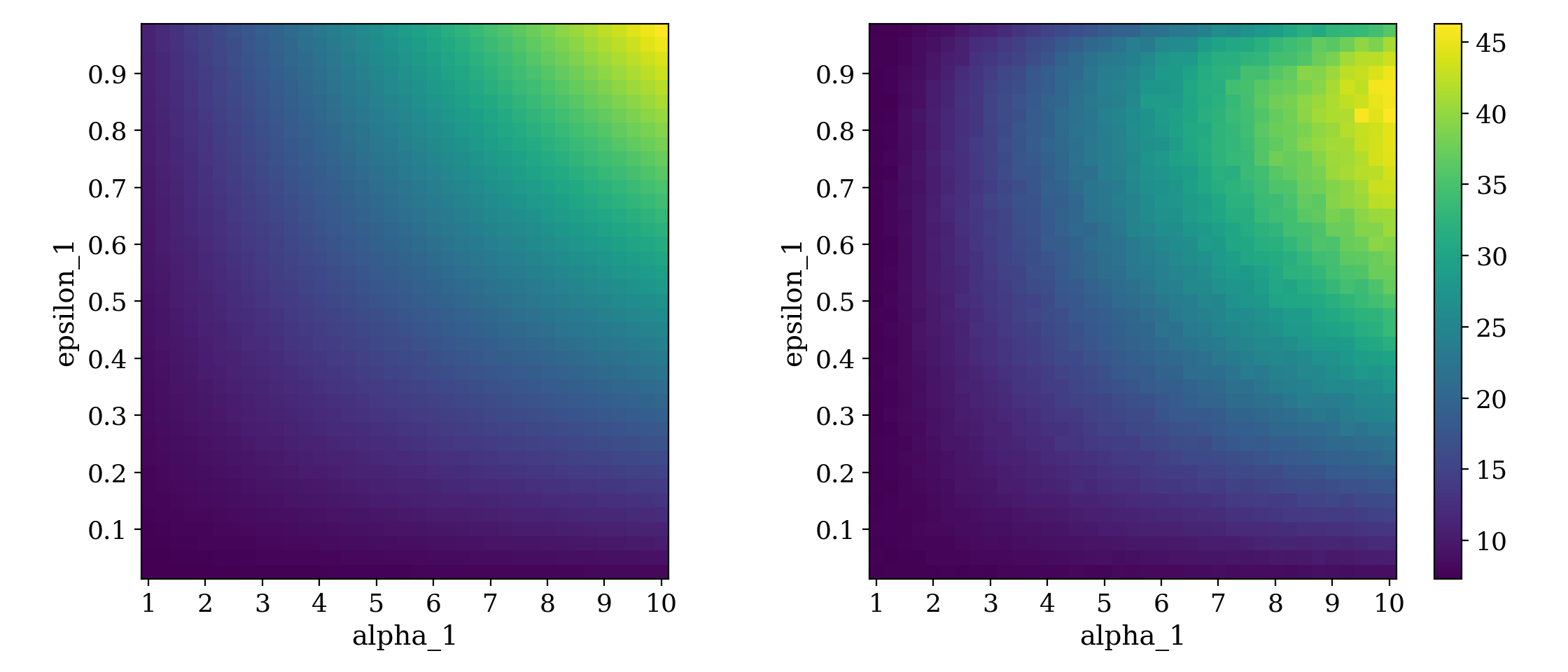}
  \caption{
    Men's average rank under DA.
    The left panel computes the prediction
    $(\epsilonDotAlpha/\alpha_{\min}) \ln n = (\epsilonDotAlpha) \ln n$,
    while the right panel is an average over 200 realizations.
  }\label{fig:vary_women_tier_avg_men}

  \includegraphics[scale=0.45]{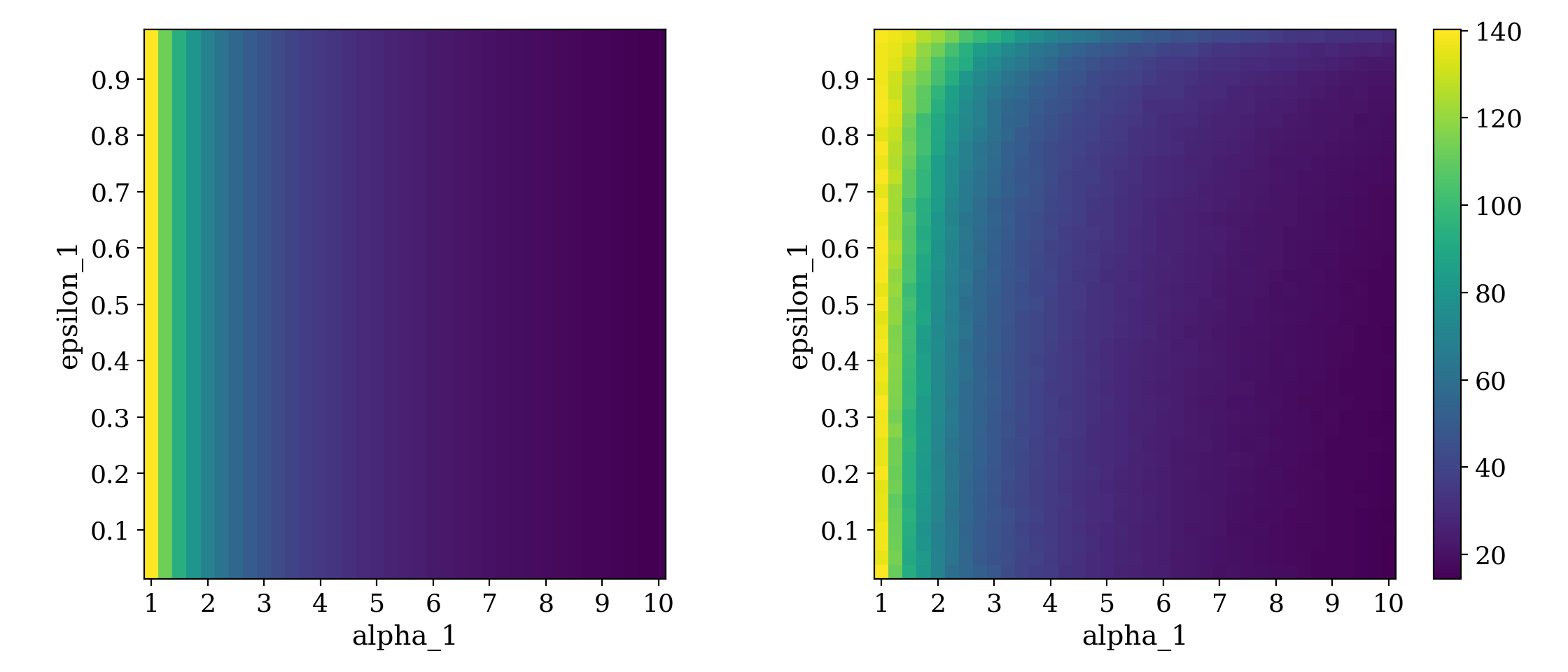}
  \caption{
    Top tier women's average rank under DA.
    The left panel computes the prediction
    $(\alpha_{\min}/\alpha_i) (n / \ln n)
    = (1/\alpha_1) (n / \ln n)$, while the
    right panel is an average over 200 realizations.
  }\label{fig:vary_women_tier_avg_women_1}

  \includegraphics[scale=0.45]{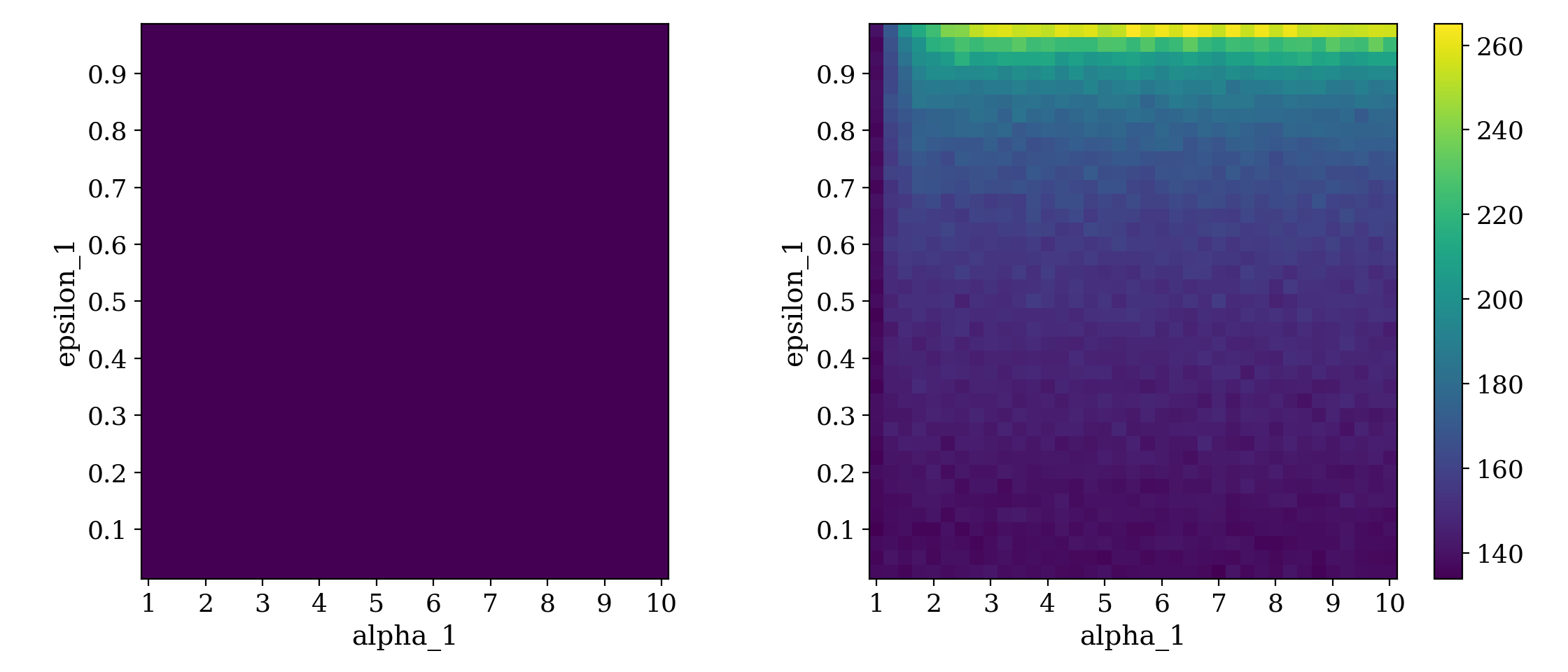}
  \caption{
    Bottom tier women's average rank under DA.
    Our prediction is the constant $n/\ln n$
    as $\bm{\alpha}, \bm{\epsilon}$ change,
    whereas the right panel is an average over 200 realizations.
  }\label{fig:vary_women_tier_avg_women_2}
\end{figure}


\subsection{ Men divided into tiers }

Figures~\ref{fig:vary_men_tier_avg_men_1},
\ref{fig:vary_men_tier_avg_men_2}, and~\ref{fig:vary_men_tier_avg_women}
showcase the expected rank in a
market where men are broken into two tiers,
while women have a constant public score.
In such a market, our prediction for the expected 
rank of men in tier $j$
and the expected rank of women are respectively:
\[
  \frac{1}{(\bm{\delta}\cdot\bm{\beta}^{-1})}\cdot
  \frac{\ln n}{\beta_j}
  \qquad\text{ and }\qquad
  (\deltaDotBeta)(\deltaDotBeta^{-1}) \cdot
  \frac{n}{\ln n}.
\]


\begin{figure}[ph]
Figures~\ref{fig:vary_men_tier_avg_men_1},
\ref{fig:vary_men_tier_avg_men_2}, and~\ref{fig:vary_men_tier_avg_women}
display expected rank in a
market with men broken into two tiers.
    \centering
    \includegraphics[scale=0.45]{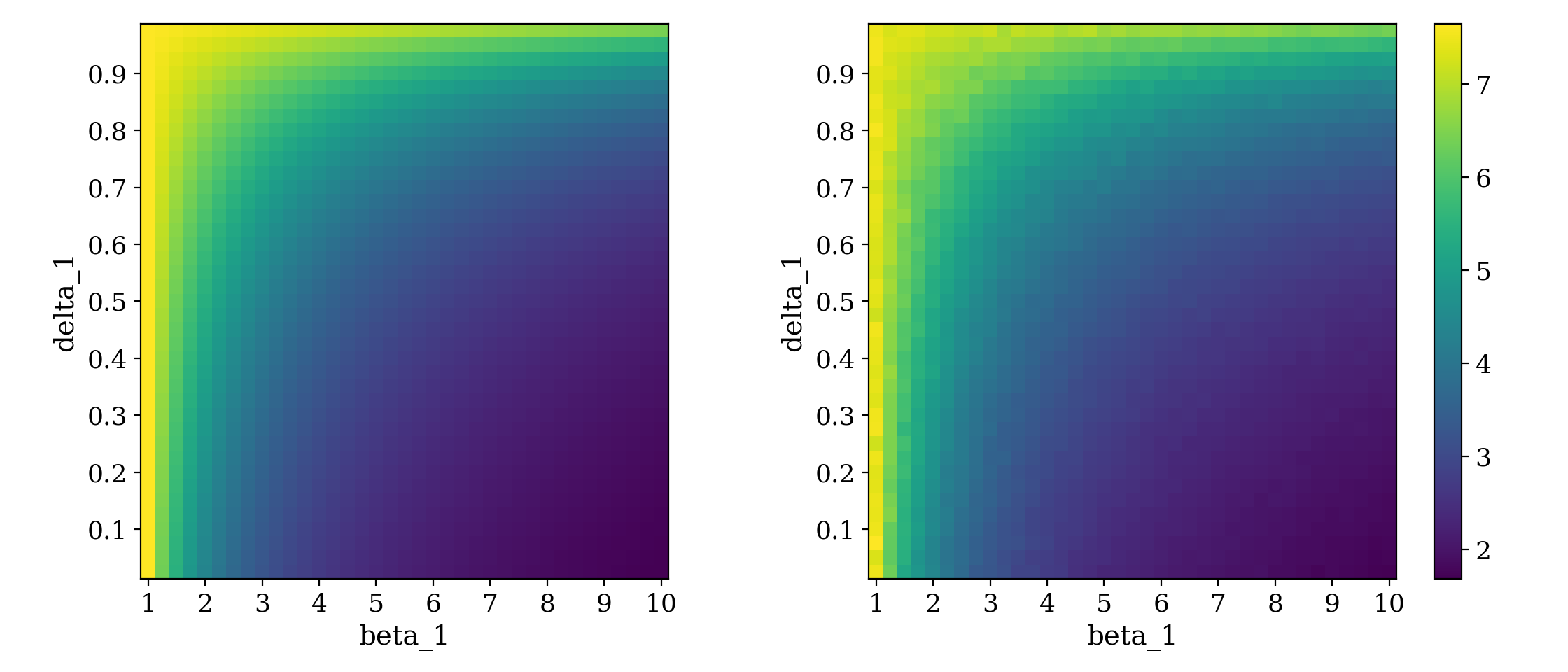}
    \caption{
        Top tier men's average ranks in DA.
        The left panel computes the prediction
        $\ln n / (\beta_1\bm{\delta}\cdot\bm{\beta}^{-1})$,
        while the right panel is an average over 200 realizations.
    }\label{fig:vary_men_tier_avg_men_1}

    \includegraphics[scale=0.45]{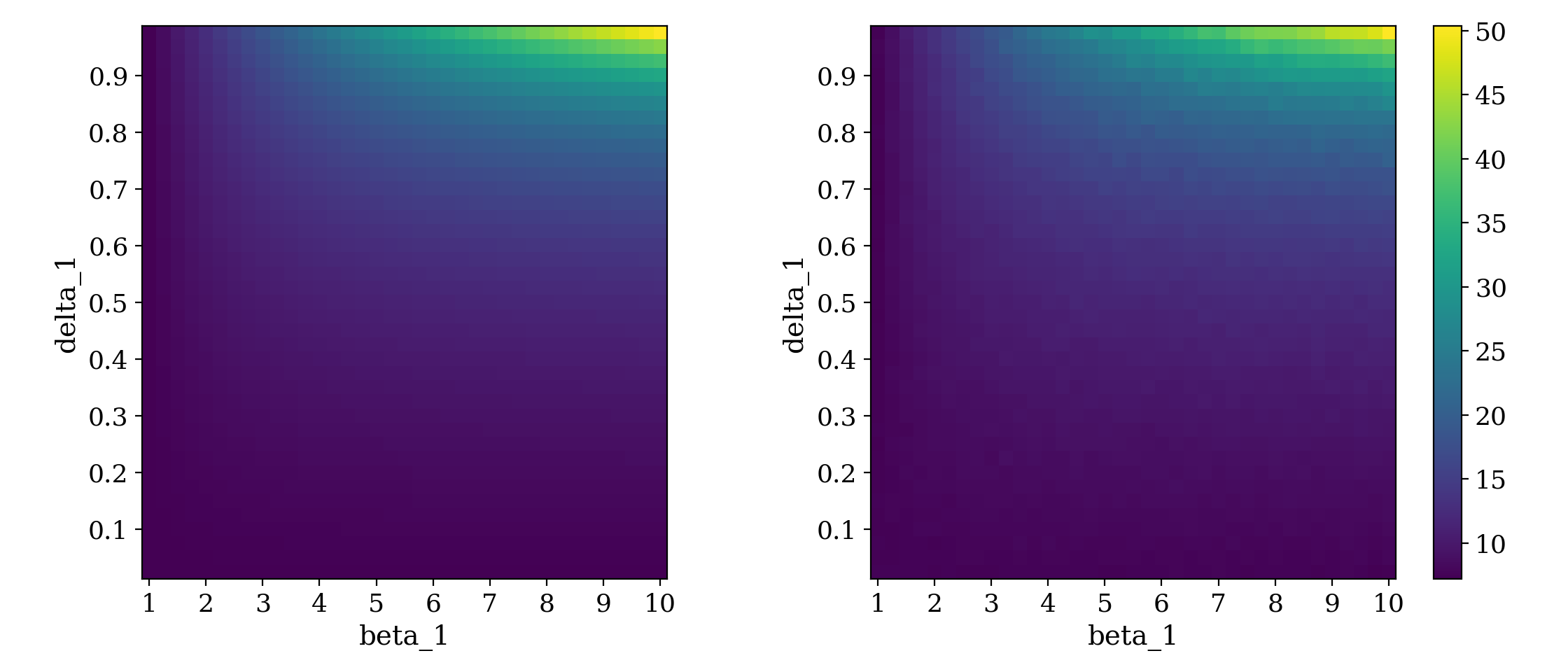}
    \caption{
        Bottom tier men's average ranks in DA.
        The left panel computes the predicted value
        $\ln n / (\beta_{\min}\bm{\delta}\cdot\bm{\beta}^{-1})
        = \ln n / (\bm{\delta}\cdot\bm{\beta}^{-1})$,
        while the right panel is an average over 200 realizations.
    }\label{fig:vary_men_tier_avg_men_2}

    \includegraphics[scale=0.45]{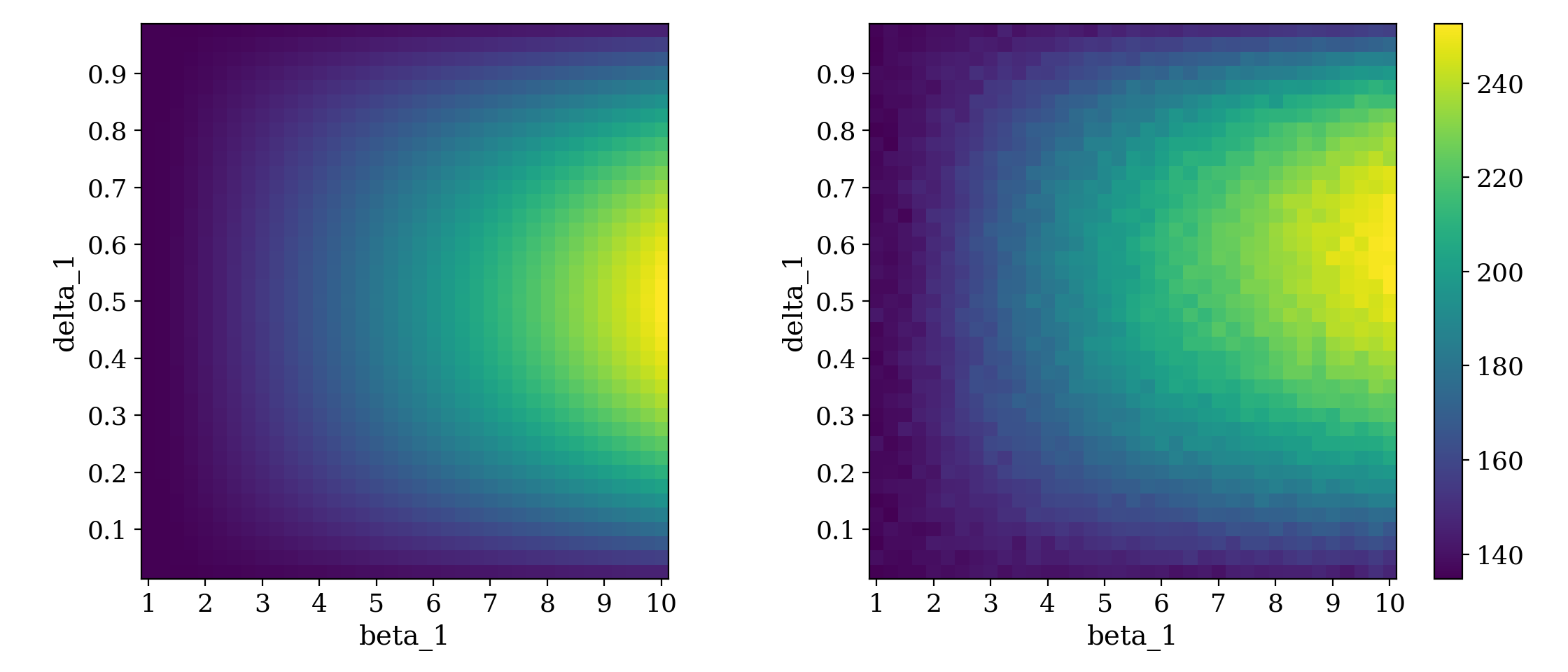}
    \caption{
        Women's average ranks in DA.
        The left panel computes the prediction
        $(\deltaDotBeta)(\deltaDotBeta^{-1})(n / \ln n)$,
        while the right panel is an average over 200 realizations.
    }\label{fig:vary_men_tier_avg_women}
\end{figure}

We again take $n=1000$ agents on each side.
The tiers of men have fraction
$\bm{\delta} = (\delta_1, 1-\delta_1)$
and weight $\bm{\beta} = (\beta_1, 1)$, with $\beta_1 > 1$.
Each plot takes $\beta_1$ ranging from 1 to 10 at each multiple of $0.25$, and $\delta_1$ ranging from 0.025 to 0.975 at each multiple of 0.025.
Because the women's side is balanced, the number of proposals in these
markets does not suffer from a great loss of accuracy in certain
parameter settings (as in the previous regime when $\epsilon_{\min}$
was small). 
However, lower order terms still make a visible difference,
especially in the rank achieved by the women.

\subsection{Distribution of matched pairs among tiers}
\label{sectionDistributionMatchExperiment}
As we have seen from proposition~\ref{thrmMenRanksProportional}, each individual man from tier $i$ has a more advantageous expected rank of partner than another man from tier $j$ whenever $\beta_i > \beta_j$. In our last experiment, we want to take a macro viewpoint and explore the distribution of matched pairs across tiers on both side in the man-optimal stable matching in a tiered market.

We demonstrate this effect by considering a sequence of balanced markets with two tiers with equal size on each side, i.e. $\bm{\delta}=\bm{\epsilon}=(0.5,0.5)$, with public scores $\bm{\beta}=(3,1)$ and $\bm{\alpha}=(5,1)$ for men and women, respectively. The market size $n$ grows from $2^4$ to $2^{18}$ at each integer power of 2. In such markets, the distribution of matched pair can be solely characterized by the fraction of men in tier 1 who are matched to women in tier 1, denoted by $m_{11}$. For each market configuration in the sequence, we simulate 1,000 realizations of man-proposing deferred acceptance, and recorded the values of $m_{11}$ for each realization. The result is shown in figure~\ref{fig:distribution_of_pairs}.

\begin{figure}[h]
    \centering
    \includegraphics[scale=0.6]{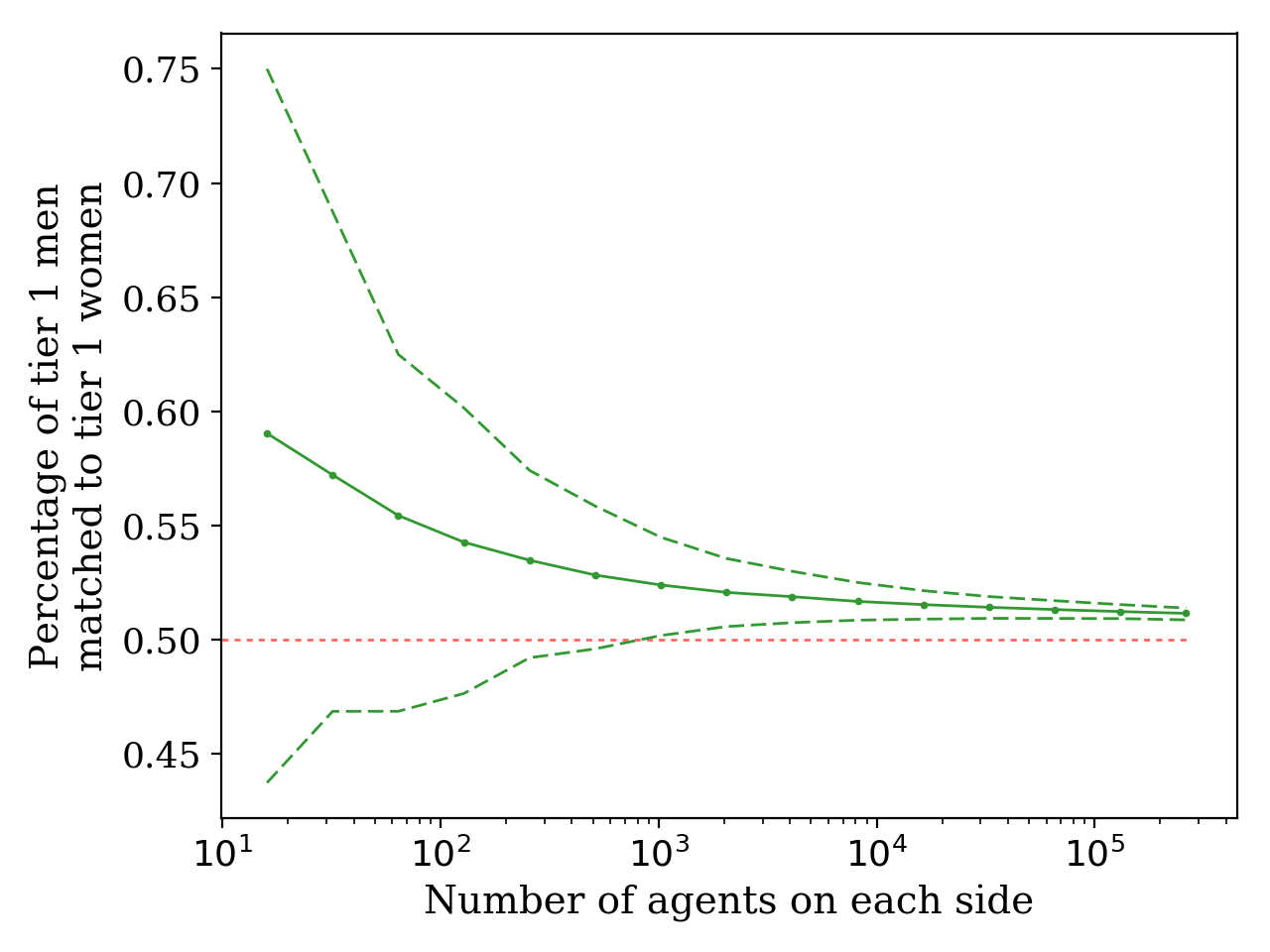}
    
    \caption{Percentage of tier 1 men matched to tier 1 women under the man-optimal outcome in markets with parameters $\bm{\delta}=\bm{\epsilon}=(0.5,0.5)$, $\bm{\beta}=(3,1)$, and $\bm{\alpha}=(5,1)$. The two dashed curves indicate 3 and 97 percentile, respectively.}\label{fig:distribution_of_pairs}
\end{figure}

The simulation suggests that the distribution of matched pairs gets closer to uniformity as the market size increases, with a slight skew benefiting the better tier. For example, for a market with 1,000 men and women on each side with two tiers of equal size and $\bm{\beta}$ and $\bm{\alpha}$ specified above, $52.4\pm 1.1\%$ of the men in tier 1 are matched to women in tier 1. That is, top-tier men are only slightly more likely than bottom tier men to match to top tier women, even though at a micro level each man in the first tier on average does three times better than those in the second tier. In the future, it may also be of interest to examine how the tier structure determines the deviation from uniformity.

On the other hand, we want to stress that the approximately uniform distribution of matched pairs among tiers relies heavily on our assumption of bounded public scores. The result may cease to hold if we allow the gap in scores to grow with the market size. Figure~\ref{fig:distribution_of_pairs_non_uniform} shows the deviation from the uniform distribution when the gap in scores of the top and bottom tier men grows in polynomial order, namely $\sqrt{n}/2$.
We hypothesize that in this case the fraction of tier 1 men matched to tier 1 women converges to the solution to the equation $x^5 + x=1$, approximately $0.755$. Note that this would be identical to the match type distribution limit if the top tier men were deterministically preferred over the lower tier men.
It is an interesting direction to study the matching dynamics when scores grow with market size (for example, in poly-log or polynomial order). 

\begin{figure}[h]
    \centering
    \includegraphics[scale=0.6]{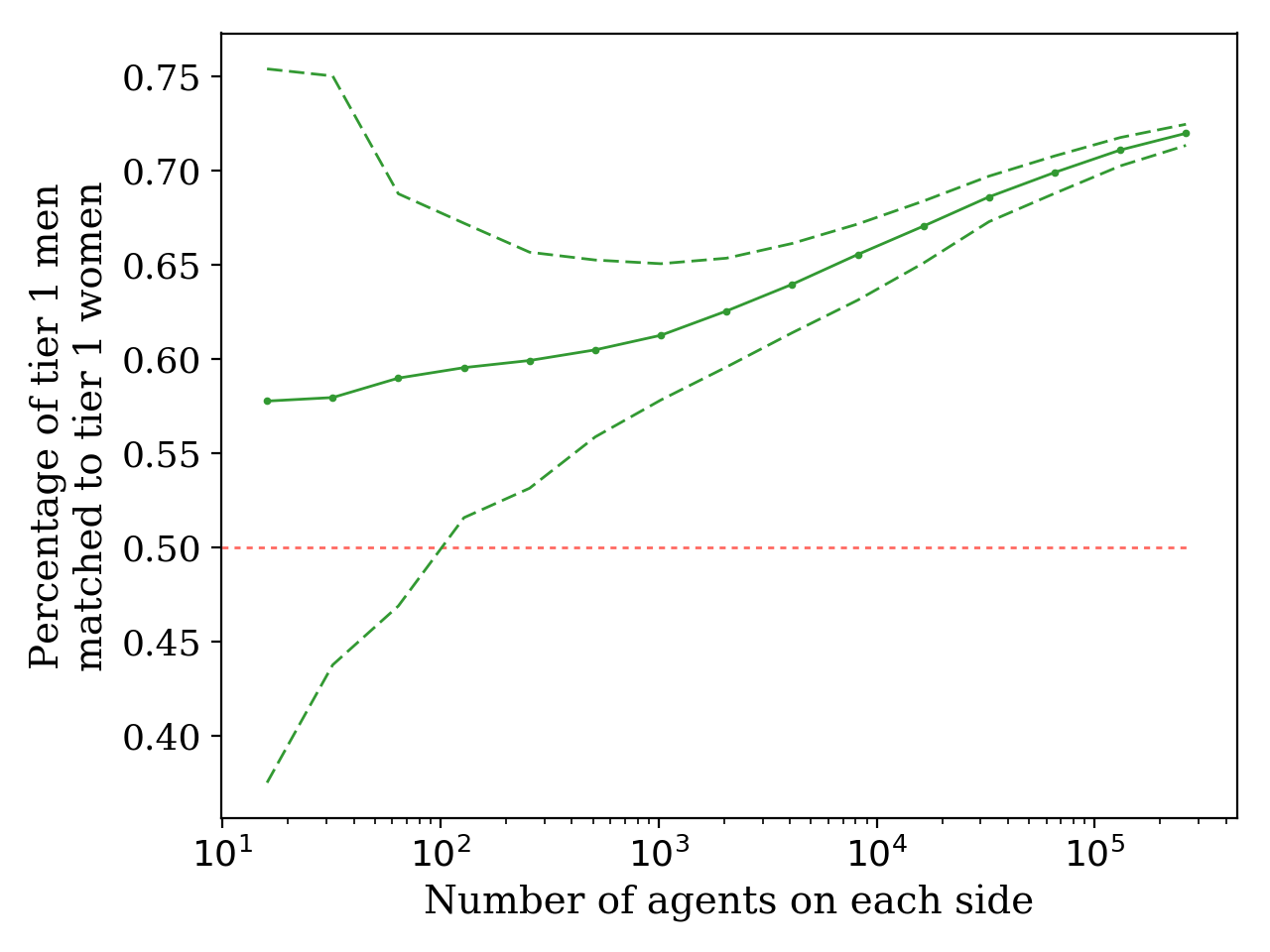}
    
    \caption{Non-uniform distribution of matched pairs among tiers when scores may grow with market size. The fraction of tier 1 men matched to tier 1 women under the man-optimal outcome deviates from $0.5$ in markets with parameters $\bm{\delta}=\bm{\epsilon}=(0.5,0.5)$, $\bm{\beta}=(\sqrt{n}/2,1)$, and $\bm{\alpha}=(5,1)$. The solid curve marks the average across 200 runs, and the two dashed curves indicate 3 and 97 percentile, respectively.}\label{fig:distribution_of_pairs_non_uniform}
\end{figure}

\section{Summary}
\label{sec:summary}

The model and findings in this paper contribute to the understanding of random stable matching markets. Indeed, the results quantify the effect of competition that arises from heterogeneous quality in agents, specifically, when the agents fall into different constant-factor tiers of quality. Novel technical tools are developed in order to reason about the proposal dynamics of deferred acceptance. 

Relaxing some of the modeling assumptions raises interesting questions that cannot be trivially answered. This includes having non-constant (size, or public score) tiers,
personalized private scores which give agents different distributions of preferences,
and imbalance in the number of agents on each side of the market. Moreover, it is natural to ask when one should expect the matching to be sorted, i.e., higher tiers will be more likely to match with higher tiers (e.g., \cite{hitsch2010matching} demonstrates the presence of sorting in dating markets).



\bibliography{matching}

\begin{thebibliography}{ABKS17}

\bibitem[ABH14]{ashlagi2014stability}
Itai Ashlagi, Mark Braverman, and Avinatan Hassidim.
\newblock Stability in large matching markets with complementarities.
\newblock {\em Operations Research}, 62(4):713--732, 2014.

\bibitem[ABKS17]{AshlagiBKS17}
Itai Ashlagi, Mark Braverman, Yash Kanoria, and Peng Shi.
\newblock Communication requirements and informative signaling in matching
  markets.
\newblock In {\em Proceedings of the 2017 ACM Conference on Economics and
  Computation}, EC '17, page 263, New York, NY, USA, 2017. Association for
  Computing Machinery.

\bibitem[AKL17]{ashlagi2017unbalanced}
Itai Ashlagi, Yash Kanoria, and Jacob~D Leshno.
\newblock Unbalanced random matching markets: The stark effect of competition.
\newblock {\em Journal of Political Economy}, 125(1):69--98, 2017.

\bibitem[Bra63]{brayton1963asymptotic}
Robert~King Brayton.
\newblock On the asymptotic behavior of the number of trials necessary to
  complete a set with random selection.
\newblock {\em Journal of Mathematical Analysis and Applications}, 7(1):31--61,
  1963.

\bibitem[BST17]{BeyhaghiST17}
Hedyeh Beyhaghi, Daniela Sab{\'{a}}n, and {\'{E}}va Tardos.
\newblock Effect of selfish choices in deferred acceptance with short lists.
\newblock {\em CoRR}, abs/1701.00849, 2017.

\bibitem[CS14]{coles2014optimal}
Peter Coles and Ran Shorrer.
\newblock Optimal truncation in matching markets.
\newblock {\em Games and Economic Behavior}, 87:591--615, 2014.

\bibitem[CT19]{cai2019short}
Linda Cai and Clayton Thomas.
\newblock The short-side advantage in random matching markets.
\newblock {\em arXiv preprint arXiv:1910.04406}, 2019.

\bibitem[DP12]{doumas2012couponsRevisited}
Aristides Doumas and Vassilis Papanicolaou.
\newblock The coupon collector's problem revisited: Asymptotics of the
  variance.
\newblock {\em Advances in Applied Probability - ADVAN APPL PROBAB}, 44, 03
  2012.

\bibitem[GMM19]{gimbert2019popularity}
Hugo Gimbert, Claire Mathieu, and Simon Mauras.
\newblock Two-sided matching markets with correlated random preferences have
  few stable pairs.
\newblock {\em arXiv preprint arXiv:1904.03890}, 2019.

\bibitem[Gon14]{Gonczarowski14}
Yannai~A. Gonczarowski.
\newblock Manipulation of stable matchings using minimal blacklists.
\newblock In {\em Proceedings of the Fifteenth ACM Conference on Economics and
  Computation}, EC '14, page 449, New York, NY, USA, 2014. Association for
  Computing Machinery.

\bibitem[GS62]{gale1962college}
David Gale and Lloyd~S Shapley.
\newblock College admissions and the stability of marriage.
\newblock {\em The American Mathematical Monthly}, 69(1):9--15, 1962.

\bibitem[HHA10]{hitsch2010matching}
Gunter~J Hitsch, Ali Horta{\c{c}}su, and Dan Ariely.
\newblock Matching and sorting in online dating.
\newblock {\em American Economic Review}, 100(1):130--63, 2010.

\bibitem[IM15]{immorlica2015incentives}
Nicole Immorlica and Mohammad Mahdian.
\newblock Incentives in large random two-sided markets.
\newblock {\em ACM Transactions on Economics and Computation (TEAC)},
  3(3):1--25, 2015.

\bibitem[KMP90]{knuth1990stable}
Donald~E Knuth, Rajeev Motwani, and Boris Pittel.
\newblock Stable husbands.
\newblock {\em Random Structures \& Algorithms}, 1(1):1--14, 1990.

\bibitem[KMQ20]{kanoria2020random}
Yash Kanoria, Seungki Min, and Pengyu Qian.
\newblock Which random matching markets exhibit a stark effect of competition?
\newblock {\em arXiv preprint arXiv:2006.14653}, 2020.

\bibitem[Knu76]{Knuth76}
D.~E. Knuth.
\newblock {\em Mariages Stable}.
\newblock Université de Montréal Press, 1976.
\newblock Translated as “Stable Marriage and Its Relation to Other
  Combinatorial Problems, CRM Proceedings and Lecture Notes, 1997.

\bibitem[KP09]{kojima2009incentives}
Fuhito Kojima and Parag~A Pathak.
\newblock Incentives and stability in large two-sided matching markets.
\newblock {\em American Economic Review}, 99(3):608--27, 2009.

\bibitem[Lee16]{Lee16}
SangMok Lee.
\newblock {Incentive Compatibility of Large Centralized Matching Markets}.
\newblock {\em The Review of Economic Studies}, 84(1):444--463, 09 2016.

\bibitem[MW70]{mcvitie1970stable}
David~G McVitie and Leslie~B Wilson.
\newblock Stable marriage assignment for unequal sets.
\newblock {\em BIT Numerical Mathematics}, 10(3):295--309, 1970.

\bibitem[Pit89]{pittel1989average}
Boris Pittel.
\newblock The average number of stable matchings.
\newblock {\em SIAM Journal on Discrete Mathematics}, 2(4):530--549, 1989.

\bibitem[Pit92]{pittel1992likely}
Boris Pittel.
\newblock On likely solutions of a stable marriage problem.
\newblock {\em The Annals of Applied Probability}, pages 358--401, 1992.

\bibitem[Pit19]{pittel2019likely}
Boris Pittel.
\newblock On likely solutions of the stable matching problem with unequal
  numbers of men and women.
\newblock {\em Mathematics of Operations Research}, 44(1):122--146, 2019.

\bibitem[Ros06]{ross2006probability}
Sheldon~M. Ross.
\newblock {\em Introduction to Probability Models, Tenth Edition}.
\newblock Academic Press, Inc., USA, 2006.

\bibitem[Wil72]{wilson1972analysis}
LB~Wilson.
\newblock An analysis of the stable marriage assignment algorithm.
\newblock {\em BIT Numerical Mathematics}, 12(4):569--575, 1972.

\end{thebibliography}

\appendix

\section{Expectation of the coupon collector}
\label{appendixCouponExpectation}

Recall that $T$ denotes the number of draws to a coupon collector with
distribution $\W$, as per definition~\ref{defCouponCollector}.
We now give a simple representation of $T$, generalizing the common
representation used for the case where $\W$ is the uniform
distribution, and use it to prove a simple upper bound on the expectation
of $T$.

\begin{proposition}\label{thrmDAIsSumGeo}
  We can write $T = \sum_{i=1}^n \tau^{(i)}$, where
  each $\tau^{(i)}$ is distributed according to a geometric
  distribution $\Geo(p_i)$ for some value of $p_i$
  depending on the randomly determined order in which distinct 
  coupons are collected.
  The value of $p_i$ can be bounded by
  \[ (n - i+1)\pim 
    \le p_i \le (n - i + 1)\pima.
  \]
  Moreover, conditioned on the order in which the
  distinct coupons are collected, each draw to $\tau^{(i)}$
  is independent.
\end{proposition}
\begin{proof}
  Let $\tau^{(i)}$ denote the number of draws between the appearance of the
  $(i-1)$th and $i$th distinct coupon.
  Because each draw is independent, the probability that each draw goes to
  a new coupon is exactly the sum of the weights of the $n-i+1$
  unseen coupons, divided by $n\epsilonDotAlpha$.
  Thus, each $\tau^{(i)}$ is a geometrically distributed,
  and furthermore $\tau^{(i)}$ are independent if we condition
  on the realized values of $p_i$.
  Moreover, assuming that every uncollected coupon is in the bottom tier,
  we can bound $p_i$ by
  \[ p_i \ge \frac{(n-i+1)\alpha_{\min}}{n\epsilonDotAlpha}
    = (n-i+1)\pim,
  \]
  and on the other hand, assuming every uncollected
  coupon is in the top tier,
  \[ p_i \le \frac{(n-i+1)\alpha_{\max}}{n\epsilonDotAlpha}
    = (n-i+1)\pima.
  \]


\end{proof}

\begin{proposition}\label{thrmSimpleCouponUpperBound}
  We have 
  \[ \E{T} \le \pim^{-1} H_n
    = \big(1 + \widetilde O(1/n)\big)
    \frac{\epsilonDotAlpha}{\alpha_{\min}} n \ln n
  \]
  where $H_n = \ln n + O(1)$ is the $n$th harmonic number.
\end{proposition}
\begin{proof}
  Write $T = \sum_{i=1}^n \tau^{(i)}$, where $\tau^{(i)} \sim \Geo(p_i)$
  is as above, so that $p_i \ge (n - i + 1)\pim$.
  Recall that, for any $p \ge q$, the distribution $\Geo(p)$ is
  statistically dominated by $\Geo(q)$.
  Letting $T_{\mathrm{ub}} = \sum_{i=1}^n \tau_{\mathrm{ub}}^{(i)}$,
  where $\tau_{\mathrm{ub}}^{(i)}$ are independent draws from $\Geo((n-i+1)\pim)$.
  For any fixed order $C$ of collection of distinct coupons,
  $\tau^{(i)}$ are independent, and thus $(T | C)\preceq T_{\mathrm{ub}}$.
  So we have $T \preceq T_{\mathrm{ub}}$ overall.
  Thus,
  \[ \E{T} \le \E{T_{\mathrm{ub}}}
    = \sum_{i=1}^n \frac{1}{(n - i+1)\pim}
    = \pim^{-1} H_n.
  \]
\end{proof}

We go on to show that this simple upper bound is tight up to lower order terms.
The key intuition is that, for most runs of the coupon collector,
we spend most of the time waiting for coupons in
the bottom tier to be collected.
This intuition is made formal using a standard probability theory
technique, sometimes called ``Poissonization''.
We include the proof, as in~\cite{ross2006probability},
in appendix~\ref{appendixProbability} for completeness.

\begin{restatable}{lemma}{RestateThrmCouponPoissonization} 
  \label{thrmCouponExpectation}
  Let $T_{\D}$ be a coupon collector random variable with probabilities
  $\D = (p_i)_{i\in [n]}$.
  For $i=1,\ldots,n$, let $X_i$ be an independent draw from $\Exp(p_i)$,
  and let $X = \max_{i\in[n]} X_i$.
  Then we have $\E{T_\D} = \E{X}$.
\end{restatable}

\begin{proposition}\label{thrmCouponExpectationLowerBound}
  For $T = T_{\W}$, we have
  \[ \E{T} \ge
    \left(1 - \frac{\ln(1/\epsilon_{\min})}{\ln n}\right)
    \frac{\epsilonDotAlpha}{\alpha_{\min}} n \ln n
    = \pim^{-1}H_n - O(n)
  .\]
\end{proposition}

\begin{proof}
  Let $X$ be the maximum of $n$ draws from $\Exp(\pi_i)$, with
  one exponential corresponding to each woman in $\W$,
  so that $\E{T} = \E{X}$.
  Let $\epsilon_{\min}$ be the fraction of the $n$ women which
  are in the bottom tier,
  and let $X_{1}$ be the maximum of $\epsilon_{\min}n$
  independent draws from $\Exp(\pi_{\min})$.
  It's clear that $X \succeq X_1$, and thus $\E{X} \ge \E{X_1}$.

  Using another standard fact from probability theory,
  namely proposition~\ref{thrmMaximumExponential},
  we can calculate the expectation of $X_1$:
  \[
    \E{X_1} = \frac{H_{\epsilon_{\min}n}}{\pim}
    \ge \frac{\ln(\epsilon_{\min}n)}{\pim}
    = \frac{\ln(n)}{\pim} + \frac{\ln(\epsilon_{\min})}{\pim}
    \ge \frac{\ln(n)}{\pim} - O(n)
  \]
\end{proof}
\begin{remark}\label{remarkLogEspMinErrorAppendix}
  For all fixed tier structures, the above estimate is equal to
  the upper bound of $\pim^{-1}H_n$ up to lower order terms.
  However, if the lowest tier is very small ($\epsilon_{\min}$ close to
  zero) then the estimate is quite crude for reasonable values of $n$.
\end{remark}

Our main result on the expectation of the coupon collector is an immediate
consequence of~\ref{thrmSimpleCouponUpperBound} 
and~\ref{thrmCouponExpectationLowerBound}.

\RestateThrmCouponExpectation*

\section{Concentration of the coupon collector}
\label{appendixCouponConcentration}

For many of our results, we will need fairly tight
bounds on the tails of the coupon collector random variable $T$.
We proceed to give these bounds here, using the standard techniques of
``general Chernoff bounds''.

For the upper tail of $T$, we are able to achieve polynomial concentration
immediately above our estimate for the expectation of $T$.
\begin{proposition} \label{thrmCouponUpperTail}
  \label{thrmChernoffCouponUpperTail}
  For all $C>0$, we have 
  \[ \P{T \ge (1 + C)\frac{\epsilonDotAlpha}{\alpha_{\min}} n\ln n }
    \le O(n^{-C/2})
  \]
\end{proposition}
\begin{proof}
  Let $T_{\mathrm{ub}} = \sum_{j=1}^n T_j$,
  where $T_j$ are independent draws from $\Geo(j\pim)$
  (for $\pim \le 1/n$).
  By~\ref{thrmDAIsSumGeo}, we know
  $T$ is statistically dominated by $T_{\mathrm{ub}}$ (where we change the
  index of summation for convenience).
  Thus, it suffices to prove that
  $\P{T_{\mathrm{ub}} \ge (1 + C)H_n \pim^{-1} }\le 2n^{-C/2}$.

  We start by separately handling the largest component of
  $T_{\mathrm{ub}}$, namely $T_1$:
  \begin{equation*}
    \P{ T_1 \ge (C/2)\pim^{-1}H_n }
    = (1 - \pim)^{(C/2)\pim^{-1}H_n}
    \le \exp(-(C/2)H_n)
    \le n^{-C/2}
  \end{equation*}

  Now let $T_{-1} = \sum_{j=2}^n T_j$.
  For some $a>0$ which we will specify soon,
  \begin{align*}
    \P{T_{-1} \ge (1+C/2)H_n\pim^{-1}}
    & = \P{ \exp(a T_{-1}) \ge \exp(a(1+C/2)\pim^{-1}H_n) } \\
    & \le \frac{ \E{ \exp(a T_{-1}) } }{ \exp(a(1+C/2)\pim^{-1}H_n) }
    \tag{*}
  \end{align*}
  Due to the independence of the $T_j$s,
  the numerator above equals $\prod_{i=2}^n \E{\exp(a T_j)}$.
  \begin{align*}
    \forall j > 2:\qquad \E{ \exp( a T_j ) } 
    & = \sum_{i=1}^\infty (1-j\pim)^{i-1}(j\pim) e^{ai} \\
    & = \frac{j\pim e^a}{1 - (1 - j\pim)e^a}
    = \frac{j\pim}{j\pim - (1 - e^{-a})}
  \end{align*}
  Note that the above expectation only exists for $j=2$ (the largest $T_j$)
  when $(1 - 2\pim)e^a < 1$.
  We can now see a convenient value for $a$, namely
  \begin{align*}
    a = -\log\left({1 - \pim}\right)
    \qquad & \Longrightarrow\qquad
    \E{ \exp( a T_j ) } 
    = \frac{j\pim}{j\pim - \pim}
    = \frac{j}{j-1} \\
    & \Longrightarrow\qquad
    \E{ \exp( a T_{-1} ) }
    = \prod_{i=2}^n \frac{i}{i-1}
    = n
  \end{align*}

  A short analytic exercise proves that
  $a = - \log\left({1 - \pim}\right) \ge \pim$, so
  the denominator of (*) can be bounded as:
  \[
    \exp(a(1+C/2)\pim^{-1}H_n) \ge 
    \exp((1+C/2)H_n) \ge 
    n^{1+C/2}
  \]
  and finally
  \[
    \P{T_{-1} \ge (C/2)H_n\pim^{-1}}
    \le n^{- C/2}
  \]

  Now, taking a union bound over the above two probabilities gets the
  desired result.
\end{proof}
\begin{remark}
  The proof above hints that it is likely not possible to qualitatively
  improve the above concentration much.
  Indeed, the largest geometric wait time $T_1$
  is larger than $C\pim^{-1}H_n$ with probability very close to $n^{-C}$,
  so this single wait time can push that entire $T_{\mathrm{ub}}$ above its
  expectation with only polynomially small probability.
  We expect a similar phenomenon with $T$ itself.
\end{remark}

We also need to reason about the lower tail of the coupon collector process.
Our result for the lower tail is significantly coarser, because the
simple lower bounds we can put on the geometric waiting times no longer
closely align with the main ``driving factor'' of the coupon collector
(the wait time to collect the last tier).
Nonetheless, these bounds suffice for our later purposes.
\begin{proposition}\label{thrmCouponLowerTail}
  For any $0 < c < 1$ and any $\eta > 0$, we have
  \[ \P{T \le c\frac{\epsilonDotAlpha}{\alpha_{\max}} n\ln n}
    \le O(n^{c - 1 - \eta})
  \]
\end{proposition}
\begin{proof}
  By~\ref{thrmDAIsSumGeo},
  we can write $T = \sum_{i=1}^n \tau^{(i)}$,
  where $\tau^{(i)}$ is a geometric
  random variable with parameter $p_i$, and $p_i \le (n - i+1) \pi_{\max}$.
  Note that for $n - i + 1 > \pi_{\max}^{-1}$,
  this upper bound on $p_i$ is greater than $1$.
  We simply ignore these terms to get a lower bound:
  $T \succeq T_{\mathrm{lb}}$ for 
  \[ T_{\mathrm{lb}} := \sum_{j=1}^{1/\pi_{\max}} T_j,
    \qquad \text{where $T_j \sim \Geo(j\pi_{\max})$ are independent}
  \]
  (where we change the index of summation for convenience).
  We now apply a Chernoff bound to get a lower tail
  bound on $T_{\mathrm{lb}}$.

  We have
  \begin{align*}
    \P{T_{\mathrm{lb}} \le c\pima^{-1} H_{\pima^{-1}}}
    & = \P{e^{-aT_{\mathrm{lb}}} \ge \exp({-ac\pima^{-1} H_{\pima^{-1}}})} \\
    & \le  \exp({ac\pima^{-1} H_{\pima^{-1}}}) \E{e^{-aT_{\mathrm{lb}}} }
    \tag{*}
  \end{align*}

  A calculation reveals
  \[
    \E{ e^{-a T_j} } = \frac{j\pima}{j\pima + e^{a} - 1}
    \le \frac{j\pima}{j\pima + a}
  \]
  since $e^a - 1 \ge a$. Now set $a = \pima$, and 
  the expectation in (*) can be bounded as
  \[
    \E{e^{-aT_{\mathrm{lb}}} } \le \prod_{j=1}^{\pima^{-1}} \frac{j}{j + 1}
    \le \pima \tag{\dag}
  \]

  In total,
  \begin{align*}
    \P{T_{\mathrm{lb}}
    \le c\pima^{-1} H_{\pima^{-1}}}
    & \le \exp(c H_{\pima^{-1}}) \pima \\
    & = O( \pima^{1-c} ) = O( n^{c-1} )
  \end{align*}
  Noting that 
  \[ \pima^{-1} H_{\pima^{-1} } 
    = \frac{\epsilonDotAlpha}{\alpha_{\max}}n 
      \left(\ln n + \log\frac{\epsilonDotAlpha}{\alpha_{\max}} + O(1)\right)
    = (1 - O(1/\ln n))\frac{\epsilonDotAlpha}{\alpha_{\max}}n\log n
  \]
  yields the desired result for any $\eta>0$.
\end{proof}


We can also modify the above proof to show the following,
which is needed at one point for a technical reason:

\begin{corollary} \label{thrmChernoffCouponLastTwoMenLowerTail}
  For any $\pi_{\max} = \Theta(1/n)$ and constant $k$,
  let $T_{\mathrm{lb}}' = \sum_{j=k}^{1/\pi_{\max}} T_j$,
  where $T_j$ are independent draws from $\Geo(j\pi_{\max})$.
  Then for all $0 < c < 1$ and $\eta > 0$, we have
  \[ \P{T_{\mathrm{lb}}' \le c\frac{\epsilonDotAlpha}{\alpha_{\max}} n\ln n}
    \le O(n^{c-1-\eta}). \]
\end{corollary}
\begin{proof}
  The only change we need to make is that the index of the product in
  ($\dag$) starts at $k$, instead of $1$.
  This increases the bound for $\E{e^{-aT_{\mathrm{lb}}'}}$ by a factor of
  $k=O(1)$, the thus the final probability by a factor of $k = O(1)$.
\end{proof}

\subsection{Concentration around the mean}

In this subsection, we upper bound the variance of the coupon collector
process in order to get (fairly weak) concentration results around the
mean.
This is actually not needed for any of our other results, because
once we start reasoning about DA (without re-proposals), we
bound the variance in the total number of proposals in a different way.
We include this result out of any possible interest.

\begin{proposition}
  For any $\epsilon > 0$,
  \[ T = (1 \pm \epsilon) \frac{\epsilonDotAlpha}{\alpha_{\min}} n \ln n
  \]
  with probability approaching $1$.
\end{proposition}
\begin{proof}
  It will suffice to calculate a crude upper bound
  for the variance of $T$ and apply Chebyshev's inequality.
  As shown in~\ref{thrmDAIsSumGeo}, we have
  $T = \sum_{i=1}^n \tau^{(i)}$, where $\tau^{(i)}$ is some geometrically
  distributed random variable with success parameter $p_i$ at least
  $(n-i+1)\pi_{\min}$.
  Thus, the variance of $\tau^{(i)}$ is
  \[ \frac{1 - p_i}{p_i^2} \le (n - i+1)^{-2}\pim^{-2}
  \]

  We would like to conclude saying that $\sum_{i=1}^n \tau^{(i)}$ is at
  most $\sum_{i=1}^n \Var(\tau^{(i)})$.
  However, this would not take into account the variation in wait times we
  experience due to the different orders in which we might collect the
  distinct coupons. Thus, in the below we condition on the order in which
  coupons are collected and use a trick related to the ``law of total
  variance''.

  Let $P = (c_1, c_2,\ldots, c_n)$ denote the random variable giving
  the order in which the distinct coupons
  are collected. Using the law of total expectation and the definition of
  (conditional) variance, the following holds for any random variables $T, P$:
  \begin{align*}
    \Var(T) & = \E{T^2} - \E{T}^2 \\
      & = \E{ \E{T^2 | P} } - \E{T}^2 \\
      & = \E{ \Var(T | P) + \E{T | P}^2 } - \E{T}^2. \tag{*}
  \end{align*}
  Conditioned on $P$, the values of each $\tau^{(i)}$ are independent, 
  so we have 
  \[ \Var(T | P) = \sum_{i=1}^n \Var(\tau^{(i)} | P) \le
    \sum_{i=1}^n (n - i+1)^{-2}\pim^{-2} = O(\pim^{-2}) = O(n^2).
  \]
  Moreover, as discussed in the proof of~\ref{thrmDAIsSumGeo},
  regardless of the value of $P$ the conditional distribution $T | P$
  is statistically dominated by $T_{\mathrm{ub}}$, which has
  expectation
  \[ \E{T | P} \le \E{T_{\mathrm{ub}}} = \pim^{-1} H_n
  \]
  On the other hand, proposition~\ref{thrmCouponExpectationLowerBound}
  shows that
  \[ \E{T} \ge (1 - O(1/\ln n))\pim^{-1} H_n
  \]
  overall.

  Plugging the above bounds into (*) gets us 
  \begin{align*} 
    \Var(T) 
    & \le O(n^2) + \Big( 1 - \big(1 - O(1/\ln n)\big)^2 \Big) 
      (\pim^{-1}H_n)^2 \\
    & = O(n^2) + O(1/\ln n)(\pim^{-1}H_n)^2 \\
    & = O(n^2 \ln n)
  \end{align*}
  Let $f(n) = ({\epsilonDotAlpha}/{\alpha_{\min}}) n \ln n$.
  Chebyshev's inequality finally tells us, for $n$ large enough that
  $|f(n) - \E{X}| = O(1/\ln n) f(n) < (\epsilon/2) f(n)$, we have
  \begin{align*}
    \P{ T \ne (1 \pm \epsilon)f(n) }
    & \le \P{ | T - \E{T} | \ge (\epsilon/2)f(n) } \\
    & = \frac{O(n^2 \ln n)}{(\epsilon/2)^2 f(n)^2 } \\
    & = O\big( 1/(\epsilon^2 \ln n) \big)
  \end{align*}

\end{proof}

\section{Missing proof on the expected
total number of proposals}
\label{appendixTotalProposalsExpectation}

\RestateThrmDACentralConcentration*

\begin{proof}
  Because $S\preceq T$, it suffices to prove a lower bound on $\E{S}$.

  Consider running deferred acceptances with re-proposals to derive the
  random variable $T$.
  We can write $S = T - R$, where $R = \sum_{i=1}^T R_i$ and
  $R_i$ is the indicator of whether the $i$th proposal is a repeat.
  By~\ref{thrmMaxIndividualProposals},
  with probability $1 - 1/n^2$, no man makes more than $O(\ln^2 n)$
  proposals. Let this event be denote by $E$.
  We thus have
  \[ \E{R_i | E} \le \frac{\alpha_{\max}\cdot O(\ln^2 n)}{\alpha_{\min} n}
    = O\left(\frac{\ln^2 n}{n}\right)
    \qquad\Longrightarrow\qquad \E{R|E} = O(\ln^2 n/n)\cdot\E{T | E}.
  \]
  To complete the proof, we use the fact that $T$ is not heavy-tailed enough
  for the low probability event $E$ to effect its expectation much.

  Specifically, proposition~\ref{thrmCouponUpperTail} shows that there exists
  a constant $K$ such that, for each $C=1,2,\ldots$, we have
  $\P{T \ge (1 + C)K n\ln n } \le O(n^{-C/2})$.
  Let $B$ be the event that $T\le 5Kn\ln n$, and
  for each integer $i\ge 5$, let $B_i$ be the event that 
  $iKn\ln n \le T \le (i+1)Kn\ln n$.
  We have
  \begin{align*}
    \E{T \1{\bar E \vee \bar B } }
    & \le \E{T \1{\bar E \wedge B} } + \sum_{i=5}^\infty \E{T \1{B_i} } \\
    & \le
    \frac{1}{n^2}\cdot \left( 5 K n\ln n \right)
      + \sum_{i=5}^\infty O(n^{-i/2}) (1 + i)K n\ln n
    = o(1)
  \end{align*}

  Thus, all told we we have
  \begin{align*}
    \E{R}
    & \le \big(1 - O(1/n^2)\big)\E{R | E \wedge B} 
      + \E{R \1{\bar E \vee \bar B} } \\
    & \le (1 - 1/n^2)\cdot O\left(\frac{\ln^2 n}{n}\right)
        \cdot 5Kn\ln n
      + \E{T \1{\bar E \vee \bar B} }
    = O(\ln^3 n) \\
  \Longrightarrow 
  \qquad \E{S} 
  & = \E{T} - \E{R} \ge \big(1 - O(\ln^2 n/n)\big)\E{T}.
  \end{align*}





\end{proof}

\section{Reaching smooth matching states}
\label{appendixReachingSmooth}

Recall definition~\ref{defPartialMatchingState},
which defines $\mu_{-L}$ as the state of
deferred acceptance after letting all men outside of set $L$ propose until
they find a match,
and~\ref{defSmoothMatchingState}, which calls $\mu_{-L}$ smooth if at most
$C_1 n\ln n$ proposals have been made overall,
and at most $n^{1 - C_2}$ women have received fewer than 
$C_3 \ln n$ proposals.

\RestateThrmSmoothWHP*
\begin{proof}
  By \ref{thrmStatisticalDominance} and \ref{thrmCouponUpperTail},
  there exists a $C_1$ such that
  the total number of proposals made in DA is $C_1 n\ln n$ with probability
  $1 - 1/n$. Thus, the same will hold before the men in $L$ propose.

  Now, consider running deferred acceptance with re-proposals
  with all men not in $L$, and let $T'$ denote the total number of proposals.
  Observe that this process will terminate as soon as $n-2$ distinct women
  receive proposals.
  Thus, by generalizing the standard argument used to bound the coupon
  collector with equal probabilities,
  we can write $T' = \sum_{i=1}^{n-2} \tau^{(i)}$, where $\tau^{(i)}$ are
  geometrically distributed random variables with parameter at most
  $(n - i + 1)\pi_{\max}$ (although the parameter depends on the order in
  which coupons are collected).
  A detailed description of this is given in~\ref{thrmDAIsSumGeo}.
  We thus know that $T' \succeq T'_{\mathrm{lb}}$, where
  $T_{\mathrm{lb}}' = \sum_{j=3}^{1/\pi_{\max}} T_j$ for independent draws
  $T_j \sim \Geo(j\pi_{\max})$ (similarly to~\ref{thrmCouponLowerTail}).
  Corollary~\ref{thrmChernoffCouponLastTwoMenLowerTail} shows that, up to
  constants, the same lower tail that held in 
  \ref{thrmCouponLowerTail} also applies here.
  In particular, there is a constant $c$ such that the total number of
  proposals is at least $c n \ln n$ with probability $1 - 1/n^{1/2}$.

  Now, consider the first $c n \ln n$ proposals of deferred
  acceptance with re-proposals, and let $X$ denote the number of proposals
  received by a fixed woman $w$.
  Each draw is random and independent over the women,
  and we can assume without loss of generality that $w$
  belongs to the bottom tier of women.
  Thus $w$ is selected each draw with probability
  $\pi_{\min}$ and
  $\E{X} = \pi_{\min} c n \ln n = O(\ln n)$.
  As the draws are independent, a standard ``multiplicative Chernoff bound''
  \ref{thrmMultChernoff} applies,
  and gets us that
  \[ \P{X \le \frac 1 2 \E{X}}
      \le \exp\left(-\frac{\E{X}}{4}\right)
      = n^{-D}
  \]
  for some constant $D$.
  Thus, the expected number of women with fewer than 
  $\E{X}/2$ proposals is thus at most $n^{1 - D}$.
  Thus Markov's inequality tells us that, for $cn\ln n$ proposals
  uniformly drawn from $\W$,
  \[ \P{\text{more than $n^{1 - D/2}$ women 
    have fewer than $\E{X}/2$ proposals}} 
    \le \frac{n^{1-D}}{n^{1-D/2}}
    = n^{-D/2}
  \]
  As deferred acceptance with re-proposals runs for $cn\ln n$ proposals
  with probability $1 - n^{1/2}$, we get that
  the number of women with fewer than $\E{X}/2$
  proposals (in deferred acceptance with re-proposals)
  is less than $n^{1 - C_2}$
  with probability $1 - O(n^{-D/2})$, where $C_2 = D/2$.

  To simply relate the above to ordinary DA,
  we use proposition~\ref{thrmNoTrippleProps} (proven below).
  Consider converting the run of deferred acceptance with
  re-proposals into a run of DA by ignoring repeated proposals.
  By proposition~\ref{thrmNoTrippleProps},
  with probability $1 - \widetilde O(1/n)$,
  every woman has at least $1/3$ of the
  proposals in DA as she got in deferred acceptance with re-proposals.
  Thus, with probability $1 - O(n^{-D/2})$ overall, 
  at most $n^{1 - C_2}$ women received fewer than $\E{X}/6 = C_3 \ln n$
  proposals.
\end{proof}

\begin{proposition}\label{thrmNoTrippleProps}
  Consider deferred acceptance with re-proposals.
  With probability $1 - \widetilde O(1/n)$, no man proposes to a single
  woman more than $3$ times.
\end{proposition}
\begin{proof}
  Proposition~\ref{thrmMaxIndividualProposals} states that,
  with probability $1-1/n$, no man makes more than $O(\ln^2 n)$ proposals
  in DA.
  Observe that the proof of theorem~\ref{thrmMaxPropsLastMan}
  would go through even if the last man $m_*$ sampled women with
  replacement, and thus proposition~\ref{thrmMaxIndividualProposals}
  holds for deferred acceptance with re-proposals as well.
  For a fixed man-woman pair $(m,w)$, the probability that $m$ proposes to
  $w$ at least $3$ times is at most
  \[ \binom{O(\ln^2 n)}{3} \pi_{\max}^3 
    = O\left(\frac{\ln^6 n}{n^3}\right).
  \]
  Taking a union bound over the $n^2$ pairs $(m,w)$ gets the desired
  result.
\end{proof}

\section{Expected rank in smooth matching states}
\label{sectionSmooth}

For this entire section,
consider a fixed pair of men $L = \{m_1, m_2\}$ with public scores
$\beta_1$ and $\beta_2$.
We consider all men other than $L$ proposing until they find a match,
and we make the following assumption
on $\mu_{-L}$, the state of DA before $m_1, m_2$ propose:
\begin{assumption}\label{assumeSmooth}
  For the entirety of section~\ref{sectionSmooth},
  fix $\mu_{-L}$ and assume $\mu_{-L}$ is smooth.
\end{assumption}
In particular, all men other than $\{m_1, m_2\}$ have been accepted, 
we assume that at most $C_1 n\ln n$ total proposals have been made,
and we assume all but a $n^{-C_2}$ fraction of women have received 
at least $C_3 \ln n$ proposals.
Denote these women by $\Whigh$, and the remaining women by $\Wlow$. Denote the distribution induced
by sampling from $\Whigh$ proportionally to public scores
by $\mathcalWhigh$ (and similarly with $\mathcalWlow$).
We use $\Es{L}{}$ to denote taking an expectation over the random
process of $m_1, m_2$ proposing in DA after starting from state $\mu_{-L}$.

For this section, we assume without loss of generality
that all $\beta_i$ are rescaled to be
at least $1$ (e.g. we can simply set $\beta_{\min}=1$).
For any woman $w$, recall that $\Gamma_w$ denotes the sum of the public scores
of men who have proposed to $w$ \emph{before $m_1$ or $m_2$ start
proposing}. For $\beta \ge 1$, define
\begin{align*}
  p_{w}(\beta) := \frac{\beta}{\beta + \Gamma_w}, &&
  p(\beta) := \Es{w\sim\W}{ p_w(\beta) }.
\end{align*}
That is, $p_w(\beta)$ is the probability that a woman $w$ in $\mu_{-L}$
accepts a proposal from a man with public score $\beta$,
and $p(\beta)$ is the probability that a woman randomly drawn according to
men's preferences accepts such a proposal.

Our plan is to show that the number of proposals each $m_i$ makes is
closely related to a geometric random variable with parameter $p(\beta_i)$.
We start off by calculating the order of magnitude of $p(\beta)$, and
studying how $p(\beta)$ scales between $\beta_1$ and
$\beta_2$.

\subsection{Scaling of probability of acceptance}
\label{sectionScalingProbability}

\begin{proposition}\label{thrmPBetaMagnitude}
  For any constant $\beta$, we have
  $p(\beta) = \Theta(1 / \ln n)$.
\end{proposition}
\begin{proof}
  Recall that the total number of proposals in $\mu_{-L}$ is at 
  most $C_1 n\ln n$. Thus, $\frac 1 n \sum_{w\in W} \Gamma_w
  \le \beta_{\max}C_1 \ln n$. Note the fact that
  that sampling according to $\mathcal W$ only affects the
  average by constant factors. In particular, we have
  \[ \Es{w\sim\W}{\Gamma_w} 
    = \frac{1}{n\epsilonDotAlpha}\sum_{w\in W} \alpha(w)\Gamma_w
    \le \frac{\alpha_{\max}}{n \alpha_{\min}} \sum_{w\in W} \Gamma_w
    \le O(\ln n).
  \]
  Observe that the function $f(x) = \beta / (\beta + x)$ is convex.
  Applying Jensen's inequality to the calculation of $p(\beta)$, we get
  \[
    p({\beta}) \ge \frac{\beta}{\beta + \Es{w\sim\mathcal{W}}{\Gamma_w}}
    \ge \Omega\left(\frac{1}{\ln n}\right)
  \]
  
  Because tier weights $\alpha_i$ are constant,
  when we sample $w\sim \W$, the probability of drawing
  a woman from $\Wlow$ is at most $\alpha_{\max}n^{1-C_2} / (\alpha_{\min}n)
  = O(n^{-C_2})$. For $w\in\Whigh$, we have
  $p_w({\beta}) \le O(1/\ln n)$.
  Thus, overall we get
  \begin{align*}
    p({\beta})
    & = \big(1 - O(n^{-C_2})\big) 
      \Es{ w\sim \mathcalWhigh} {p_w(\beta) }
      + O(n^{-C_2}) \Es{ w\sim \mathcalWlow}{ p_w(\beta) } \\
    & \le \big(1 - O(n^{-C_2})\big) O(1/\ln n) + O(n^{-C_2})
      = O(1/\ln n)
  \end{align*}
\end{proof}

\begin{proposition}\label{thrmBetaScalesWithinLogError}
  For any constant $\beta > 1$, we have
  \[ p(\beta) = \big(1 - O(1/\ln n)\big) \beta p(1). \]
\end{proposition}
\begin{proof}
  Observe that
  $p_{w}(\beta) \le \beta p_{w}(1)$ pointwise
  for each $w$, and thus taking expectations,
  $p(\beta) \le \beta p(1)$.

  On the other hand, if $\Gamma_w > 0$, we have
  \[ p_{w}(\beta) = \beta \cdot \frac{1}{1 + \Gamma_w} 
    \cdot \frac{1 + \Gamma}{\beta + \Gamma}
    = \beta p_w(1) \left(1 - \frac{\beta - 1}{\beta + \Gamma_w}\right)
    \ge \beta p_w(1) \big(1 - O(1/\Gamma_w)\big). \]
  Applying the law of total probability to both $p(1)$ and $p(\beta)$, we have 
  \begin{align*}
    p({\beta})
    & = \big(1 - O(n^{-C_2})\big) \Es{ w\sim \mathcalWhigh}
      {\beta\big(1 - O(1/\log n)\big)p_w(1) }
    + O(n^{-C_2}) \Es{ w\sim \mathcalWlow}{ p_\beta(w) } \\
    & \ge  \beta \big(1 - O(1/\log n)\big)
    \Es{ w\sim \mathcalWhigh}{ p_w(1) } \tag{*}
  \end{align*}
  We can formally show that $p(1)$ is close to 
  $\Es{ w\sim \mathcalWhigh}{ p_w(1) }$ as follows:
  \begin{align*}
    p({1})
    & = \big(1 - O(n^{-C_2})\big) \Es{ w\sim \mathcalWhigh}{ p_w(1) }
      + O(n^{-C_2}) \Es{ w\sim \mathcalWlow}{ p_\beta(w) } \\
      & \le \big(1 - O(n^{-C_2})\big)
      \Es{ w\sim \mathcalWhigh}{ p_w(1) } + O(n^{-C_2}) \\
    \Longrightarrow \qquad
    \Es{ w\sim \mathcalWhigh }{ p_w(1) }
      & \ge \frac{p(1) - O(n^{-C_2})}{1 - O(n^{-C_2})}
      \ge \big(1 + O(n^{-C_2/2})\big) p(1).
  \end{align*}
  Where the last inequality uses the fact that $p(1) \ge \Omega(1/\ln n)$.
  Thus, plugging into (*), we get
  \[
    p({\beta})
    \ge  \beta \big(1 - O(1/\log n)\big) p(1).
  \]
\end{proof}

\subsection{The expected number of proposals made}
\label{sectionExpectedNumberProps}
\ %

We start by bounding the order of magnitude of the number of proposals a
man needs to make, even if the matching state has changed a noticeable
amount from the initial state $\mu_{-L}$.
\begin{proposition}\label{thrmContinuedProposalExpectation}
  Consider any matching state $(\mu, P)$
  (with any number of remaining unmatched men)
  in which fewer than $C_4 n \ln n$
  proposals have been made overall, for some constant $C_4$.
  Suppose $m_i$ has proposed to fewer than $n/4$ women total,
  and let $\rcontd_i$ be the number of additional proposals $m_i$ will make
  before DA terminates.
  Then $\E{\rcontd_i} \le O(\ln n)$.
\end{proposition}
\begin{proof}
  Consider letting all men other than $m_i$ propose until all are accepted
  by some woman.
  Certainly this is statistically dominated by the random variable $T$
  giving the total number of proposals when starting from an empty matching
  state. Thus, the number of additional proposals is 
  $C_1n\ln n$ with probability $1 - O(1/n)$.
  Now, we track the probability of acceptance of $m_i$ as
  in~\ref{thrmMaxPropsLastMan}, but this time we consider the expectation,
  not the concentration.

  Let $K$ denote the number of times $m_i$ is tentatively accepted by some
  woman. Whenever a man $m$ other than $m_i$ is making a proposal, consider
  the event that $m$ proposes to either an unmatched woman, or the current
  match $w$ of $m_i$. The probability that, in this event, the proposal goes to
  $w$ is at most the constant 
  $q = \alpha_{\max}/(\alpha_{\min} + \alpha_{\max})$.
  Note that $m_i$ cannot be kicked back out from $w$ unless $w$ receives
  another proposal before the unmatched woman gets a proposal.
  Thus, $\P{ K \ge k+1 | K \ge k} \le q$, and $K$ is statistically
  dominated by a geometric distribution with parameter $q$.

  Now, for $k=1,\ldots,K$, let $P_k$ denote the number of proposals
  $m$ makes between his $(k-1)$th tentative acceptance and his $k$th
  tentative acceptance. For any specific one of these proposals,
  let $W^*$ denote the set of women $m_i$ has not yet proposed to,
  and also the weighted distribution over these women.
  Let $\Gamma_w$ denote the total weight of proposals that a woman 
  $w\in W^*$ has received before that proposal.
  As long as there have been at most $O(n\ln n)$ proposals overall, and
  $m_i$ has proposed to fewer than $n/2$ women, we have
  $\Es{w\sim \W^*}{\Gamma_w} \le O(\ln n)$ and
  $\Es{w\sim \W^*}{\beta_i / (\beta_i + \Gamma_w) } \ge \Omega(1/\ln n)$
  by Jensen's inequality.
  Thus, in these cases $P_k$ is statistically dominated by a geometric
  distribution with parameter $\Omega(1/\ln n)$.

  With probability $1 - 1/n$, by~\ref{thrmMaxPropsLastMan} we know
  that $m_i$ makes fewer than $O(\ln^2 n)$ proposals, so the above bounds on $K$ and each $P_k$
  hold. With the remaining probability, simply use the upper bound of $n$
  proposals. Thus, all told we have
  \[
    \E{ \rcontd_i }
      = \E{ \sum_{k=1}^K P_k }
      \le \big(1 - O(1/n)\big) (1/q) O(\ln n) + O(1/n)\cdot n
      = O(\ln n).
  \]
\end{proof}

We now turn to a more fine-grained study of the number of proposals $m_1$
and $m_2$ make, denoted $r_1$ and $r_2$.
We reason about the random process of DA starting from $\mu_{-L}$
by separating it into two phases:
\begin{itemize}
  \item \textbf{Phase (1).}
    First, consider $m_1$ proposing until his first acceptance
    (possibly kicking out man $m_1'$), say by woman $w_1$.
    Next let $m_2$ propose until his first acceptance
    (possibly kicking out man $m_2'$ (where $m_2'$ may equal $m_1$)),
    say by woman $w_2$.
    Let $\rinit_i$ denote the number of proposals $m_i$ makes during this
    phase.
  \item \textbf{Phase (2).}
    Afterwards, run DA as normal, staring with
    men $L' = \{m_1', m_2'\}$ proposing.
    Let $\rrest_i$ denote the number of proposals $m_i$ makes during this
    phase (or set $\rrest_i = 0$ if $m_i$ ends up staying at $w_i$
    in the final matching).
\end{itemize}
Note that $r_i = \rinit_i + \rrest_i$.

Now, for each $m_i$ separately, consider the following random process:
each time step, sample a $w\sim\W$ independently and with replacement,
and have $w$ accept $m_i$ with probability
${\beta_i}/({\beta_i + \Gamma_w})$ independently each time.
For $i=1,2$, let $\rgeo_i$ denote the number of draws until 
the first acceptance.
Then $\rgeo_i$ is distributed exactly according to $\Geo(p(\beta_i))$,
so proposition~\ref{thrmBetaScalesWithinLogError}
immediately implies the following:
\begin{proposition}\label{thrmInitialGeometricExpectation}
  For $i=1,2$, we have 
  \[ \E{\rgeo_i} = \big(1 + O(1/\ln n)\big)\frac{1}{\beta_i p(1)}. \]
\end{proposition}

Note that the random process we consider for $\rgeo_i$ is slightly
different than deferred acceptance with re-proposals,
because in deferred acceptance women never accept a proposal from a man who
has already proposed to them.
Nonetheless, we will be able to show that in expectation, $\rgeo_i$ is
only a $o(1)$ additive constant away from $\rinit_i$.
The core reason for this is the following:
by proposition~\ref{thrmMaxIndividualProposals},
$m_1$ and $m_2$ make fewer than $O(\ln^2 n)$ proposals with very high
probability. In particular, this is true during phase (1).
Thus, there is at most a $O(\ln^4 n / n)$ probability that
$m_1$ and $m_2$ propose to the same woman,
or make a repeat proposal themselves, during phase (1).
So phase (1) can contribute at most $O(\log^6 n / n)$
to the expectation in this case.
In any other case, the distribution of proposals in phase (1) is
distributed exactly as in $\rgeo_i$.

\begin{proposition}\label{thrmInitCloseToGeo}
  For $i=1,2$, we have
  \[
    \Es{L}{\rinit_i} = \E{\rgeo_i} \pm o(1)
  \]
\end{proposition}
\begin{proof}
  In phase (1), the changes we need to make to turn the distribution of
  $\rgeo_i$ into the distribution of $\rinit_i$ are exactly the following:
  \begin{itemize}
    \item If $m_1$ or $m_2$ were rejected by a woman they already proposed
      to, that proposal should be ignored.
      If $m_1$ or $m_2$ were \emph{accepted} by a woman they already
      proposed to, that proposal should also be ignored,
      and $m_i$ should continue proposing to more women.
    \item If $m_2$ proposes to a woman who $m_1$ proposed to during
      $\rinit_1$, the probability that $m_2$ is accepted should change
      from $\beta_2 / (\beta_2 + \Gamma_w)$
      to $\beta_2 / (\beta_2 + \beta_1 + \Gamma_w)$
  \end{itemize}
  Thus, differences arise between $\rgeo_i$ and $\rinit_i$ only when a
  repeated proposal is made.

  With probability $1 - O(1/n^2)$, neither $m_1$ nor $m_2$ make more than
  $C \ln^2 n$ proposals (even in deferred acceptance with re-proposals),
  for some constant $C$.
  Thus, the probability that any pair of proposals is repeated is at most
  \[ \binom{2C\ln^2 n }{2} n \pi_{\max}^2
    = O\left( \frac{\ln^4 n}{n} \right).
  \]
  If any of above corrections need to be made, imagine stopping
  running $\rgeo_i$ and run ordinary deferred acceptance from the current
  matching state\footnote{
    Formally, this describes a coupling of the joint distributions
    of $(\rinit_1,\rinit_2)$ and $(\rgeo_1, \rgeo_2)$ in which
    each joint distribution differs with probability $\widetilde O(1/n)$.
  }. In this case, we still know by proposition~\ref{thrmMaxIndividualProposals} that 
  $r_i = O(\ln^2 n)$ with probability $1 - 1/n^2$.
  Thus, all told we have
  \[ \E{\rinit_i} = \E{\rgeo_i} 
    + O\left(\frac{\log^4 n}{n}\right)\cdot O(\ln^2 n)
    + O\left(\frac{1}{n^2}\right)\cdot n
    = \E{\rgeo_i} + o(1).
  \]
\end{proof}

We now separately consider phase (2).
During phase (2), between times where $m_1$ or $m_2$ are
proposing, many other proposals might be made by different men,
changing the effective values of $\Gamma_w$ and thus the probability
that a proposal by $m_1$ or $m_2$ is accepted.
Thus, we'd like to say that the rank of $m_1$ and $m_2$ is
approximately their rank in phase (1).

Indeed, it turns out we have $\E{\rrest_i} = O(1)$.
The proof sketch is the following:
Because at most a $n^{-C_2}$ fraction of women are in $\Wlow$,
there is a $O(\ln^2 n/ n^{C_2})$ chance that $m_1$ or $m_2$'s first
acceptance is to a woman in $\Wlow$.
Thus, this case can contribute at most $\poly\log n / n^{C_2}$ to the
expectation.
In the other case, a woman in $\Whigh$ has accepted $m_1$ or $m_2$,
and she will only reject him with probability
$O(1/\ln n)$ (to prove this, we will use the fact that a proposal is at
most a constant times less likely to go to an unmatched woman than the match
of $m_1$ or $m_2$, so these women won't receive too many additional
proposals).
This probability is non-negligible, however,
proposition~\ref{thrmContinuedProposalExpectation} says that
the number of additional proposals in phase (2) has expectation at
most $O(\ln n)$. Thus, phase (2) contributes $O(1)$ proposals 
in expectation.

\begin{proposition}\label{thrmRestConstant}
  For $i=1,2$, we have
  \[
    \Es{L}{\rrest_i} = O(1)
  \]
\end{proposition}
\begin{proof}

  Phase (2) starts with a set of at most
  two men $L' = \{m_1', m_2'\}$ unmatched.
  During this phase, we no longer need to very carefully 
  track the relationship between the
  acceptance probabilities and the initial matching state $\mu_{-L}$.
  Thus, we can consider each $m_i$ separately for $i=1,2$.
  For $j\ne i$, let $m_j$ propose until he finds a match,
  and recall that this add fewer than $C_1 n\ln n$ proposals
  with probability $1-1/n^2$.

  Again, with probability $1-1/n^2$, we get that $m_i$ makes fewer than
  $C\ln^2 n$ proposals, so the probability that he made even a single
  proposal to one of the $n^{1-C_2}$ women in $\Wlow$ is at most
  \[ \frac{1}{n^2}
    + C\ln^2 n \cdot \frac{\alpha_{\max} n^{1-C_2}}{\alpha_{\min} n} 
    = O\left(\frac{\ln^2 n}{n^{C_2}}\right).
  \]
  In particular, the probability that $m_i$ \emph{matched} to a woman in
  $\Wlow$ is $O\left({\ln^2 n}/{n^{C_2}}\right)$.

  Consider the case where $m_i$ matched to women
  $w_i$ in $\Whigh$. Let this event be denoted $H$.
  Let the (at most) two unmatched women in this case be $w^u_1$ and $w^u_2$.
  For any subsequent proposal in DA, consider the event $P$ 
  that a proposal goes to any woman in $\{ w_i, w^u_1, w^u_2 \}$.
  Conditioned on a proposal landing in $P$,
  there is a constant lower bound $q = \alpha_{\min}/(3\alpha_{\max})$
  on the probability that a proposal goes to $w^u_1$ and $w^u_2$.
  Let the number of proposals to $w_i$ before
  both $w^u_1$ and $w^u_2$ see a proposal (and thus DA terminates)
  be denoted $K$.
  For $j=1,\ldots,K$, let $X_j$ be the event that $w_i$ accepts the $j$th
  proposal made to her.
  Then the probability that $w_i$ accepts a new proposal can be upper
  bounded as follows:
  \[ \E{ \P{X_1 \cup \ldots \cup X_K|K} }
    \le \E{ \P{X_1|K} + \ldots + \P{X_K|K} }
    \le \E{K}\cdot O(1/\ln n)
    = O(1/\ln n),
  \]
  where we observe that the expectation of $K$ is of constant order\footnote{
    One can formally verify this as follows:
    Let $K'$ be defined by iteratively sampling over $\{w_i, w^u_1,
    w^u_2\}$, where each $w^u_i$ is sampled with probability $q$
    and $w_i$ with probability $1 - 2q$, and letting $K'$ be the number of
    draws needed until $w^u_1$ and $w^u_2$ have both been sampled.
    Then certainly $K\preceq K'$.
    But $K'$ can be written as $G_1 + G_2$, where $G_1\sim\Geo(2q)$
    and $G_2\sim\Geo(q)$. Thus, $\E{K'} = 3/(2q) = O(1)$.
  }.

  Moreover, observe that the total number of remaining proposals
  made to \emph{any} woman is statistically dominated by $G_1 + G_2$, where
  $G_1\sim \Geo(2\pim)$ and $G_2\sim \Geo(\pim)$.
  Thus, with probability $1 - 1/n^2$, the total number of remaining
  proposals is $O(n\ln n)$, so
  proposition~\ref{thrmContinuedProposalExpectation} applies, and even if
  $m_i$ is rejected from $w_i$, he makes $O(\ln n)$ additional proposals in
  expectation.
  Thus, in event $H$ we have
  \[ \E{\rrest_i | E}
    = O\left(\frac{1}{n^2}\right)\cdot n
    + O\left(\frac{1}{\ln n}\right)\cdot O(\ln n)
    = O(1)
  \]

  All told, we have
  \[ \E{\rrest_i}
    = O\left(\frac{1}{n^2}\right)\cdot n
    + O\left(\frac{\ln^2 n}{n^{C_2}}\right)\cdot O(\ln^2 n)
    + \left(1 - O\left(\frac{\ln^2 n}{n^{C_2}}\right)\right)
      \E{\rrest_i | E}
    = O(1).
  \]
\end{proof}

Finally, combining
propositions~\ref{thrmInitialGeometricExpectation},~\ref{thrmInitCloseToGeo}
and~\ref{thrmRestConstant} gets the main result of this subsection.
\RestateThrmSmoothRanksScale*
\begin{proof}
  By proposition~\ref{thrmPBetaMagnitude}, we have 
  $\E{\rgeo_i} = \Theta(\ln n)$, so
  \begin{align*}
    \Es{L}{r_i} 
    & = \E{\rinit_i} + \E{\rrest_i} \\
    & = \E{\rgeo_i} + O(1) \\
    & = \big(1 + O(1/\ln n)\big)\frac{1}{\beta_i p(1)}.
  \end{align*}
  Thus both $\Es{L}{r_1}$ and $\Es{L}{r_2}$ relate to $1/p(1)$ as follows:
  \[
    \big(1 + O(1/\ln n)\big)^{-1} \beta_1 \Es{L}{r_1}
    = \frac{1}{p(1)}
    = \big(1 + O(1/\ln n)\big)^{-1} \beta_2 \Es{L}{r_2},
  \]
  which proves our result.
\end{proof}

\subsection{The covariance in the number of proposals}
\label{sectionCovarianceNumberProps}
\ %

We continue to reason about $r_i = \rinit_i + \rrest_i$ 
for $i=1,2$, as defined in~\ref{sectionExpectedNumberProps}.
Now, we are interested in their joint
distribution, and whether there is significant \emph{correlation}
in $r_1$ and $r_2$. 
In particular, we want to show that $\E{r_1 r_2}\approx \E{r_1}\E{r_2}$.
We find that the correlation is indeed lower order,
for similar reasons to those exploited above:
most of the expected value of $r_1 r_2$ comes from $\rinit_1 \rinit_2$, i.e. proposals in phase (1),
and the two men in $L$ only interfere with each other
in phase (1) with probability $\poly\log n/n$.

We have
\[ \Es{L}{r_1 r_2}
  = \Es{L}{ \rinit_1\rinit_2 + \rinit_1\rrest_2
    + \rrest_1\rinit_2 + \rrest_1\rrest_2 }.
\]
We proceed to separately reason about these contributions.
We find that $\rinit_1 \rinit_2$ has the bulk of the contribution,
while all other terms are lower order.

\begin{proposition}
  \label{thrmInitInitCovar}
  We have
  \[ \E{\rinit_1 \rinit_2} 
    = \E{r_1}\E{r_2} \pm O(\ln n)
  \]
\end{proposition}
\begin{proof}
  We showed in proposition~\ref{thrmInitCloseToGeo} that jointly,
  you only need to change the distribution of
  $(\rinit_1,\rinit_2)$ from that of $(\rgeo_1,\rgeo_2)$
  with probability $O(\ln^4 n / n)$.
  Moreover, with probability $1 - O(1/n^3)$, by
  proposition~\ref{thrmMaxIndividualProposals} both $r_1$ and $r_2$ are at
  most $O(\ln^2 n)$, so we get
  \[
    \E{\rinit_1 \rinit_2}
    = \E{\rgeo_1 \rgeo_2} 
    \pm O\left(\frac{\ln^4 n}{n}\right)\cdot O(\ln^4 n) 
    + \frac{1}{n^3}\cdot n^2
    = \E{\rgeo_1 \rgeo_2} \pm o(1).
  \]
  Moreover, $\rgeo_i$ are independent, and by
  propositions~\ref{thrmInitCloseToGeo} and~\ref{thrmRestConstant},
  we have $\E{\rgeo_i} = \E{r_i} \pm O(1)$.
  Thus, 
  \begin{align*}
    \E{\rgeo_1 \rgeo_2}
    = \E{\rgeo_1} \E{\rgeo_2}
    & = \big(\E{r_1} \pm O(1)\big)\big( \E{r_2} \pm O(1)\big) \\
    & = \E{r_1}\E{r_2} \pm O(\E{r_1} + \E{r_2}).
  \end{align*}
\end{proof}

To complete our proof that the covariance of $r_1$ and $r_2$ is
$o(\ln^2 n)$, it will suffice to show that the remaining terms of
are $O(\ln^{3/2} n)$. While it may be possible to improve this to
$O(\ln n)$, the proof of our bound is very simple, largely due to the
following variation of proposition~\ref{thrmMaxPropsLastMan}:

\begin{proposition} \label{thrmQuasipolyPropBound}
  For $i=1,2$, we have $r_i \le O(\ln^{3/2} n)$ with probability
  at least $1 - \exp(-\sqrt{\ln n})$.
\end{proposition}
\begin{proof}
  For $j\ne i$, consider letting $m_j$ propose until he finds a match.
  The number of total proposals (from all men)
  this requires is statistically 
  dominated by $\Geo(\pim)$, and thus less than $O(n\ln n)$
  with probability $1 - 1/n$.
  In this case, at most $O(n\ln n)$ proposals have been made
  in total when $m_i$ starts proposing,
  so proposition~\ref{thrmMaxPropsLastMan} holds as written.
  The only change we need to make to achieve our result is that,
  in point~\ref{itemMaxPropsLastProbWrapup}, there exists a constant
  $K$ such that the chance $m_i$ makes more than $K\ln^{3/2} n$ proposals
  is bounded by
  \[ \big(1-O(1/\ln n)\big)^{K\ln^{3/2} n}
    \leq \exp(-\ln^{1/2} n)
  \]
\end{proof}


\begin{proposition}
  \label{thrmInitRestCovar}
  For $i\ne j$, we have
  \[ \E{\rinit_i \rrest_j} = O(\ln^{3/2} n)
  \]
\end{proposition}
\begin{proof}
  Let $E$ denote the event that $\rinit_i \le O(\ln^{3/2} n)$.
  By~\ref{thrmQuasipolyPropBound}, we know $E$ holds
  with probability $1 - \exp(-\sqrt{\ln n})$.
  Now, image letting $m_i$ and $m_j$ propose in phase (1), 
  then letting $m_i$ find a match in phase (2). 
  Because $\rinit_i$ and $\rrest_j$ occur during different phases
  of the algorithm, and we condition on $\rinit_i$ not being too
  large, we can separately reason about $\rrest_j$ (without, 
  for example, worrying about its effect on $\rrest_i$).

  With probability $1 - 1/n^2$, we know that
  the total number of proposals required for before
  $m_i$ finds a match is $O(n\ln n)$ additional proposals.
  Using this, it is not hard to check that
  the proof of~\ref{thrmRestConstant} goes through
  conditioned on event $E$, and thus $\E{\rrest_j \big| E} = O(1)$.
  Even if $E$ does not hold, by~\ref{thrmMaxIndividualProposals}, with
  probability $1 - 1/n^3$ both $m_1$ and $m_2$ make $O(\ln^2 n)$ proposals.
  Thus, all told we have
  \begin{align*}
    \E{\rinit_i \rrest_j}
    & \le \big(1 - \exp(-\sqrt{\ln n})\big) 
      \cdot O(\ln^{3/2} n) \E{\rrest_j \big| E} \\
      & \qquad + \exp(-\sqrt{\ln n}) \big(1 - 1/n^3\big) O(\ln^2 n)\cdot O(\ln^2 n)
      + \frac{1}{n^3} n\cdot n \\
    & = O(\ln^{3/2} n)\cdot O(1) + o(1)
  \end{align*}
\end{proof}

\begin{proposition}
  \label{thrmRestRestCovar}
  We have
  \[ \E{\rrest_1 \rrest_2} = O(\ln n)
  \]
\end{proposition}
\begin{proof}
  The joint distribution of $\rrest_1$ and $\rrest_2$ is difficult to
  reason about, because (unlike in the case of $\rinit_1$ and $\rrest_2$)
  both men have to take turns proposing during phase (2).
  Thus, this proof focuses on generalizing the arguments
  of~\ref{thrmRestConstant} to show that the product
  $\rrest_1 \rrest_2$ is nonzero with probability at most $O(1/\ln^2 n)$.

  To this end, let $w_1, w_2$ be the tentative matches of $m_1, m_2$ at the
  end of phase (1), and let $w^u_1, w^u_2$ be the unmatched women.
  Suppose $w_1$ and $w_2$ are both in $\Whigh$.
  As we showed in~\ref{thrmRestConstant}, this occurs with probability
  $1 - O(\ln^4 n/n^{C_2})$.
  Recall that women in $\Whigh$ will accept a new proposal with 
  probability at most $O(1/\ln n)$. 

  For any subsequent proposal in DA, consider the event $P$ 
  that a proposal goes to any woman in $\{ w_1, w_2, w^u_1, w^u_2 \}$.
  Conditioned on a proposal landing in $P$,
  there is a constant lower bound $q = \alpha_{\min}/(4\alpha_{\max})$
  on the probability that a proposal goes to $w^u_i$ for $i=1,2$.
  Before both $w^u_1$ and $w^u_2$ see a proposal (and thus DA terminates),
  let the number of proposals to $w_1$ be denoted $K_1$
  and the number of proposals to $w_2$ be $K_2$.
  For $i=1,2$ and $j=1,\ldots,K$,
  let $X^i_j$ be the event that $w_i$ accepts the $j$th
  proposal made to her.
  Then the probability that both $w_1$ and $w_2$ a new proposal can be upper
  bounded as follows:
  \begin{align*}
    & \Es{K_1, K_2}{ \P{\big(X^1_1 \cup \ldots \cup X^1_{K_1}\big)
      \cap\big(X^2_1 \cup \ldots \cup X^2_{K_2}\big)} } 
    \\
    &\qquad = \Es{K_1, K_2}{ \P{\big(X^1_1 \cup \ldots \cup X^1_{K_1}\big)}
      \cdot \P{\big(X^2_1 \cup \ldots \cup X^2_{K_2}\big)} } \\
    &\qquad \le \Es{K_1, K_2}{ \big(\P{X^1_1}+ \ldots +\P{X^1_{K_1}}\big)
      \cdot \big(\P{X^2_1}+ \ldots + \P{X^2_{K_2}} \big) } \\
    &\qquad \le \Es{K_1,K_2}{ K_1 K_2 \cdot O(1/\ln^2 n) } \\
    &\qquad \le \E{K_1 K_2}\cdot O(1/\ln^2 n)
    = O(1/\ln^2 n),
  \end{align*}
  where we use the fact that whether $w_1$ or $w_2$ accepts proposals is
  independent, and the observation that the expectation of 
  $K_1 K_2$ is constant\footnote{
    One can formally verify this as follows:
    Let $K'$ be defined by iteratively sampling over
    $\{w_1, w_2, w^u_1, w^u_2\}$,
    where each $w^u_i$ is sampled with probability $q$
    and each $w_i$ with probability $1/2 - q$, and letting $K'$ be the number of
    draws needed until $w^u_1$ and $w^u_2$ have both been sampled.
    Then certainly $K_i \le K_1 + K_2 \preceq K'$.
    But $K'$ can be written as $G_1 + G_2$, where $G_1\sim\Geo(2q)$
    and $G_2\sim\Geo(q)$. Thus,
    $\E{K_1 K_2} \le \E{(K')^2} \le \Var(G_1 + G_2) = O(1)$.
  }.
  
  We showed above that with probability
  $\big(1 - O(\ln^4 n/n^{-C_2})\big)\big(1 - O(1/\ln^2 n)\big)$, 
  either $m_1$ or $m_2$ remain at their tentative
  match from phrase (1), and $\rrest_1\rrest_2 = 0$.
  By proposition~\ref{thrmQuasipolyPropBound},
  even if $m_1$ and $m_2$ both leave their tentative
  match from phase (1), with probability $1 - \exp(-\sqrt{\ln n})$,
  each of them make at most $O(\ln^{3/2} n)$ additional proposals.
  Even if this fails, by proposition~\ref{thrmMaxIndividualProposals},
  $m_1$ and $m_2$ make $O(\ln^2 n)$ proposals with probability
  $1 - 1/n^3$.
  All told, we have
  \begin{align*}
    \E{\rrest_1 \rrest_2}
      & = \big(1 - O(\ln^4 n/n^{-C_2})\big)\cdot
        O\left(\frac{1}{\ln^2 n}\right) \big(O(\ln^{3/2} n)\big)^2 \\
      & \qquad + \Big( O(\ln^4 n/n^{-C_2}) + \exp(-\sqrt{\ln n})\Big)
        \big(O(\ln^{2} n)\big)^2
      + \frac{1}{n^3} n^2 \\
    & = O(\ln n)
  \end{align*}
\end{proof}

Combining propositions~\ref{thrmInitInitCovar},~\ref{thrmInitRestCovar},
and~\ref{thrmRestRestCovar}, we get our main result on covariance in smooth
matching states.
\RestateThrmSmoothRanksCovariance*
\begin{proof}
  We have
  \[
    \Cov(r_i, r_j) = \E{r_i r_j} - \E{r_i}\E{r_j}
    = \pm O(\ln n) + O(\ln^{3/2} n) + O(\ln n)
    = O(\ln^{3/2} n)
  \]
\end{proof}

\section{Missing proofs for the expected rank of men}
\label{appendixMenRanksProof}

First, we prove two crucial lemmas.

\RestateThrmMenRanksProportional*
\begin{proof}
  Let all men other than $L = \{m_i, m_j\}$ propose,
  for $m_i$ in tier $i$ and $m_j$ in tier $j$.
  By proposition~\ref{thrmSmoothWHP}, with probability 
  $1 - n^{-\Omega(1)}$, the matching state $\mu_{-L}$ is smooth,
  and proposition~\ref{thrmSmoothedRanksProportional} applies.

  First, we verify that conditioning on $\mu_{-L}$ being smooth cannot
  change the expectation of $r_j$ much.
  With probability $1 - 1/n^2$, $m_j$ makes at most $O(\ln^2 n)$ proposals, so
  \begin{align*}
    & \E{ r_j } = 
    \frac{1}{n^2} \cdot n + n^{-\Omega(1)}\cdot O(\ln^2 n) +
    \big(1 - n^{-\Omega(1)}\big) \E{ r_j \big| \text{ $\mu_{-L}$ is smooth} }
    \\ \Longrightarrow \qquad 
    & \E{ r_j \big| \text{ $\mu_{-L}$ is smooth} }
    = \big(1 + n^{-\Omega(1)}\big) \E{ r_j } - o(1).
  \end{align*}

  Now we can relate this to $r_i$
  using~\ref{thrmSmoothedRanksProportional}.
  \begin{align*}
    \E{r_i} 
    & = \frac{1}{n^2} \cdot n + n^{-\Omega(1)}\cdot O(\ln^2 n)
      + \big(1 - n^{-\Omega(1)}\big)
      \Es{\mu_{-L}}{ \Es{L}{ r_i } | \text{ $\mu_{-L}$ is smooth}} \\
    & = o(1) + 
      \Es{\mu_{-L}}{ \big(1\pm O(1/\ln n)\big)\frac{\beta_j}{\beta_i}\Es{L}{ r_j }
      \Big| \text{ $\mu_{-L}$ is smooth}} 
    =\big(1\pm O(1/\ln n)\big)\frac{\beta_j}{\beta_i} \E{r_j}.
  \end{align*}
\end{proof}

\RestateThrmMenRanksExpectation*
\begin{proof}
  Let the number of tiers be $k$.
  Using theorem~\ref{thrmDACentralConcentration}
  and symmetry for men in the same tier, we get the following:
  \begin{align*}
    \big(1 \pm O(1/\ln n)\big)
    \frac{\epsilonDotAlpha}{\alpha_{\min}} n \ln n
    = \E{S}
    & = \sum_{i=1}^k n \delta_i \E{r_i} \\
    & = \sum_{i=1}^k n \delta_i \frac{\beta_{j}}{\beta_i} 
      \big(1 \pm O(1/\ln n)\big)\E{r_{j}} \\
    & = \big(1 \pm O(1/\ln n)\big)
      (n\bm{\delta}\cdot\bm{\beta}^{-1}) \beta_j \E{r_{j}},
  \end{align*}
  and the theorem follows.
\end{proof}

Now we prove our main theorem.

\RestateThrmMenRanks*

\begin{proof}
  Let $r_1, r_2$ be the ranks of an arbitrary pair of men
  $L = \{m_1, m_2\}$, both in tier $j$.
  To prove the theorem, it will suffice to bound the covariance of $r_1$
  and $r_2$ using~\ref{thrmSmoothCovar}.
  Recall from definition~\ref{defPartialMatchingState} that $\mu_{-L}$ denotes the partial matching state excluding $L$. Let $U$ be a random variable indicating one of three things
  about $\mu_{-L}$:
  $U = \texttt{s}$ if $\mu_{-L}$ is smooth (this occurs with probability $1
  - n^{-\Omega(1)}$),
  $U = \texttt{t}$ if $\mu_{-L}$ is not smooth but at most $O(n\ln n)$
  total proposals have been made (this occurs with probability $n^{-\Omega(1)}$),
  and $U = \texttt{c}$ otherwise (this occurs with probability $1/n^3$
  by proposition~\ref{thrmCouponUpperTail}).

  The ``law of total covariance'' allows us to bound the total covariance
  of $r_1$ and $r_2$ by separately consider the cases $U = \texttt{s},
  \texttt{t}, \texttt{c}$. The first case is exactly
  proposition~\ref{thrmSmoothCovar}: $\Cov(r_1, r_2 | U=\texttt{s} )
  \le O(\ln^{3/2}n)$. The second cases uses the fact that
  (under the assumption that at most $O(n\ln n)$ proposals have been made
  total) proposition~\ref{thrmMaxPropsLastMan} says that $r_1, r_2 \le
  O(\ln^2 n)$ with probability $1 - 1/n^3$.
  Thus, we have $\Cov(r_1, r_2 | U=\texttt{t} )
  \le O(\ln^4 n) + \frac{1}{n^3}\cdot n^3 = O(\ln^4 n)$.
  The final case simply uses $r_1, r_2\le n$:
  $\Cov(r_1, r_2 | U=\texttt{c} ) \le n^2$. 
  We also bound the variance of the random variable $\E{r_i | U}$ for
  $i=1,2$. Because $\E{r_i | U}$ takes a constant value ($\E{r_i |
  U=\texttt{s}}$) with high probability, we have
  $ \underset{U}{\Var}( \E{r_i | U} )
  \le O(\ln^4 n)\cdot\frac{1}{n^{\Omega(1)}} + n^2 \frac{1}{n^3} = o(1)$.
  All told, we get
  \begin{align*}
    \Cov(r_1, r_2) 
    & \le \Es{U}{ \Cov(r_1, r_2 | U) }
      + \underset{U}{\Cov}( \E{r_1 | U} , \E{r_w | U} ) \\
    & \le O(\ln^{3/2} n) + n^{-\Omega(1)}\cdot O(\ln^4 n) +
      \frac{1}{n^3}\cdot n^2 + o(1) 
    = O(\ln^{3/2} n)
  \end{align*}

  Additionally observe that, because $r_i \le O(\ln^2 n)$ with probability
  $1 - 1/n^3$ (\ref{thrmMaxIndividualProposals} again),
  we have $\Var(r_i) \le O(\ln^4 n)$.
  Summing over the $(\delta_j n)^2$ pairs of men thus gets us
  \[ \Var(\BarRM_j)
    \le \frac{1}{(\delta_j n)^2}
    \left( (\delta_j n)\cdot O(\ln^4 n) + (\delta_j n)^2 \cdot O(\ln^{3/2} n)
    \right)
    = O(\ln^{3/2} n)
  \]

  Let $f(n) = {\epsilonDotAlpha}/(\alpha_{\min} \beta_j
  \deltaDotBeta^{-1} ) \ln n$.
  Finally, Chebyshev's inequality plus theorem~\ref{thrmMenRanksExpectation} 
  says that for any $\epsilon > 0$ and $n$
  large enough, we have
  \begin{align*}
    \P{ \BarRM_j \ne (1 \pm \epsilon)f(n) }
    & \le \P{ | \BarRM_j - \E{\BarRM_j} | \ge (\epsilon/2)f(n) } \\
    & = \frac{O(\ln^{3/2} n)}{(\epsilon/2)^2 f(n)^2 }
    = O\big( 1/(\epsilon^2 \sqrt{\ln n} ) \big)
  \end{align*}
\end{proof}

\section{Proofs for expected rank of women}
\label{appendixWomenRanks}

In this appendix, we prove our concentration results for the average
rank of women in each tier.
First we prove the following basic lemma:

\begin{proposition}
  \label{thrmPolyLogReProposals}
  With probability $1 - O(1/\log n)$, the total number of re-proposals
  in deferred acceptance with re-proposals 
  is $O(\log^4 n)$.
\end{proposition}
\begin{proof}
  By proposition~\ref{thrmCouponUpperTail}, we know there are $O(n\log n)$
  proposals in deferred acceptance with re-proposals
  with probability $1 - 1/n^3$.
  By corollary~\ref{thrmMaxIndividualProposals},
  with probability $1 - 1/n^3$, each man makes as most $O(\log^2 n)$
  proposals, so in this case the probability that a proposal is
  a repeat is $O(\log^2 n/ n)$.
  If either of the above do not hold, we assume the number of re-proposals
  can be as high as the number of proposals overall,
  which is $O(n\log n)$ by proposition~\ref{thrmCouponCentralConcentration}.
  Thus, the expected number of re-proposals is at most
  \[ \big(1 - 1/n^3\big) O\left( \frac{\log^2 n}{n}\right) 
    \cdot O(n\log n )
    + \big(1/n^3\big) \cdot O(n\log n) = O(\log^3 n)
  \]
  Thus Markov's inequality tells us that, there are more than 
  $\Omega(\log^4 n)$ re-proposals with probability
  only $1/\log n$.
\end{proof}

Next, we fix some woman and consider the number of proposals
she gets from each tier of men.
We also consider the total number of proposals received by
the tier $i$ which $w$ is in (though we only need a one sided bound for this).

\begin{proposition} \label{thrmTotalPropsTierToTier}
  Let $S_{j\to i}$ denote the total number of proposals made by men in tier
  $j$ to women in tier $i$.
  Furthermore, fix a woman $w$ in tier $i$,
  and let $S_{j\to w}$ denote the
  number of proposals made by men in tier $j$ to $w$.
  For any $\epsilon>0$,
  with probability $1 - O(1/(\epsilon^2 \sqrt{\ln n}))$, we have
  \begin{align*}
    S_{j\to i} 
    & \le (1 + \epsilon)
    \frac{\epsilon_i\alpha_i}{\alpha_{\min}}\cdot 
    \frac{\delta_j \beta_j^{-1}}{\bm{\delta}\cdot\bm{\beta}^{-1}}
    n\ln n
    \\
    S_{j\to w}
    & = (1 \pm \epsilon)
    \frac{\alpha_i}{\alpha_{\min}}\cdot 
    \frac{\delta_j \beta_j^{-1}}{\bm{\delta}\cdot\bm{\beta}^{-1}}
    \ln n
  \end{align*}
\end{proposition}
\begin{proof}
  Let $S_j$ denote the total number of proposals by men in tier $j$.
  By theorem~\ref{thrmMenRanksCentralConcentration}, 
  \begin{align*}
    S_j \le (1 + \epsilon) C_j n\ln n
    &&
    \text{for }\ C_j =
    \frac{\epsilonDotAlpha}{\alpha_{\min}}\cdot 
    \frac{\delta_j \beta_j^{-1}}{\bm{\delta}\cdot\bm{\beta}^{-1}}
  \end{align*}
  with probability $1 - O(1/(\epsilon^2\sqrt{\ln n}))$.
  Now, the marginal probability that an individual proposal goes to some woman in
  tier $i$ is $\epsilon_i \alpha_i / (\epsilonDotAlpha)$.
  However, we still need to handle the fact that men do not make 
  repeat proposals.

  Recall that by~\ref{thrmMaxIndividualProposals},
  with probability $1 - 1/n^2$ no man makes more than
  $O(\ln^2 n)$ proposals.
  So in this case, the probability that a proposal goes to tier $j$ of
  women is always upper bounded by
  $\epsilon_i\alpha_i/(\epsilonDotAlpha) + O(\ln^2 n/n)$.
  Thus, when this holds, the number of proposals going to tier $j$
  is statistically dominated by the sum
  of $K = (1 + \epsilon)C_j n\ln n$ independent Bernoulli trials
  with success parameter
  $p = \epsilon_i\alpha_i/(\epsilonDotAlpha) + O(\ln^2 n/n)$.
  Note that 
  \begin{align*}
    Kp \le (1+2\epsilon)C_{j\to i} n\ln n
    && \text{for }\ C_{j\to i} =
    \frac{\epsilon_i\alpha_i}{\alpha_{\min}}\cdot
    \frac{\delta_j \beta_j^{-1}}{\bm{\delta}\cdot\bm{\beta}^{-1}}
  \end{align*}
  for $n$ large enough.
  The concentration of this sum can be bounded with
  a standard application of Chernoff~\ref{thrmMultChernoff}
  to get
  \[
    \P{ S_{j\to i} \ge (1 + \epsilon)(Kp) } \le
    \exp( - \epsilon^2 (Kp) / 3) = n^{-\Omega(\epsilon^2 \ln n)}.
  \]
  For any $\epsilon > 0$, the above probability is
  $o(1/\sqrt{\ln n})$, so with probability
  $1 - O(1/(\epsilon^2\sqrt{\ln n}))$ overall, we have
  \[ S_{j\to i} \ge (1 + 3\epsilon) C_{j\to i} n\ln n.
  \]

  For the second concentration result,
  we consider deferred acceptance with re-proposals.
  First we prove the lower bound.
  Let $T_j$ denote the number of proposals by men in tier $j$,
  and let $T_{j\to w}$
  be those which additionally go to woman $w$.
  By the same proof as~\ref{thrmStatisticalDominance}, we have
  $S_{j} \preceq T_{j}$.
  Thus, by theorem~\ref{thrmMenRanksCentralConcentration},
  we have $T_j \ge (1 - \epsilon) C_j n\ln n$ with probability
  $1 - O(1/(\epsilon^2 \sqrt{\ln n}))$.

  Consider the first $(1 - \epsilon) C_j n\ln n$ proposals by men in tier
  $j$ in deferred acceptance with re-proposals.
  There are i.i.d. and go to $w$ with
  probability $\pi_i = \alpha_i / (n\epsilonDotAlpha)$.
  The probability that $w$ gets less than $(1-\epsilon)\mu$
  proposals, where $\mu := (1-\epsilon)\pi_i C_j n\ln n = \Theta(\ln n)$,
  can be bounded with a standard application of
  Chernoff~\ref{thrmMultChernoff}:
  \[ \P{ T_{j \to w} \le (1-\epsilon)\mu }
    \le \exp( -\epsilon^2 \mu /2)
    = n^{-\Omega(\epsilon^2)}.
  \]
  Now, with high probability, every man proposes to at most $O(\ln^2 n)$
  distinct women, by proposition~\ref{thrmMaxIndividualProposals}.
  When this holds, among the first $(1-\epsilon)\mu = O(\ln n)$ proposals,
  the probability that even a single one of those proposals is a repeat is
  $O(\ln^3 n / n)$.
  Thus, with probability $1 - \widetilde O(1/n)$,
  we have that $w$ did not receive
  a single repeated proposals among her first $(1 - \epsilon)\mu$.
  Thus, with probability $1 - n^{-\Omega(\epsilon^2)}$,
  $w$ received at least $(1 - \epsilon)\mu$ proposals in
  the first $(1 - \epsilon)C_j n\ln n$ proposals by men in tier $j$.
  By~\ref{thrmMenRanksCentralConcentration}, tier $j$ makes this many
  proposals with probability $1 - O(1/(\epsilon^2\sqrt{\ln n}))$, so
  overall
  \[ S_{j\to w} \ge (1 - 2\epsilon)\pi_i C_j n\ln n
    = (1 - 2\epsilon)
    \frac{\alpha_i}{\alpha_{\min}}\cdot 
    \frac{\delta_j \beta_j^{-1}}{\bm{\delta}\cdot\bm{\beta}^{-1}}
    \ln n
  \]
  with probability $1 - O(1/(\epsilon^2\sqrt{\ln n}))$.

  We now prove the upper bound of the second concentration result.
  By theorem~\ref{thrmMenRanksCentralConcentration} and
  proposition~\ref{thrmPolyLogReProposals},
  we have $T_j \le (1 + \epsilon) C_j n\ln n + O(\log^4 n)
  \le (1 + 2\epsilon) C_j n\ln n$
  with probability
  $1 - O(1/(\epsilon^2 \sqrt{\ln n})) - O(1/\log n)$.
  Consider the first $(1 + 2\epsilon)C_j n\ln n$ proposals according to $T$,
  which go to $w$ with probability $\pi_i$.
  By another application of Chernoff~\ref{thrmMultChernoff},
  \[ \P{ S_{j \to w} \ge (1+\epsilon)\mu }
   \le \P{ T_{j \to w} \ge (1+\epsilon)\mu }
    \le \exp( -\epsilon^2 \mu /3)
    = n^{-\Omega(\epsilon^2)},
  \]
  where $\mu = (1 + 2\epsilon)C_j n\ln n$ is the expected number of proposals
  to $w$ in this process.
  Thus,
  \[ S_{j\to w} \le (1 + 3\epsilon)\pi_i C_j n\ln n
    = (1 + 3\epsilon)
    \frac{\alpha_i}{\alpha_{\min}}\cdot 
    \frac{\delta_j \beta_j^{-1}}{\bm{\delta}\cdot\bm{\beta}^{-1}}
    \ln n
  \]
  with probability $1 - O(1/(\epsilon^2\sqrt{\ln n}))$ overall.




\end{proof}
\begin{remark}
    The proofs in appendix~\ref{sectionSmooth} reveal that, in some sense,
    the rank of a man $m$ behaves like a geometric distribution 
    (with parameter given by the probability of his proposal
    being accepted).
    The number of proposals a woman $w$ receive, on the other hand,
    intuitively behaves more like a \emph{binomial} distribution
    (with total number of trials given by the number of proposals overall,
    and success probability given by $\pi_i = \alpha_i / (n\epsilonDotAlpha)$
    for a woman in tier $i$).
    Indeed, a binomial distribution with $n\ln n$
    trials and success probability $1/n$ concentrates to
    $(1\pm \epsilon)\ln n$ with high probability,
    and correspondingly we found above that the number of proposals 
    received by the women actually
    concentrates much better than the number of proposals made by the men.
    
    Note, however, that the rank achieved by an individual woman 
    does not concentrate. At some intuitive level, this is because
    the rank of $w$ behaves like an exponential distribution
    with rate parameter given by the weight of proposals $\Gamma_w$ saw.
    
\end{remark}

Next, we provide a lemma that states most women receive few proposals. 

\begin{proposition}\label{thrmWomenAtMostLogProposals}
  With probability $1 - O(1/n)$, no woman receives more
  than $O(\log n)$ proposals.
\end{proposition}
\begin{proof}
  Fix a woman $w$, without loss of generality in the highest tier,
  and consider deferred acceptance with re-proposals.
  Note that the number of proposals $w$ gets in this process statistically
  dominates her number of proposals in DA.
  By proposition~\ref{thrmCouponUpperTail}, there exists a constant $K$
  such that this process terminates before $Kn\ln n$ total proposals
  with probability $1 - 1/n^2$.

  Now, consider the first $Kn\ln n$ proposals in deferred acceptance with
  re-proposals, and let $X$ be the number of these proposals which went to $w$.
  Since proposals are independent, there is a $\pima$ chance
  that each one goes to $w$. Thus, the number of proposals $w$ receives is
  a sum of $Kn\ln n$ independent Bernoulli trials success $\pima$.
  Thus, $\E{X} = \pima Kn\ln n = \Theta(\ln n)$,
  and a standard Chernoff bound 
  says that there exists a constant $\delta>0$ such that
  \[ \P{ X \ge (1+\delta) \E{X} }
    \le \exp(-\delta \E{X} / 3 ) = 1/n^2.
  \]
  Thus, the number of proposals $w$ receives is $O(\ln n)$ with probability
  $1 - O(1/n^2)$. Taking a union bound over the $n$ women, we have the result.
\end{proof}

Finally we can prove our main result:

\RestateThrmWomenRanks*
\begin{proof}
  We study the rank that a woman $w$ achieves by letting DA run until
  completion, then generating the rest of $w$ ranking based on the
  weight of proposals she received in DA.
  Specifically, for each woman $w$ let $\Gamma_w$ be the sum of
  weights of all men who proposed to $w$.
  We denote by $\E{\cdot | \Gamma}$ the expectation conditioned on the
  state at the end of deferred acceptance.
  For each man $m$ who did not propose to $w$ during DA,
  let $a_{m,w} = 1$ if $m$ is preferred to $w$'s match once we
  generate the rest of $w$'s preference list, and $0$ otherwise.
  We have $r_w = 1+\sum_{m} a_{m,w}$, where the sum runs over all
  men who did not propose to $w$ during DA.
  Note that, for each man $m$ with fitness $\beta(m)$,
  $\E{a_{m,w}} = \beta(m) / (\beta(m) + \Gamma_w)$,
  by~\ref{thrmWomenDeferedDecisions}.
  However, $a_{m,w}$ are not independent, for instance because
  if many men are ranked worse than $w$'s match, this is likely
  because $w$ ranks her match very highly
  (this is the detail that prevents the rank of $w$ from concentrating).
  
  For each man $m$, let $a_{m,w}'$ be an independent 
  random variable which is $1$ with probability
  $\beta(m) / (\beta(m) + \Gamma_w)$ and zero otherwise.
  Set $r_w' = \sum_m a_{m,w}'$, where the sum runs over all men.
  We use $r_w'$ to study the expectation of $r_w$.
  Let $P_w$ denote the number of proposals $w$ received during DA.
  Note that $w$ will never rank any man who proposed to her during
  DA above her eventual match.
  However, we have
  $\E{r_w} = \E{\sum_m a_{m,w}} = \E{\sum_m a_{m,w}'}$, where the sum
  runs over all men who have not proposed to $w$,
  by the linearity of expectation.
  Thus, $1 + \E{r_w'} \ge \E{r_w} \ge \E{r_w'} - P_w$.

  Our first task is to get a lower bound on the expectation of the
  average of $r_w'$ across tier $i$ women.
  By~\ref{thrmTotalPropsTierToTier} and a union bound,
  we know that with probability  $1 - O(1/(\epsilon^2 \ln n))$,
  for each tier $j$ of men, tier $i$ of women received
  $(1\pm\epsilon)C_{j\to i} n\ln n$ proposals from men in tier $j$,
  where $C_{j\to i} = (\epsilon_i\alpha_i/\alpha_{\min})
  (\delta_j\beta_j^{-1} / \deltaDotBeta^{-1})$.
  Let this event be denoted $E$.
  In this case,
  \[  (\epsilon_i n)^{-1}\sum_{w\in T_i} \Gamma_w
    = (\epsilon_i n)^{-1}(1\pm\epsilon) \sum_j 
    \frac{\epsilon_i\alpha_i}{\alpha_{\min}}\cdot 
    \frac{\delta_j}{\bm{\delta}\cdot\bm{\beta}^{-1}}
    n\ln n
    = (1\pm\epsilon)
    \frac{\alpha_i}{\alpha_{\min}}\cdot 
    \frac{\ln n}{\bm{\delta}\cdot\bm{\beta}^{-1}}
  \]

  Let $\overline R_i' = (\epsilon_i n)^{-1} \sum_{w\in T_i} r_w'$.
  By Jensen's inequality (applied to $\E{r'_w}
  = \sum_m \beta(m) / (\beta(m)+\Gamma_w)$, which is a convex
  function of $\Gamma_w$),
  we know that when $E$ holds and for $n$ large enough,
  \begingroup
  \allowdisplaybreaks
  \begin{align*}
    \E{ \overline R_i' \Big| \Gamma }
    & = (\epsilon_i n)^{-1} \sum_{w\in T_i} \E{r_w'} \\
    & \ge \sum_m \frac{\beta(m)}{\beta(m) +
      (\epsilon_i n)^{-1} \sum_{w\in T_i}\E{\Gamma_w}} \\
    & \ge \sum_m \frac{\beta(m)}{(1 + 2\epsilon)C_\Gamma \ln n}
      && C_\Gamma = \frac{\alpha_i}{\alpha_{\min}(\deltaDotBeta^{-1})} \\
    & \ge \frac{n\deltaDotBeta}{(1 + 2\epsilon)C_\Gamma \ln n} \\
    & \ge (1 - 3\epsilon) C^{(i)} \frac{n}{\ln n}
      && C^{(i)} = (\deltaDotBeta)(\deltaDotBeta^{-1})
        \frac{\alpha_{\min}}{\alpha_i}
  \end{align*}
  \endgroup

  To complete the picture, we also need an upper bound on the expectation
  of $r_w'$ for a single woman.
  Consider a woman $w$ in tier $i$, and let 
  $C_{j\to w} = (\alpha_i/\alpha_{\min})
  (\delta_j\beta_j^{-1} / \deltaDotBeta^{-1})$.
  By proposition~\ref{thrmTotalPropsTierToTier} and a union bound,
  the probability $w$ receives less than $(1 - \epsilon)C_{j\to w} \ln n$
  proposals from men in tier $j$, for each $j$ simultaneously,
  is $O(1/(\epsilon^2\sqrt{\ln n}))$.
  By Markov's inequality, we know that with probability 
  $1 - O(\epsilon^2 \ln^{1/4} n)$, at least 
  $\epsilon_i n / \log^{1/4} n$ women in tier $i$ have
  at least $(1 - \epsilon)C_{j\to w} \ln n$
  proposals from men in tier $j$, for each $j$ simultaneously.
  These women have 
  \begin{align*}
    \Gamma_w 
    & \ge (1 - \epsilon)\sum_j \beta_j C_{j\to w} \ln n
      = (1 - \epsilon)C_\Gamma \ln n \\
    \E{r_w' | \Gamma} 
    & \le \sum_m \frac{\beta(m)}{\Gamma_w}
    \le (1 + 2\epsilon)\frac{n\deltaDotBeta}{C_\Gamma \ln n}
    = (1 + 2\epsilon)C^{(i)}\frac{n}{\ln n}.
   \end{align*}
  
  Note the following consequence of~\ref{thrmSmoothWHP}:
  a woman $w$ receives $\Omega(\ln n)$ proposals with
  probability $1 - n^{-\Omega(1)}$.
  Thus, by Markov's inequality, with probability $1 - n^{-\Omega(1)}$,
  at most $n^{1 - \Omega(1)}$ women receive fewer
  than $O(\ln n)$ proposals.
  With the remaining women, it's possible that many women
  get rank up to $n$. All told, we have that with probability
  $1 -  O(1/(\epsilon^2 \ln^{1/4} n))$,
  \begin{align*}
    \E{\overline{R}_i' | \Gamma}
    & \le \big(1 - O(1/\ln^{1/4}n)\big) 
      (1 + 2\epsilon)C^{(i)} \frac{n}{\ln n} \\
    & \qquad + O(1/\ln^{1/4} n)\cdot O(n/\ln n)
      + n^{-\Omega(1)}\cdot n \\
    & \le (1 + 3\epsilon)C^{(i)} \frac{n}{\ln n} \\
    \implies \E{\overline{R}_i' | \Gamma}
    & = (1\pm 3\epsilon)C^{(i)} \frac{n}{\ln n} 
  \end{align*}

  Now, conditioned on the realized weights $\{\Gamma_w\}$,
  the $\epsilon_i n$ variables $r_w'$ is independent, so Hoeffding's
  inequality~\ref{thrmBoundedHoeffding} immediately gives us
  \[
    \P{ \big|\overline R_i' - \E{\overline R_i'} \ge n^{3/4}\ \Big|\ \Gamma}
    \le 2\exp(- 2(\epsilon_i n)^2 (n^{3/4})^2 / (\epsilon_i n^3) )
    = 2\exp(- \Theta(\sqrt{n}) )
  \]
  
  By~\ref{thrmWomenAtMostLogProposals}, with probability $1 - 1/n$,
  no woman receives more than $O(\ln n)$ proposals,
  so $P_w \le O(\ln n)$ for all $w$. So in this case the difference between
  $\BarRW_i$ and $\overline R_i'$ is at most $O(\ln n)$,
  and we finally have that with probability
  $1 - O(1/(\epsilon^2 \ln^{1/4} n)) - 1/n - 2\exp(- \Theta(\sqrt{n}) )
  = 1-O(1/(\epsilon^2 \ln^{1/4} n))$,
  \[
    \BarRW_i = (1 \pm 4\epsilon) C^{(i)} \frac{n}{\ln n}.
  \]

\end{proof}

\section{Proofs for distribution of match types}
\label{appendixMatchType}


Given the results of the previous appendix,
our main theorem on the distribution of match types
is a fairly easy corollary.

\RestateThrmMatchTypes*

\begin{proof}
  Let $\Gamma_{j\to w} := \beta_j S_{j\to w}$ denote the sum
  of public scores of all men in tier $j$ who propose to $w$.
  By proposition~\ref{thrmTotalPropsTierToTier} and a union
  bound over the constant number of tiers,
  we have that, with probability $1 - o(1)$,
  \[ \Gamma_{j\to w} = (1\pm\epsilon) \delta_j U_i
    \qquad\qquad U_i :=
    \frac{\alpha_i}{\alpha_{\min}}\cdot 
    \frac{\ln n}{\bm{\delta}\cdot\bm{\beta}^{-1}}
  \]
  for each tier $j$ simultaneously.
  Recall that the probability of a given proposal being
  the favorite out of all those seen by $w$ is independent
  of the order in which men propose to $w$.
  Thus, when the above holds,
  the probability that $w$ is matched to a man in tier $j$
  is $(1\pm\epsilon)\delta_j U_i / ( (1\pm\epsilon) U_i) 
  = (1\pm 3\epsilon) \delta_j$. In all other cases,
  the probability is between 0 and 1.
  Thus, for $n$ large enough, the probability
  is $(1 - o(1))(1\pm 3\epsilon)\delta_j + o(1) =
  (1\pm 4\epsilon)\delta_j$ overall.
\end{proof}

\begin{remark}
  Unlike our results on the average rank of different tiers,
  the above result is only proven ``in expectation''
  instead of proving concentration.
  That is, we prove a result on the overall probability
  of certain types matching, instead
  of results on what happens for the realized distribution of match types
  (i.e. the fraction of matches which are made between
  tier $i$ and tier $j$ for each $i,j$)
  with high probability.
  
  We believe this is an artifact of our current
  proof technique, and conjecture that for any $\epsilon >0$,
  with probability approaching $1$
  we have that there are $(1\pm\epsilon)\epsilon_i\delta_j n$
  pairs formed from a woman in tier $i$ and a man in tier $j$.
  To prove this it would suffice to show that, similar
  to the situation for the men, the match for different women
  is only very weekly correlated.
\end{remark}


\section{Lemmas in probability theory}
\label{appendixProbability}

The proof of the next claim is given in \cite{ross2006probability},
Example 5.17. For completeness, we reproduce it here.

\RestateThrmCouponPoissonization*

\begin{proof}
  Consider $n$ Poisson clocks, each ticking with rate $p_i$
  for $i=1,\ldots,n$. Note that this is equivalent to having a ``master''
  Poisson clock, ticking with rate $1$, and assigning every tick
  to one of the $n$ coupons according to the distribution $(p_i)_{i\in[n]}$.
  By definition, the amount of time between ticks of a Poisson clock
  with rate $p_i$ is distributed exactly according to $\Exp(p_i)$.
  Thus, the random variable $X$ is distributed exactly as
  the of time until all of the $n$ clocks have ticked at least once.
  We can see that the discrete time coupon collector 
  $T = T_{\mathcal{D}}$ can be recovered as 
  one particular random variable in this
  continuous-time process.
  Specifically, $T$ is the number of times the ``master clock''
  ticked before all of the $n$ clocks ticked at least once.

  We have $X = \sum_{i=1}^T S_i$, where $S_i$ is the ``$i$th inter-arrival time''
  of the master clock. Note that $S_i$ is independent of $T$.
  Because the master clock ticks at rate $1$, each $S_i$ is distributed
  according to $\Exp(1)$ and has expectation $1$, so we have
  \[ \E{X} = \E{\E{X | T}} = \E{\E{\sum_{i=1}^T S_i \bigg| T}}
   = \E{\sum_{i=1}^T \E{S_i | T}} = \E{T}
  \]
\end{proof}

We believe the following result is folklore:

\begin{proposition}[Maximum of exponential distributions]
  \label{thrmMaximumExponential}
  The maximum of $k$ independent draws from $\Exp(\lambda)$
  is distributed identically to
  \[ \sum_{i=1}^k Y_i, \qquad Y_i \sim \Exp(i \lambda) \]
  (i.e. $Y_i$ are independent draws from $\Exp(i \lambda)$
  for $i=1,\ldots,n$).
  In particular, the maximum has expected value
  $H_k / \lambda$ (for $H_k$ the $k$th harmonic number)
  and variance $\Theta(1/\lambda)$.
\end{proposition}
\begin{proof}
  We actually prove a more general claim about $X_{(j)}$, the $j$th order
  statistic of $k$ independent draws from $\Exp(\lambda)$.
  If $Y_i \sim \Exp(i \lambda)$ independently for $i=1,\ldots,n$,
  I claim that $X_{(j)}$ is (jointly) distributed identically to
  \[ \sum_{i=n - j + 1}^n Y_i. \] 

  To prove this claim, induct on $j$. For $j=1$, this is just proving that
  the minimum of $n$ draws from $\Exp(\lambda)$ is distributed like
  $\Exp(n\lambda)$, which is a classic exercise in probability theory.

  For $j>1$, condition on $X_{(j-1)} = x$.
  Now, $X_{(j)}$ is distributed like the minimum of $n - j + 1$ independent
  draws from $U \sim \Exp(\lambda)$, conditioned on each of those draws being $\ge x$.
  By the ``memoryless'' property of the exponential distribution,
  the distribution of $U - x$ conditioned on $U \ge x$ is identical
  to the distribution of $U$.
  Thus, $X_{(j)} - X_{(j-1)}$ is distributed exactly as the minimum of $n-j+1$
  draws from $\Exp(\lambda)$, or equivalently one draw from
  $\Exp((n-j+1)\lambda)$. The claim then follows by induction.

  Finally, the calculation of expectation and variance follows from the
  expectation and variance of $\Exp(i \lambda)$ (and the
  bound on variance follows because $\sum_{i \ge 1} \nicefrac{1}{i^2} =
  \Theta(1)$).
\end{proof}

We also need the following standard concentration inequalities:
\begin{proposition}[Multiplicative Chernoff Bound]
\label{thrmMultChernoff}
  Let $X_1,\ldots,X_n$ be independent random variables taking values in
  $\{0,1\}$. Let $X = \sum_i X_i$ and let $\mu = \E{X}$.
  Then for any $0\le\delta\le 1$, we have
  \begin{align*}
    \P{ X \le (1 - \delta)\mu } & \le \exp\left(-\delta^2\mu/2\right) \\
    \P{ X \ge (1 + \delta)\mu } & \le \exp\left(-\delta^2\mu/3\right)
  \end{align*}
  and for any $\delta\ge 1$, we have
  \[
    \P{ X \ge (1 + \delta)\mu } \le \exp\left(-\delta\mu/3\right)
  \]
\end{proposition}

\begin{proposition}[Hoeffding's Inequality]
  \label{thrmBoundedHoeffding}
  Let $X_1,\ldots,X_n$ be independent random variables taking values in
  $[a_i,b_i]$. Let $X = (1/n)\sum_i X_i$ and let $\mu = \E{X}$.
  Then for any $t\ge 0$, we have
  \[
    \P{ |X - \mu | \ge t } \le 2\exp\left(
      -\frac{2n^2 t^2}{\sum_{i=1}^n (b_i-a_i)^2}\right)
  \]
\end{proposition}

\section{More Computational Experiments}

In this section, we present three additional computational experiments in supplement to
our theoretical estimates and the experiments in~\ref{sectionSimulations}.
The first experiment compares simulation results with our predicted limits in 
theorem~\ref{thrmMenRanksCentralConcentration} and~\ref{thrmWomenRanks}. 
The second one attempts to further generalize our results by showing that in
slightly unbalanced markets the core is small and hence our predictions
should in fact apply to any stable matching in such markets.
And the last experiment takes a macro view on the distribution of matched pairs across
tiers on both sides (i.e. the fraction of tier $i$ women matched to tier
$j$ men for each pair of $i,j$).

\subsection{Numerical results and accuracy of the estimates} 
In the first experiment, we examine the accuracy of our asymptotic estimates in theorem~\ref{thrmMenRanksCentralConcentration} and~\ref{thrmWomenRanks} numerically. We consider a sequence of balanced markets characterized by the same configuration of tiers on both sides but with growing total numbers of agents. Specifically, each market consists of two tiers of men on the proposing side and three tiers of women receiving proposals. The ratio of tier sizes is fixed at $\bm{\delta}=(1/4, 3/4)$ for men and $\bm{\epsilon}=(1/16, 5/16, 5/8)$ for women; the public scores for each tier have fixed ratio of $\bm{\beta}=(3, 1)$ for men and $\bm{\alpha}=(3,2,1)$ for women. The size of the market $n$, i.e. the number agents on each side, ranges from $2^4$ to $2^{19}$ at each integer power of 2. For each market described above, we simulate 1,000 realizations of the man-proposing DA, and with each realization, we compute the average rank of partners across agents in each tier on each side. Figure~\ref{fig:avg_rank} reports the average across realizations of the per tier average rank of partners as the market size grows.
Figure~\ref{fig:rank_ratio} shows the convergence of the ratios of ranks of partners among tiers on each side.




\begin{figure}[hp]
    \centering
    \begin{subfigure}[t]{0.48\textwidth}
        \includegraphics[scale=0.42]{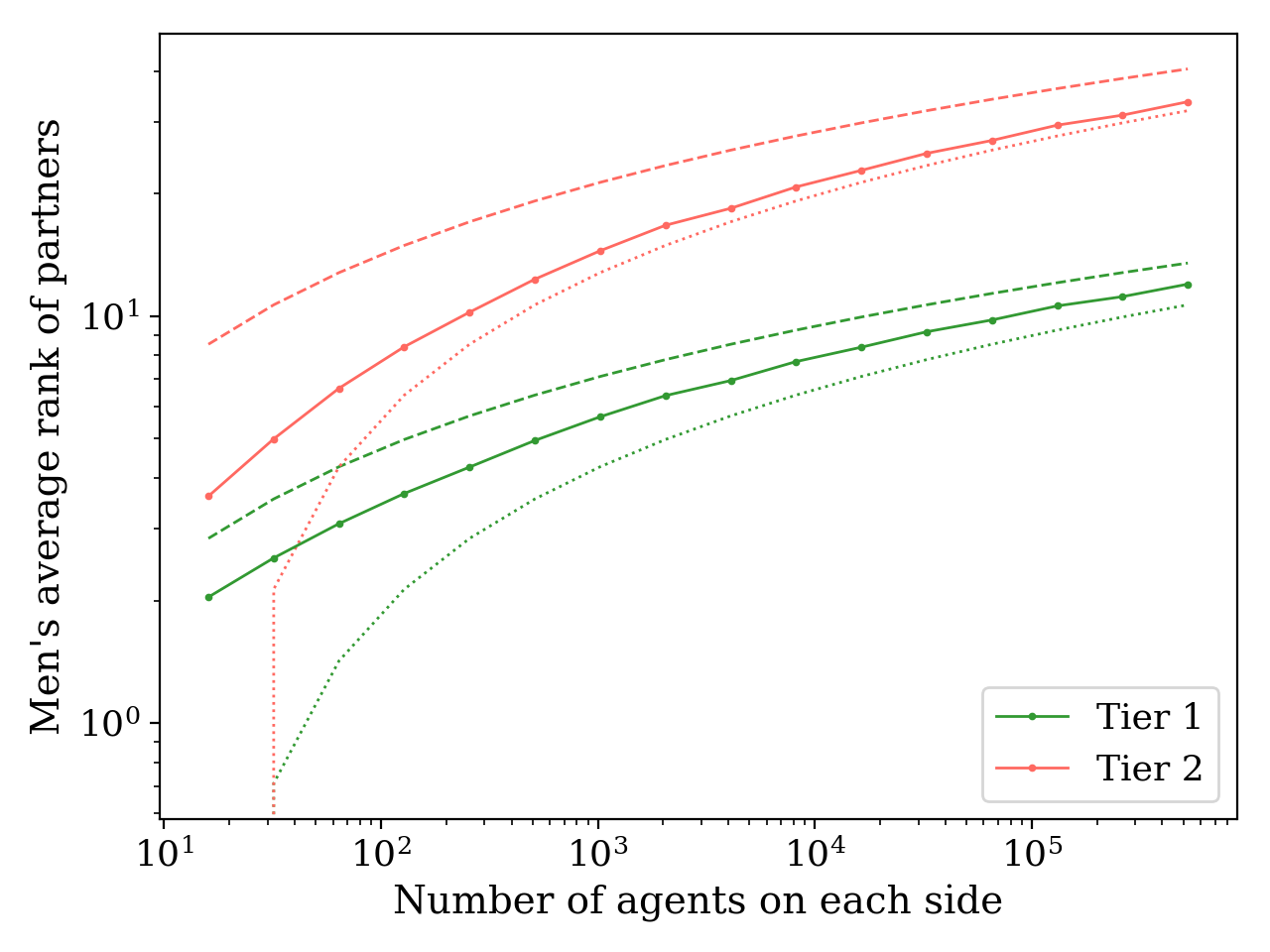}
        \caption{Average rank for men}
    \end{subfigure}
    \quad
    \begin{subfigure}[t]{0.48\textwidth}
        \includegraphics[scale=0.42]{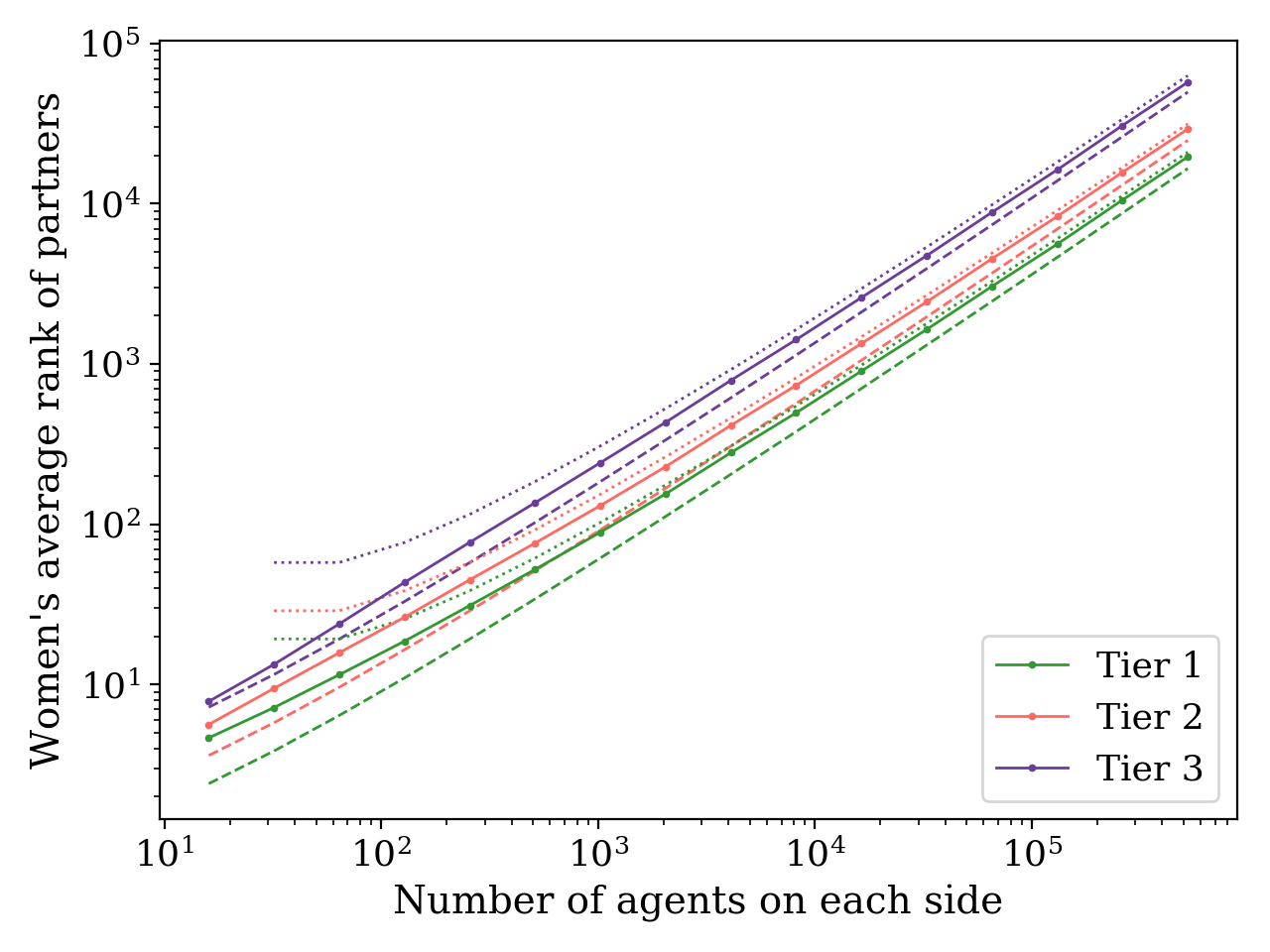}
        \caption{Average rank for women}
    \end{subfigure}
    
    \caption{
    The solid line gives
    the average ranks for men and women of each tier in a sequence of matching markets with fixed parameters $\bm{\delta}=(1/4, 3/4)$, $\bm{\epsilon}=(1/16, 5/16, 5/8)$, $\bm{\beta}=(3, 1)$, and $\bm{\alpha}=(3,2,1)$. 
    As remarks~\ref{remarkLogEspMinErrorBody} and~\ref{remarkLogEspMinErrorAppendix} discuss, our estimates
    converge quite slowly in this market, because $\epsilon_{\min}=1/16$ is small.
    The dashed lines indicate the estimates using our upper bound on the total number of proposals (based on proposition~\ref{thrmSimpleCouponUpperBound}), and the dotted lines indicate the estimates using the lower bound (based on
    proposition~\ref{thrmCouponExpectationLowerBound}, including the constants in the error term of order $O(n)$).
    }\label{fig:avg_rank}

    \begin{subfigure}[t]{0.48\textwidth}
        \includegraphics[scale=0.42]{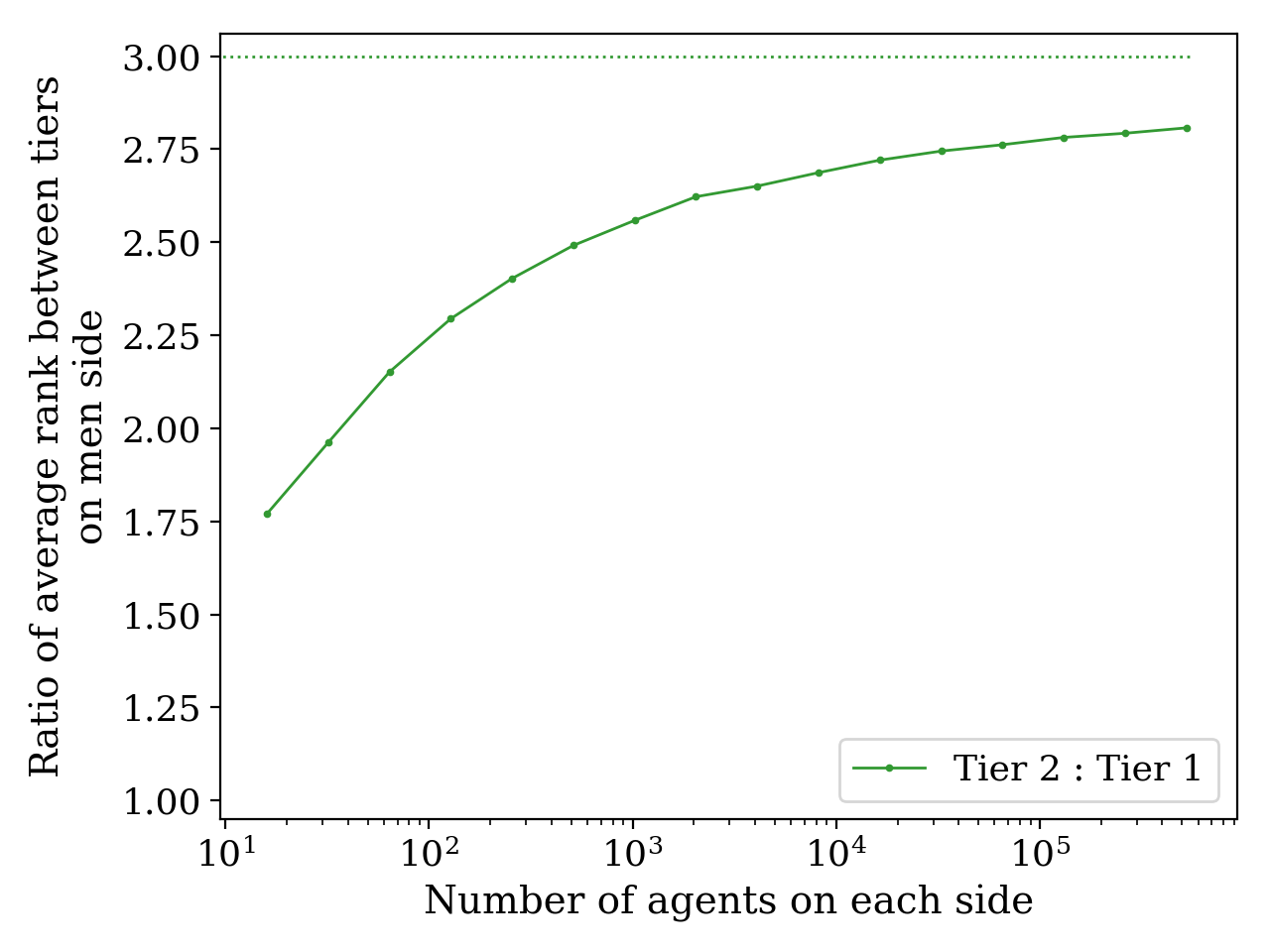}
        \caption{Ratio between average ranks of wives for men}
    \end{subfigure}
    \quad
    \begin{subfigure}[t]{0.48\textwidth}
        \includegraphics[scale=0.42]{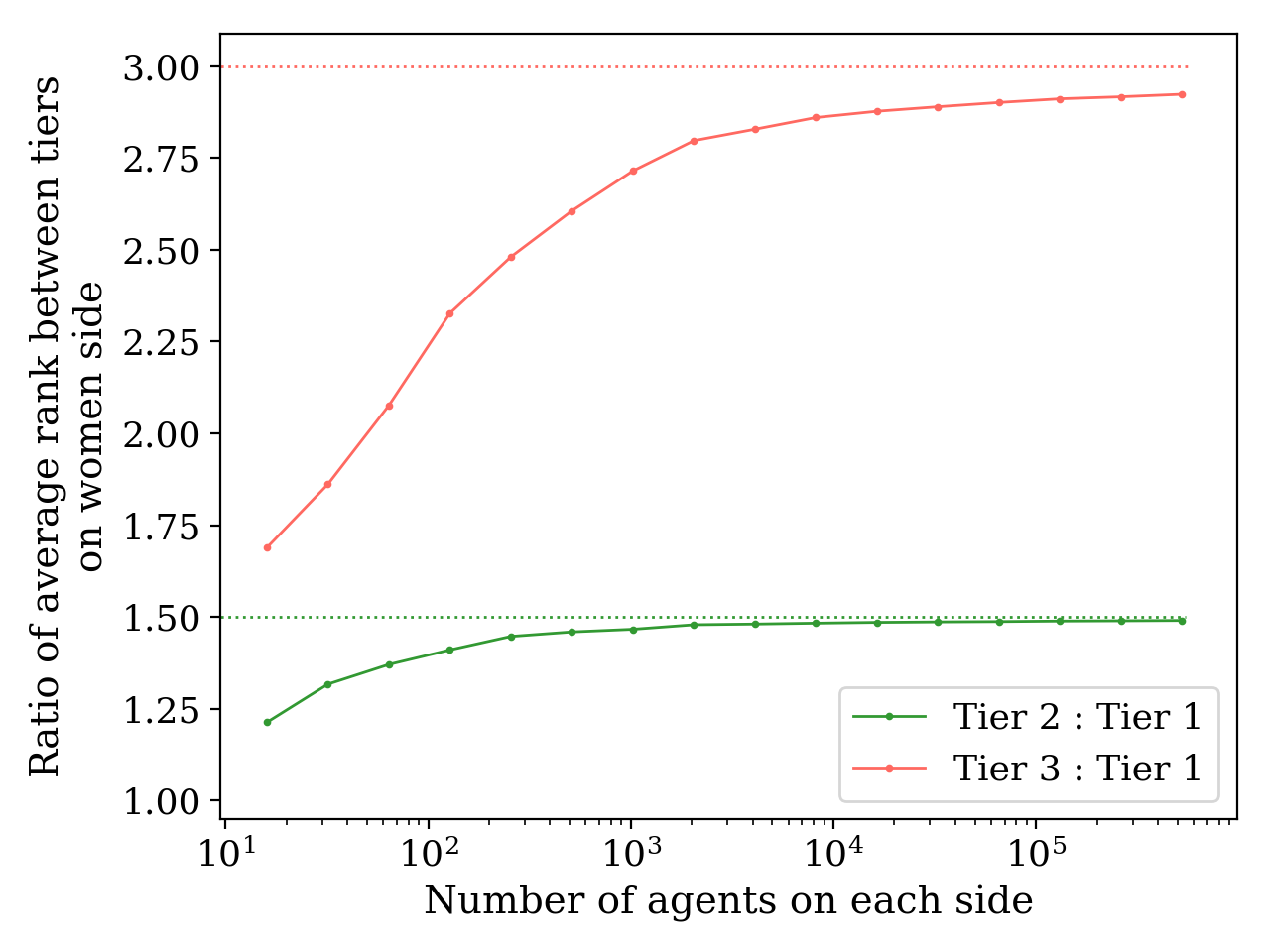}
        \caption{Ratio between average ranks of husbands for women}
    \end{subfigure}
    
    \caption{
    Ratios between average rank of partners across different tiers on each side. Theorem~\ref{thrmMenRanksCentralConcentration} and~\ref{thrmWomenRanks} imply that, for men and women respectively, the average rank within each tier is approximately proportional to the inverse of the public score of that tier in a sufficiently large market, and hence the the average rank ratio between two tiers should be close to the inverse of their public score ratio. In the simulation, the average rank ratio between the worse tier (tier 2 with $\alpha_2=1$) of men and the better (tier 1 with $\alpha_1=3$) converges to 3:1, and the rank ratios between the tier 2 with $\alpha_2=2$, tier 3 $\alpha_3=1$ and tier 1 $\alpha_1=3$ of women converge to 3:2 and 3:1, respectively. Notice again the very slow rate of convergence with the $x$-axis plotted in log scale, as is natural due to the fact that the averages
    converge at rate $O(1/\ln n)$.
    }\label{fig:rank_ratio}
\end{figure}

\subsection{Size of core for unbalanced markets}
In this experiment, we turn our attention to the more generalized setting of unbalanced markets. We provide some evidence that our prediction remains valid for slightly unbalanced markets, which are common in real life, and also that the core of such markets is small. This is a similar to the results in \cite{ashlagi2017unbalanced}, and would potentially imply that in an unbalanced setting our estimates apply to not just the man-optimal stable matching, but indeed to any stable matching, because the different between any stable matching and the man-optimal outcome is small.

We consider a one-side-tiered market, with 1,000 men in two tiers with fractional sizes $\bm{\delta}=(0.3, 0.7)$ and public scores $\bm{\beta}=(3,1)$ and a number of women in one tier ranging from 990 to 1,010. In each set-up, we compute the average rank of agents in each tier under the man-optimal outcome. Figure~\ref{fig:unbalanced_average_rank} shows the average ranks per tier across 1,000 realizations. We also computed the fraction of men in each tier with unique stable partners (i.e. those whose partners under man-optimal and woman-optimal are the same). The average percentage is shown in figure~\ref{fig:unbalanced_unique_partner_percentage}.

\begin{figure}[h]
    \centering
    \begin{subfigure}[t]{0.48\textwidth}
        \includegraphics[scale=0.42]{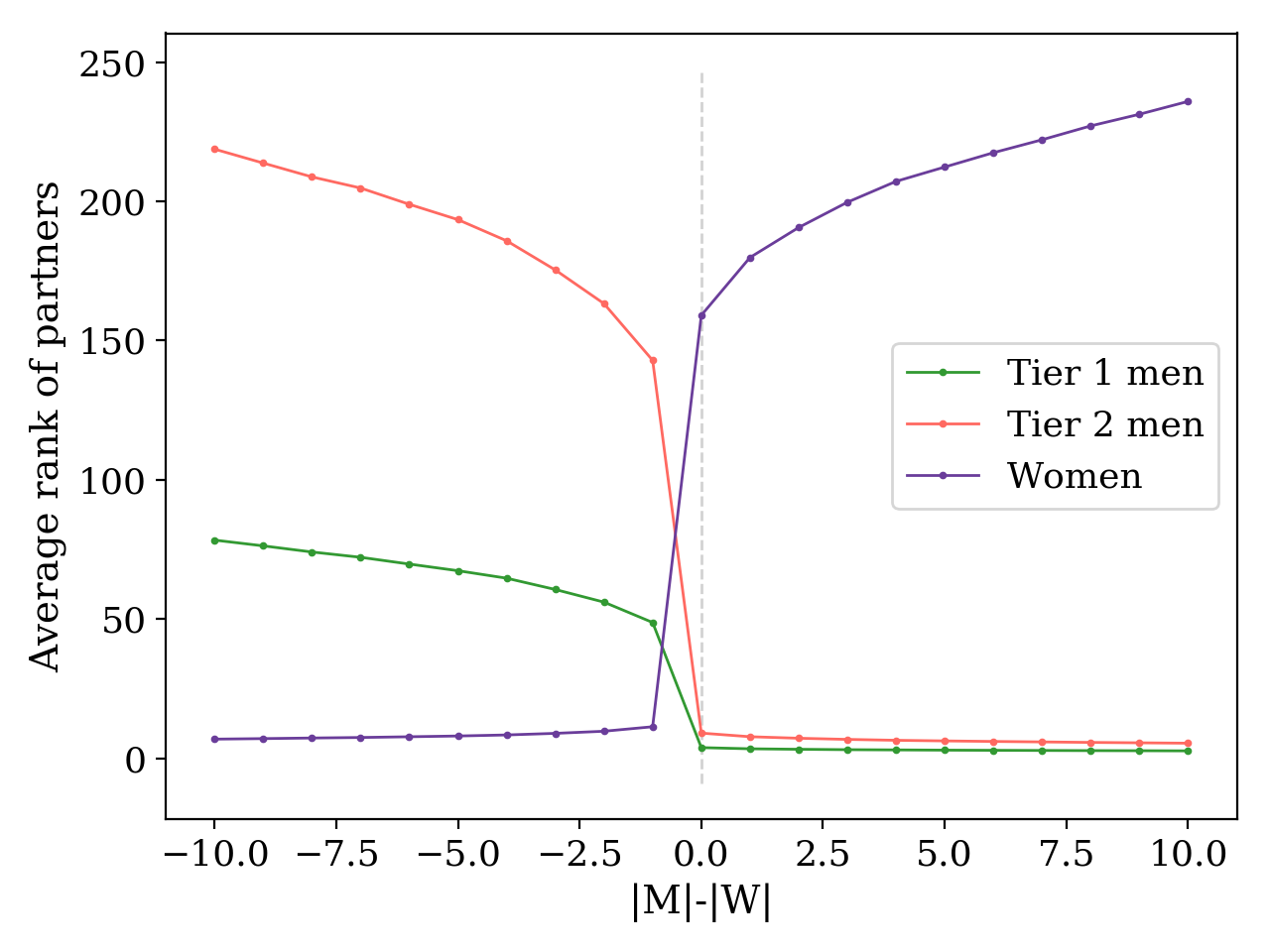}
        \caption{Average rank of partners for agents of each tier. Under the man-optimal outcome, the advantage flips to women when the men's side is longer by even slightest amount.}
        \label{fig:unbalanced_average_rank}
    \end{subfigure}
    \quad
    \begin{subfigure}[t]{0.48\textwidth}
        \includegraphics[scale=0.42]{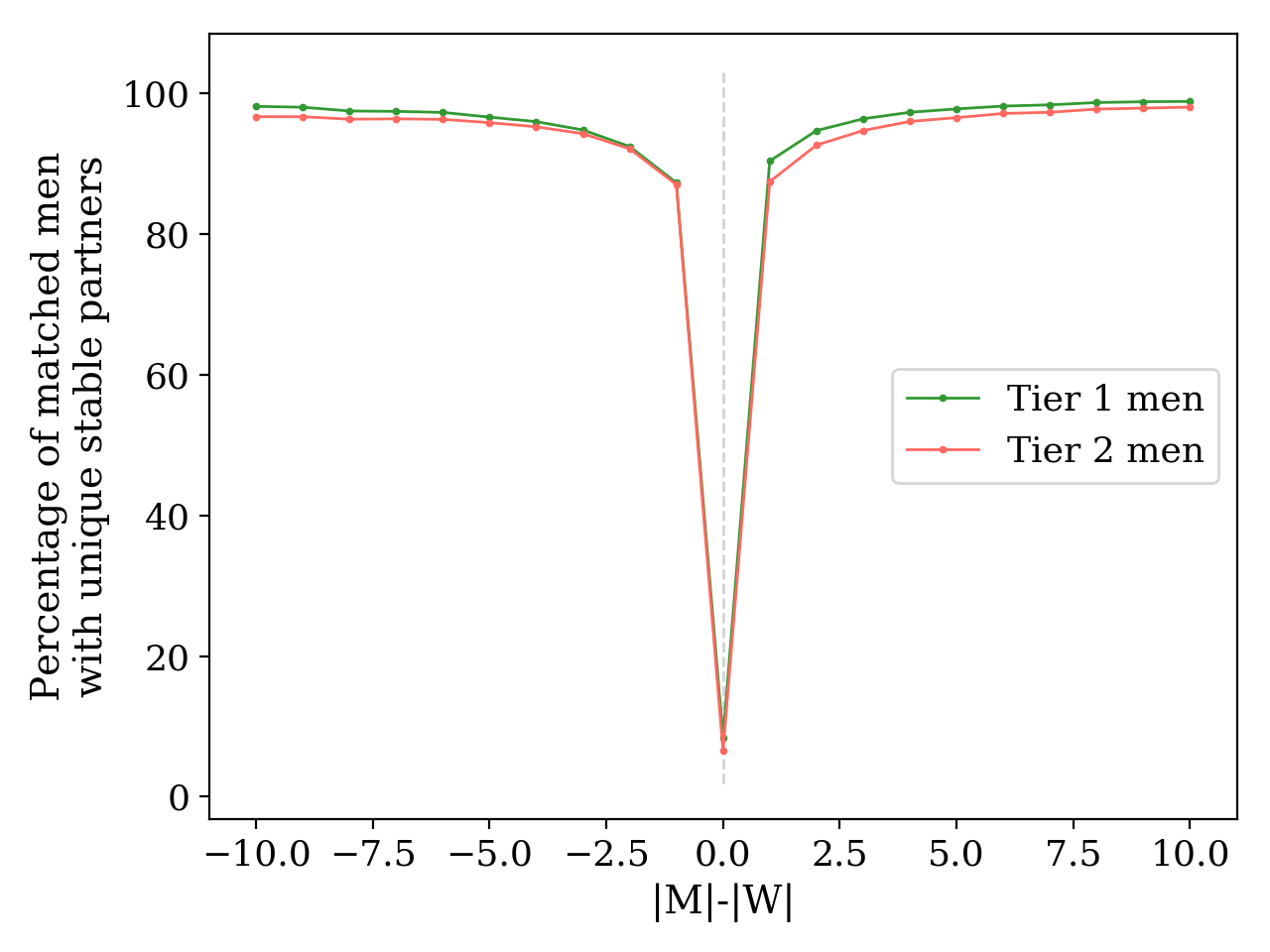}
        \caption{Fraction of men with unique stable partners. The steep dip at zero reflects a low number of agents with unique partners across all possible stable matchings when the market is balanced, hence indicating a large core. When the market becomes slightly unbalanced, however, we see that vast majority of agents have unique stable partners.}
        \label{fig:unbalanced_unique_partner_percentage}
    \end{subfigure}
    
    \caption{Matching markets with a fixed number of 1,000 men and a varying number of 990 to 1,010 women. The men side has parameters $\bm{\delta}=(0.3, 0.7)$, $\bm{\beta}=(3,1)$ and the women side is homogeneous.}\label{fig:unbalanced_market}
\end{figure}

    

\end{document}